\documentclass[a4paper,onecolumn,unpublished]{quantumarticle}
\pdfoutput=1

\usepackage[normalem]{ulem}

\usepackage[utf8]{inputenc}
\usepackage[english]{babel}
\usepackage[T1]{fontenc}
\usepackage{afterpage}
\usepackage{pifont}
\newcommand{\cmark}{\ding{51}}%
\newcommand{\xmark}{\ding{55}}%
\usepackage[backend=biber,natbib=true,sorting=none,sortcites=true,giveninits=true]{biblatex}
\usepackage[usenames,dvipsnames]{xcolor}
\usepackage{graphicx}
\usepackage[ colorlinks = true,
             linkcolor = MidnightBlue,
             urlcolor  = MidnightBlue,
             citecolor = orange,
             anchorcolor = ForestGreen,
]{hyperref}
\usepackage{float}
\graphicspath{
	{./pics/}
}

\makeatletter
\def\mathcolor#1#{\@mathcolor{#1}}
\def\@mathcolor#1#2#3{%
  \protect\leavevmode
  \begingroup
    \color#1{#2}#3%
  \endgroup
}
\makeatother

\usepackage[sans]{dsfont} % contains font for identity character
\usepackage{bbm}
\usepackage[mathscr]{euscript}
\usepackage{cancel} %to strike out things
\usepackage{amsfonts}

\usepackage{tcolorbox}
\tcbuselibrary{theorems}

\usepackage{amsmath}
\usepackage{amssymb}
\usepackage{units}
\usepackage{framed}
\usepackage{mdframed}
\usepackage{thmtools, thm-restate} %pretty  and restatable theorems
\usepackage{subcaption}
\usepackage{fullpage}
\usepackage{csquotes}
\usepackage{epigraph}
\usepackage{stmaryrd} % lightning
\usepackage{multicol}
\usepackage{enumerate}
\usepackage{tikz}
\usetikzlibrary{arrows}
\usepackage{rotating, xcolor}
\usetikzlibrary{shapes}
\usetikzlibrary{hobby}
\usetikzlibrary{decorations.markings}
\usetikzlibrary{calc,intersections,through,backgrounds}
\usetikzlibrary{decorations.pathreplacing,angles,quotes}
\usetikzlibrary{arrows,decorations.markings}
\usetikzlibrary{arrows.meta}
\usetikzlibrary{decorations.markings}
\tikzset{
  strike through/.style={
    postaction=decorate,
    decoration={
      markings,
      mark=at position 0.5 with {
        \draw[-] (-3pt,-3pt) -- (3pt, 3pt);
      }
    }
  }
}
\tikzset{
  defcolon/.style={
    postaction=decorate,
    decoration={
      markings,
      mark=at position -0.05 with {$\colon$}
    }
  }
}
\usepackage{tikz-cd}
\tikzset{degil/.style={
            decoration={markings,
            mark= at position 0.5 with {
                  \node[transform shape] (tempnode) {$\setminus$};
                  }
              },
              postaction={decorate}
}
}

\tcbset{colback=orange!5,colframe=orange!75!yellow,
fonttitle= \bfseries, float=htb}

\usepackage[capitalise, noabbrev]{cleveref}
\newenvironment{proof}{\paragraph{Proof:}}{\hfill$\square$}

\usetikzlibrary{shapes}
\usetikzlibrary{decorations.markings}

\newtheorem{theorem}{Theorem}[section]
\newtheorem{corollary}[theorem]{Corollary}
\newtheorem{lemma}[theorem]{Lemma}
\newtheorem{proposition}[theorem]{Proposition}
\newtheorem{conjecture}[theorem]{Conjecture}
\newtheorem{remark}[theorem]{Remark}
\newtheorem{definition}[theorem]{Definition}

\newtheorem{example}{Example}
\newcommand\longrsquigarrow{
\begin{tikzpicture}
\draw [decorate, decoration={zigzag, segment length=+6pt, amplitude=+.95pt,post length=+2pt}, arrows={-Classical TikZ Rightarrow}]  (0,0.1) -- (0.6,0.1); \draw[draw=none] (0,0)--(0.6,0);
\end{tikzpicture}
}

\newcommand\longlsquigarrow{
\begin{tikzpicture}
\draw [decorate, decoration={zigzag, segment length=+6pt, amplitude=+.95pt,post length=+2pt}, arrows={-Classical TikZ Rightarrow},rotate around={180:(0.3,0.1)}]  (0,0.1) -- (0.6,0.1); \draw[draw=none] (0,0)--(0.6,0);
\end{tikzpicture}
}

\DeclareMathOperator{\parent}{par}

\DeclareMathOperator{\doo}{do}

\DeclareMathOperator{\affects}{\vDash}
\DeclareMathOperator{\naffects}{\not\vDash}
\DeclareMathOperator{\given}{\,|\,}

\let\AA\relax
\newcommand{\AA}{\ensuremath{\mathcal{A}}}
\newcommand{\BB}{\ensuremath{\mathcal{B}}}
\newcommand{\CC}{\ensuremath{\mathcal{C}}}

\newcommand{\TT}{\ensuremath{\mathcal{T}}}
\newcommand{\WW}{\ensuremath{\mathcal{W}}}
\newcommand{\XX}{\ensuremath{\mathcal{X}}}
\newcommand{\YY}{\ensuremath{\mathcal{Y}}}
\newcommand{\ZZ}{\ensuremath{\mathcal{Z}}}
\newcommand{\SC}{\ensuremath{\mathcal{S}}}
\newcommand{\RC}{\ensuremath{\mathcal{R}}}
\newcommand{\Fut}{\bar{\mathcal{F}}}

\DeclareMathOperator{\spann}{span}
\DeclareMathOperator{\late}{late}
\DeclareMathOperator{\dircause}{\longrsquigarrow}
\DeclareMathOperator{\diresuac}{\longlsquigarrow}

\tikzstyle{causalarrow} = [draw=black,thick,decorate,decoration={snake,amplitude=.7mm,pre length=1mm,post length=2mm}]
\tikzstyle{affectsarrow} = [draw=black,thick]
\tikzstyle{vvarrow} = [draw=black,thick,decorate,decoration={zigzag,amplitude=+1.3pt,segment length=+7pt,pre length=1.5pt,post length=+7pt}]
\usetikzlibrary{decorations.pathmorphing}

\usepackage{tikz-3dplot}
\tdplotsetmaincoords{80}{45}
\tdplotsetrotatedcoords{-90}{180}{-90}

\newcommand{\hin}{\enquote{$\Rightarrow$}:\ }
\newcommand{\rueck}{\enquote{$\Leftarrow$}:\ }
\newcommand{\abs}[1]{\left| #1 \right|}
\newcommand{\unord}{\ensuremath{\not\preceq \not\succeq}}
\newcommand{\highlight}[2]{\colorbox{#1}{\ensuremath{#2}}}

\newtheorem{notation}{Notation}

\newcommand{\cA}{\mathcal{A}}
\newcommand{\cB}{\mathcal{B}}
\newcommand{\cC}{\mathcal{C}}

\newcommand{\cE}{\mathcal{E}}

\newcommand{\cG}{\mathcal{G}}

\newcommand{\cT}{\mathcal{T}}

\newcommand{\cX}{\mathcal{X}}
\newcommand{\cY}{\mathcal{Y}}

\newcommand{\upmodels}{\mathrel{\perp\mkern-11mu\perp}}
\newcommand{\indep}{\upmodels}

\begin{document}

\title{Characterizing Signalling: Connections between Causal Inference and Space-time Geometry}
\author{Maarten Grothus}
\email{maarten.grothus@inria.fr}
\affiliation{Institute for Theoretical Physics, ETH Zürich, 8093 Zürich, Switzerland}
\affiliation{Université Grenoble Alpes, Inria, 38000 Grenoble, France}
\affiliation{Univ.\ Grenoble Alpes, CNRS, Grenoble INP, Institut N\'eel, 38000 Grenoble, France}
\thanks{corresponding author}
\orcid{0000-0002-1561-9709}

\author{V. Vilasini}
\email{vilasini@inria.fr}
\affiliation{Institute for Theoretical Physics, ETH Zürich, 8093 Zürich, Switzerland}
\affiliation{Université Grenoble Alpes, Inria, 38000 Grenoble, France}
\orcid{0000-0002-7035-4205}

% \date{15.08.2024}

\maketitle

\begin{abstract}

\noindent Causality is pivotal to our understanding of the world, presenting itself in different forms: information-theoretic and relativistic, the former linked to the flow of information, and the latter to the structure of space-time. Leveraging a framework introduced in PRA, 106, 032204 (2022), which formally connects these two notions in general physical theories, we study their interplay. The framework defines information-theoretic causality through a causal modelling approach, which enables the inference of causal connections through agents’ interventions and correlations. First, we improve the characterization of information-theoretic signalling as defined through so-called affects relations. Specifically, we provide conditions for identifying redundancies in different parts of such a relation, introducing techniques for causal inference in unfaithful causal models (where the observable data does not “faithfully” reflect the causal dependences). In particular, this demonstrates the possibility of causal inference using the absence of signalling between certain nodes. Second, we define an order-theoretic property called conicality, showing that it is satisfied for light cones in Minkowski space-times with $d > 1$ spatial dimensions
but violated for $d = 1$. Finally, we study the embedding of information-theoretic causal models in space-time without violating relativistic principles such as no superluminal signalling
(NSS). In general, we observe that constraints imposed by NSS in a space-time and those imposed by purely information-theoretic causal inference behave differently.
We then prove a correspondence between conical space-times and faithful causal models: in both cases, there emerges a parallel between these two types of constraints. This indicates a connection between informational and geometric notions of causality, and offers new insights and tools for studying the relations between the principles of NSS and no causal loops in different space-time geometries and theories of information processing.

\end{abstract}

\newpage
{
	\hypersetup{linkcolor=black}
	\tableofcontents
}
\newpage

\section{Introduction}
A cornerstone of scientific inquiry is the search for causal explanations of the world around us.
The field of causal modelling and causal inference provides rigorous mathematical frameworks for connecting observable data with cause-effect relationships.
Originating in classical statistics \cite{Spirtes1993, Pearl2009}, where all causal relata are classical random variables, its versatility has led to widespread applications across data-driven disciplines \cite{Pearl2009, Schoolmaster2020, Laubach2021, Kleinberg2011, Raita2021, Arti2020}.
More recently, these approaches have been extended within quantum information theory, enabling causal explanations for fundamental correlations and phenomena in the quantum domain \cite{PhysRevX.7.031021, arxiv.1906.10726}.
Causal models provide an operational definition of causality in terms of information flow, which is discernible through agents' interventions.
This notion does not a priori refer to spacetime, but allows to formally capture the process of empirical science, where an experimenter may perform an intervention (measurement or operation) on a physical system, observe correlations, and reason about the fundamental patterns of causal influences that could explain those observations.
However, in other crucial disciplines such as relativistic physics, causality is intrinsically linked to the geometry of space and time.
Moreover the two notions must exhibit a compatible interplay as space-time structure does constrain the flow of information in physical processes through principles of relativistic causality, such as the impossibility of signalling outside the future light cone. 
The connections and compatibility between these two notions of causality have been formally studied in recent works involving one of us \cite{VVC, VVC_Letter, VVR}, which motivates a number of questions. 

The notion of signalling, featuring in relativistic principles, can be defined purely in information-theoretic terms by considering interventions on observable variables in a causal model and resulting changes to probabilities.
While signalling enables operational inference of causation, it is generally possible for causation (or ``information flow'') to exist without signalling or correlations at the observable level.
This occurs in scenarios where causal mechanisms are fine-tuned to wash out observable signalling or correlation, i.e., causation does not imply signalling in \emph{fine-tuned causal models} (also called unfaithful causal models).
Fine-tuning is essential to describe practical tasks such as ensuring cryptographic security and is ensured by protocol design which carefully fixes certain causal parameters to hide information from eavesdroppers.
In fundamental physics, fine-tuning in a fundamental theory is often viewed as undesirable, particularly as it often requires certain physical mechanisms of the theory to be hidden from any possible observation.
This is used as motivation to leave behind necessarily unfaithful explanations for Bell correlations using classical causal models \cite{Wood2015}, to pursue faithful explanations by means of intrinsically quantum causal models \cite{PhysRevX.7.031021,arxiv.1906.10726,Henson2014}.
However, some proposed post-quantum theories \cite{PhysRevA.53.3781,VVCJ} and QFT scenarios, where events (or nodes of the causal model) are linked to extended spacetime regions \cite{Sorkin1993}, showcase unavoidable fine-tuning.
This highlights the importance, both for fundamental and practical purposes, of carefully studying and classifying fine-tuning, and understanding when it becomes undesirable and when it can be useful.
However, this is just scarcely explored as fine-tuning complicates causal inference even in classical scenarios.
As a consequence, the relationships between signalling and causation remains less understood when fine-tuned causal models are considered.

The above considerations have direct consequences for relativistic principles. When embedding causal models (including fine-tuned ones) into space-time, the absence of superluminal signalling (NSS) does not generally imply the absence of superluminal causation \cite{VVC, VVCPR}.
This raises a fundamental question: how does the relativistic NSS principle constrain information-theoretic causal structures? 
Specifically, if space-time's causal structure is acyclic (i.e., devoid of closed timelike curves), does NSS guarantee the acyclicity of information-theoretic causal structures as well?

Surprisingly, the latter question has been shown to have a negative answer: NSS does not preclude causal loops, even in Minkowski space-time \cite{VVC_Letter}. 
This was demonstrated by constructing a cyclic causal model embedded in 1+1-dimensional Minkowski space-time, where the existence of a causal loop could be operationally verified through interventions without violating NSS.
Whether such loops can exist in Minkowski space-times with higher spatial dimensions, and the principles needed to rule them out, remain key open questions.
Addressing these issues is fundamental to understanding the principles necessary for emergence of acyclic causal structures in a physical theory, which ensure a clear causal order where one event influences another without reciprocal influence.

However, answering such questions is conceptually and technically challenging.
From an information-theoretic perspective, it requires methods for inferring causation from patterns of signalling among systems and operationally identifying redundant variables in signalling relations in the presence of fine-tuning, which is little explored even classically.
Moreover, the causal loop identified in 1+1-dimensional Minkowski space-time \cite{VVC_Letter} cannot be embedded in higher-dimensional Minkowski space-time without superluminal signalling.
It was conjectured in \cite{VVC_Letter} that such loops may generally be precluded in higher dimensions due to geometric properties of light cone intersections.
The companion paper \cite{VVC} of \cite{VVC_Letter} classified several more complex classes of operationally detectable causal loops (non-exhaustively), whose space-time embeddings must be accounted for to prove this conjecture.
Proving the conjecture also requires a rigorous formalization of the geometric properties distinguishing 1+1-dimensional and higher-dimensional Minkowski space-times, focusing on order-theoretic properties of the space-time's causal structure rather than continuous manifold-like properties.

In this work, we make progress on a number of these fronts. We derive new results for causal inference from signalling relations, shed light on different forms of fine-tuning, their relationship with signalling and their role in causal inference, and identify an order-theoretic property of space-time – referred to as \enquote{conicality} – that distinguishes Minkowski space-times of one and higher spatial dimensions.
Finally, by connecting the information-theoretic and spatio-temporal notions of causation, our main result demonstrates an alignment between information-theoretic and spatio-temporal notions of causal ordering in two scenarios: (1) faithful (non-fine-tuned) causal models embedded in arbitrary space-times and (2) arbitrary causal models embedded in conical space-times. 

To illustrate the physical significance of this main result, consider the following example. Suppose $X$ and $Y$ are individual variables where $X$ signals to $Y$: then causal inference tells us that $X$ is a cause of $Y$ and NSS would require $Y$ to be embedded in the future light cone of $X$.
The two orderings on the variables, obtained from causal inference and NSS in space-time agree.
However, for signalling relations involving sets of variables,
e.g.\ $X_1\cup X_2$ signals to $Y$,
we find that the causal ordering we can obtain from causal inference, and the ordering relative to the light cone structure that we can obtain from NSS in a space-time do not generally agree. Our main result identifies sufficient conditions under which these orderings ``align'', namely in situations (1) and (2) described above. More precisely, the result holds even when (1) rules out only a specific class of fine-tuned models rather than all of them. The summary of contributions in the following subsection gives further intuition on what this ``alignment'' of the orders entails and clarifies the relevant class of fine-tuned models, the technical definitions of which can be found in the main text.

These findings highlight connections between space-time geometry and causal inference, showing that conditions on the embedding simplify for causal models embedded consistently with NSS in conical space-times such as $d$+1-Minkowski space-times with $d>1$. Moreover, for faithful causal models, NSS in any acyclic space-time does imply no causal loops \cite{VVC}. Therefore,  our main result, showing common properties of conical space-times and faithful causal models, lends support to the conjecture that the principle of no superluminal signalling is sufficient to rule out (operationally detectable) causal loops in conical space-times. This conjecture was motivated in \cite{VVC_Letter}, the concepts and techniques introduced in the current work allow to formalize it while helping to overcome a number of challenges towards proving it.

To formally state and prove these results, we build on a previous work \cite{VVC}, which develops a causal modelling framework applicable to a general class of physical theories, while permitting fine-tuned and cyclic causal influences, and relates this to space-time structure. This framework, referred to here as the \emph{affects framework}, captures general signalling possibilities through \emph{higher-order affects relations} and by embedding causal models in a space-time, the NSS principle is formalized as a graph-theoretic compatibility condition between the higher-order affects relations and the light cone structure of space-time. This formalism enables a theory-independent study of the interplay of the two causality notions, and in particular, the $1+1$-Minkowski space-time embeddable causal loop found in \cite{VVC_Letter} was constructed within this formalism.

We advance the affects framework by characterizing operational properties of higher-order affects relations and identifying order-theoretic properties of space-time geometry. Our findings reveal parallels between these domains, while providing tools for addressing fine-tuning and redundancies, yielding further avenues for understanding causality at the intersection of information theory and space-time geometry. We hope that further investigations in these directions, which remain relatively less explored, can lead to useful insights for understanding practical (classical, quantum and post-quantum) information processing scenarios in space-time as well as how space-time structure might emerge from informational structures, at a more fundamental level.

\section{Summary of contributions}

We now outline the main contributions of this work in a more precise and technical manner using the concept of affects relations, while aiming for this summary to be  self-contained even for those without prior knowledge of the affects framework \cite{VVC} (which is reviewed in \cref{sec:review}). 
A (higher-order) affects relation, denoted as $X \affects Y \given \doo(Z)$, carries three arguments, which are disjoint sets of random variables, and it captures that an agent who intervenes on $X$ can signal to an agent who can observe data on $Y$ and is given information about interventions performed on $Z$.\footnote{Technically, this corresponds to an unconditional (higher-order) affects relation. The concept of conditional higher-order affects relations is also introduced in \cite{VVC} by including a fourth argument $X \affects Y \given \doo(Z),W$ which captures an additional post-selection on $W$ (without interventions), i.e.\ the agent receiving the signal is also given information about this post-selection. We focus on the unconditional case in the main text and generalize many of the results to conditional relations in the appendix.} Depending on whether $Z$ is empty or not, we have a zeroth-order or a higher-order relation. Here, $X$ and $Z$ are interventional arguments while $Y$ is a purely observational argument. While affects relations refer to classical variables (such as measurement settings and outcomes), they can generally arise from operations performed on non-classical systems of any underlying theory. In the following, we will refer to a random variable (which goes into the argument of an affects relation) as a node, as such variables form the (observable) nodes or vertices of the directed graphs representing information-theoretic causal structures.

{\bf I. Characterizing affects relations and applications to causal inference}  
Generally, an affects relation $X\affects Y\given \doo(Z)$ can carry some redundancies and may be operationally equivalent to a ``reduced'' affects relation on sets of smaller cardinality. The simplest case where this can happen is when $X$, $Y$ or $Z$ contain nodes which are causally disconnected from the rest and therefore redundant for the affects relation. Identifying such redundancies operationally (even in scenarios where the causal structure is not known) is important both for information-theoretic causal inference,
and for understanding precisely how relativistic causality principles in space-time constrain information processing protocols.  In \cite{VVC}, only the reducibility of affects relations in the first interventional argument $X$ was defined. In \cref{sec:affects-red}, we extend the concept of reducibility to all the different arguments of the affects relation while discussing its operational significance in identifying redundancies in these arguments. 

In \cref{sec:affects-clus}, we introduce the concept of \emph{clustering} and relate it to reducibility (\cref{theorem: clus_irred}). Clustering in an argument captures the absence of affects relations of the same form when replacing that argument with strict subsets thereof. E.g., if $X\affects Y\given \doo(Z)$, but $s_X\naffects Y\given \doo(Z)$ for all $s_X\subsetneq X$, then the original relation is clustered in the first argument. We prove in \cref{lemma: clus_finetune} that clustering (in any argument) is a signature of fine-tuning and thereby distinguish between different types of operationally detectable fine-tuning, i.e., the fine-tuning of underlying causal mechanisms which can be detected from the correlations and the affects relations. Such possibilities for information transfer between sets of systems that is not detectable within subsets of those systems \cite{Cotler2019}, also have applications in cryptography \cite{PhysRevA.59.1829, PhysRevA.61.042311, PhysRevA.59.162}, quantum error correction \cite{Gottesmanphd, Cleve1999} and distributed information processing protocols in space-time \cite{Kent_2012_summon, Kent_2012_tasks, Hayden2016}. Relations between irreducibility, clustering and fine-tuning are illustrated in \cref{fig:relations-coarse}.

The causal inference implications of these concepts are then presented in \cref{sec:affects-to-cause}. Typical causal inference results and algorithms assume no fine-tuning due to complications for causal inference that arise in the presence of fine-tuning (see \cite{Pearl2009, Costa2016} for examples). Our results show that in certain fine-tuned models, the absence of signalling between some of the nodes can be employed for causal inference. More generally, our work contributes to open questions regarding causal inference in presence of fine-tuning, as raised in Section IX.e of \cite{VVC}.

{\bf II. Order-theoretic properties of space-time} Modelling space-time structure as a partially ordered set, an order-theoretic property of the causal structure of space-time is introduced in \cref{sec:poset}, namely, \emph{conicality}. We then show that $d$+1-dimensional Minkowski space-times with $d>1$ are conical space-times (\cref{thm:conical}), while this is not the case for 1+1-dimensional Minkowski space-time. Conicality captures the requirement that the joint future region $f(L)$ (intersection of future light cones) of a set $L$ of space-time points uniquely determines the location of all points in $L$ that contribute non-trivially to $f(L)$.\footnote{
    For example, if $L=\{p,q\}$ with $p\prec q$, then $f(L)=f(q)$ and $p$ does not contribute non-trivially to $f(L)$.}
In Minkowski space-time with $d=1$, we can have two distinct pairs of space-like separated points $L_1=\{a,b\}$ and $L_2=\{x,y\}$ that have the same joint future, $f(L_1)=f(L_2)$ and conicality is thus violated (\cref{fig:d}). For $d>1$ however, our result implies that this cannot happen for any configuration of space-time points. 

Intuitively, this notion may capture the idea of ``fine-tuning'' of the space-time embedding suggested in \cite{VVC}:
To violate conicality as aforementioned, both $a$ and $x$ and $b$ and $y$, respectively, need to be light-like to one another.
This means that slight perturbations to these points which move them outside the light-like surface of one another, would no longer satisfy $f(L_1)=f(L_2)$ required to witness the failure of conicality.

We discuss further order-theoretic properties that distinguish Minkowski space-time for $d=1$ and $d>1$ spatial dimensions in \cref{sec:poset-props}. Intuitively, these distinctions are related to the fact that in $d=1$ the joint future of any two points has the same geometry as the light cone of an individual point (the unique earliest point in this joint future), but the geometries of these regions will differ for $d>1$ as there is no longer a unique earliest point in the joint future of any two points.

{\bf III. Correspondence between causal inference and space-time geometry} In \cref{sec:compat-indec}, we study the properties of affects relations and causal models, which can be embedded in a space-time compatibly, i.e., without violating no superluminal signalling (NSS). We find that an affects relation $X\affects Y\given \doo(Z)$ irreducible in the first and third arguments implies that all nodes in $X$ and in $Z$ are causes of some node in $Y$ i.e., the set $Y$ of nodes is ordered later than each node in $X$ and $Z$ relative to the \emph{information-theoretic causal order}. However, we observe that when embedding these nodes in space-time, imposing that the affects relation does not lead to superluminal signalling does not generally imply an analogous ordering of nodes relative to the \emph{light cone structure of the space-time}: the set $Y$ of nodes can generally be jointly accessible outside the future light cone of some nodes in $X$ and $Z$.
That is, the relativistic principle of NSS and purely informational principles of causal inference generally impose different ordering constraints on the relevant operational events (here, the arguments of an affects relation). Denoting $\Fut_s(\XX)$ to be the intersection of the future light cones of all elements of a set of variables $X$ (with $|X| \ge 1$) which are embedded in a space-time (which is where they are jointly accessible)\footnote{As we will see later, $X$ denotes the random variable while $\XX$ denotes the ordered random variable, which the random variable $X$ along with a space-time location at which it is embedded. Further, we will later write $\Fut(\XX)$ without a subscript for the special case that $\abs{X} = 1$, where it corresponds to the future of the ordered random variable $\XX$.},
a simplified version of the correspondence we are looking for can be stated as follows: Given an affects relation $X\affects Y\given \doo(Z)$ irreducible in $X$ and $Z$, for any $e_{XZ}\in X\cup Z$,
\begin{align}
    \label{eq: correspondence_main_uncond}
    \begin{split}
        \textsf{\textbf{Causal inference: }} & e_{XZ} \text{ is a cause of } Y \\
        \textsf{\textbf{NSS in a space-time: }}    & \Fut_s (e_{\XX\ZZ}) \supseteq \Fut_s (\YY)
    \end{split}
\end{align}

The causal inference statement is proven to always hold earlier in \cref{sec:affects-to-cause}. We show that the above spatio-temporal ordering is generally not implied by NSS in arbitrary space-times and for arbitrary affects relations (\cref{ex:non-degenerate}). However, we prove that in the following two cases, the information-theoretic order obtained from causal inference and the spatio-temporal order obtained from the NSS principle align in the sense of \cref{eq: correspondence_main_uncond} (see \cref{thm:irr-compat}, \cref{thm:irr-compat-clus}, \cref{thm:fine-tuned-embedding-cond}).

\begin{enumerate}
    \setlength\itemsep{-1pt}
    \item[(1)] affects relations that are \emph{not clustered in the third argument} but embedded compatibly in an arbitrary space-time
    \item[(2)] arbitrary affects relations but embedded compatibly in a \emph{conical space-time}.
\end{enumerate}
In both these cases, a clear ordering emerges between interventional arguments $X$ and $Z$ and the observational argument $Y$ of an irreducible affects relation $X\affects Y\given \doo(Z)$, with the former ordered before the latter relative to both the information-theoretic and spatio-temporal causal orders as per \cref{eq: correspondence_main_uncond}.\footnote{Notice that the ordering in the second line of \cref{eq: correspondence_main_uncond} is consistent with the observation that when a space-time event $p$ is in the past light cone of another space-time event $q$, then the future light cone of $p$ contains that of $q$. } This correspondence is illustrated and summarized in \cref{tab:summary}.
For many purposes, these conditions allow to reduce statements regarding causal inference and NSS for higher-order affects relations to equivalent statements in terms of the simpler zeroth-order relations.

Moreover, these two alternative conditions provide a connection between properties of (a) causal models (the property of clustering in (1)), and (b) of the geometry of space-time (the property of conicality in (2)).
As clustering implies fine-tuning, this suggests links between faithful causal models and conical space-times. This result also sheds light on the relation between fine-tuning in the causal model and a notion of fine-tuning in the space-time embedding suggested in \cite{VVC}, which plays an important role in the possibility of embeddable causal loops in 1+1-Minkowski space-time \cite{VVC_Letter}.
Another type of correspondence between affects relations without clustering in the first argument and conical space-times is discussed in \cref{sec:no-irreducible}.

We conclude by discussing the implications of our results and future outlook in \cref{sec:conclusion}, stressing the physical and information-theoretical relevance of our results.

\section{Review of the affects framework}
\label{sec:review}

In this section, we give a brief introduction into the affects framework introduced in previous work \cite{VVC} involving one of us, which formalizes the notion of signalling under minimal and theory-independent assumptions.
This generality makes it applicable to scenarios with non-classical or cylic causal influences as well as scenarios where causal influences may be fine-tuned so that they wash out certain observable correlations or possibilities for signalling. Here we review the concepts through specific examples, and refer the reader to \cref{sec:d-separation} for technical definitions of the concepts mentioned here, but not defined in full generality.

\subsection{Theory-independent causal models}
\label{sec:causal-model}

The affects framework formulates information-theoretic causality independently of a notion of space and time, by building on the \emph{causal modelling} approach. Here, a causal structure is a directed graph $\cG$ where each node can either be \emph{observed} or \emph{unobserved}. Observed nodes correspond to classical random variables or RVs (such as settings and outcomes of measurements) while unobserved nodes can be associated with systems of any physical theory (classical, quantum or post-quantum). The directed edges $\dircause$ are understood as direct causal influences between the systems involved. 

Causal models are then formulated in terms of a directed graph $\cG$ as above over a set of nodes $N$, together with a probability distribution $P_{\cG}$ over a set of RVs, which are in 1-to-1 correspondence to the observed nodes $N_\text{obs} \subseteq N$, where $P_\cG$ satisfies a linking property (the \emph{d-separation property}) relative to $\cG$.
Going forward, we will therefore refer to these RVs as observed RVs and drop the distinction between them and their associated nodes. 

\begin{notation}
    When considering sets of nodes or RVs $N_1, N_2 \subset N$ as well as individual nodes $X \in N$, we will generally denote $X \cong \{ X \}$ and $N_1 N_2 = N_1 \cup N_2$.

\end{notation}

The idea behind this linking property is that graph separation relations between  sets of observed nodes in $\cG$ should imply conditional independences between corresponding variables in the observed distribution $P_{\cG}$.
The graph separation criterion used for this purpose is called \emph{d-separation} (directed separation), which is a standard measure used in the classical \cite{Pearl1990, Pearl2009, arxiv.1302.3595}, quantum \cite{Costa2016, PhysRevX.7.031021, arxiv.1906.10726, Barrett2021} and post-quantum \cite{Henson2014} causal modelling literature.
If two sets of nodes $N_1$ and $N_2$ in $\cG$ are d-separated given a third set $N_3$, potentially empty, we will write $(N_1 \perp^d N_2 | N_3)_{\cG}$. This entails the idea that certain kinds of paths between $N_1$ and $N_2$ are blocked by $N_3$. Conditional independence is defined as follows.

\begin{definition}[Conditional Independence]
\label{def:cond_indep}
    Let $\cG$ be a causal structure associated with a probability distribution $P_{\cG}$. For three disjoint sets of RVs $A, B, C$ with $A$ and $B$ non-empty, we say that $A$ and $B$ are \emph{conditionally independent given} $C$, denoted as $(A \indep B|C)_\cG$, if the marginal of $P_{\cG}$ on these variables 
    satisfies $P_{\cG}(AB|C)=P_{\cG} (A|C) P_{\cG}(B|C)$.
    Note that $\indep$ is hence not a property of the graph alone, but of the distribution associated to its nodes.
\end{definition}

Then the \emph{d-separation property} entails that whenever $X_1$, $X_2$ and $X_3$ are disjoint sets of RVs associated with observed nodes of $\cG$, then 

\begin{equation}
    \label{eq:compat}
    (X_1 \perp^d X_2 | X_3)_{\cG} \quad \implies \quad (X_1 \indep X_2 | X_3)_{\cG}
\end{equation}
In particular, when $X_3=\emptyset$, we write $(X_1 \perp^d X_2)_{\cG} \implies (X_1 \indep X_2)_{\cG}$. When $X_1$ and $X_2$ are individual RVs, $(X_1 \perp^d X_2)_{\cG}$ is equivalent to saying that neither of them is a cause of the other (i.e.\ there are no directed paths $X_1\dircause ... \dircause X_2$ or vice versa, connecting them in $\cG$) and they share no common ancestor (no common node $C$ which has directed paths to both) in the graph. This encapsulates the Reichenbach principle of common cause \cite{Reichenbach1956}, which fundamentally asserts that correlations between events must have some underlying causal explanation (in terms of one event being a cause of another, or common causes).

\begin{remark}
    \label{remark: RV_values}
    In a mild abuse of notation, which is a common convention within causal modelling, we use the same notation to denote a random variable (RV) $X$ and values of the RV. This means that when considering distinct RVs $X$ and $Y$, equations such as $X=Y$ refer to relationships between their values, in this case indicating that both variables are correlated such that they always assume the same \underline{value}. 
\end{remark}

Intuitively, the causal structure captures the flow of information through the network of nodes. How this flow of information is modelled specifically -- the \emph{causal mechanisms} -- is dependent on the respective theory.
For example, in classical probabilistic theories, for each node, the mechanisms are provided by the conditional distributions $P (N|\parent(N))$, where $\parent(N)$ stands for the set of all parents of the node $N$ in the causal structure.
These distributions can be thought of as classical channels or stochastic maps, and in deterministic theories, we would have functions $f_N:\parent(N)\mapsto N$ for each node instead, determining the value of $N$ given the values of all its parents (cf. \cref{remark: RV_values}). Analogously, in quantum causal models, we can refer to conditional density operators $\rho_{N|\parent(N)}$ to characterize the causal mechanisms, these are representations of quantum channels \cite{PhysRevX.7.031021, Barrett2021}.

Commonly, the literature adopts a \emph{bottom-up} approach to causal modelling  (e.g., \cite{Pearl2009, Costa2016, PhysRevX.7.031021, arxiv.1906.10726, Henson2014}) which starts with assumptions on the causal mechanisms, deriving conditions (like d-separation) on the observed distributions. These often focus on faithful (not fine-tuned) and acyclic causal models\footnote{See \cite{Bongers_2021, Barrett2021} for recent developments in bottom-up approaches for cyclic causal models.}. These approaches have proven useful for characterizing the nature of causation and correlations in specific theories, like quantum theory. On the other hand, the affects framework \cite{VVC} pursues a \emph{top-down} approach that applies to any causal mechanisms which satisfy the d-separation property on a level of observed correlations. While it assumes the existence of a causal structure $\cG$ (not necessarily acyclic), it does not assume a specific meaning to the causal arrows $\dircause$, and can be applied to rather general notions of causality in different theories. This makes it well-suited for deriving general results on causal inference and signalling or generic impossibility results for a wide class of theories.

Further, the concept of fine-tuning in a causal model will be important in this work as we do not assume its absence. A causal model is considered \emph{unfaithful} or \emph{fine-tuned} if the converse implication of \cref{eq:compat} is not satisfied for some disjoint sets of observed RVs $X_1$, $X_2$ and $X_3$:
\begin{equation}
    \label{eq:fine-tuning}
    (X_1 \perp^d X_2 | X_3 )_{\cG}
    \quad \,\,\,\not\!\!\!\Longleftarrow \quad
    (X_1 \indep X_2 | X_3)_{\cG}
\end{equation}
This captures the idea that independences in the distribution faithfully reflect the connectivity of the causal graph. This can fail in causal models where the underlying causal mechanisms are fine-tuned to hide certain causal connections from being detectable through probabilistic dependences. Usually, this property is only evaluated in the original causal model (associated with $\cG$).
Within this work, we extend this notion to also account for interventions (described in the next section), which alter the graph $\cG$ in a specific manner.

\subsection{Interventions and (higher-order) affects relations}
\label{sec:affects}

\begin{figure}[t]
	\centering
    \begin{subfigure}[b]{0.4\textwidth}
        \centering
        \begin{tikzpicture}
            \begin{scope}[every node/.style={circle,thick,draw,inner sep=0pt,minimum size=0.8cm}]
                \node (X) at (-3,0) {$S$};
                \node (Y) at (0,0) {$C$};
                \node (Z) at (-1.5,-2) {$G$};
            \end{scope}

            \begin{scope}[>={Stealth[black]},
                          every node/.style={fill=white,circle},
                          every edge/.style=vvarrow]
                \path [->] (X) edge (Y);
                \path [->] (Z) edge (X);
                \path [->] (Z) edge (Y);
            \end{scope}
        \end{tikzpicture}
        \caption{Pre-intervention structure $\mathcal{G}$.}
        \label{fig:pre-intervention}
    \end{subfigure}%
    \begin{subfigure}[b]{0.4\textwidth}
        \centering
        \begin{tikzpicture}
            \begin{scope}[every node/.style={circle,thick,draw,inner sep=0pt,minimum size=0.8cm}]
                \node (X) at (-3,0) {$S$};
                \node (Y) at (0,0) {$C$};
                \node (Z) at (-1.5,-2) {$G$};
            \end{scope}

            \begin{scope}[>={Stealth[black]},
                          every node/.style={fill=white,circle},
                          every edge/.style=vvarrow]
                \path [->] (X) edge (Y);
                \path [->] (Z) edge (Y);
            \end{scope}
        \end{tikzpicture}
        \caption{Post-intervention structure $\mathcal{G}_{\doo(S)}$.}
        \label{fig:post-intervention}
    \end{subfigure}%
	\caption[Pre-intervention, augmented and post-intervention causal structures.]{
        Starting from the original (pre-intervention) causal structure, 
        an intervention on the node $S$ is performed by removing all its incoming edges, yielding the post-intervention causal structure $\cG_{\doo(S)}$.
        For the associated distributions, an intervention then corresponds to fixing $S$ to fix a certain value $S = s$ in the post-intervention causal model.}
	\label{fig:interventions}
\end{figure}

Generally, the same observed probability distribution can admit several different causal explanations.
This fundamentally arises from the fact that correlations are symmetric, while causation, as represented in the causal structure, is an asymmetric, directed relation.
For instance, the existence of correlation between the prevalence of smoking (modelled as a variable $S$) and the incidence of cancer (modelled as a variable $C$) in a population does not necessarily allow to infer that smoking is a cause of cancer, because it is possible to have a common factor $G$ (such as a genetic predisposition) that influences both an individual's likelihood to (become addicted to) smoking and the onset of cancer, and could thereby explain the same correlations.

Therefore, to deduce causal relations, we need to supplement the correlations, collected through \emph{passive observation}
of the variables involved, with \emph{free interventions}, which also actively control the variables. Indeed, such interventions form the basis of controlled trials, for instance where one wishes to deduce if a drug causes recovery from a disease.

In \cref{fig:interventions}, we illustrate how this idea is captured within causal models.
In our example, to infer whether $S$ is a cause of $C$, one performs an intervention on $S$, denoted as $\doo(S)$. Such an intervention would correspond to actively forcing people to start or quit smoking, which may be practically infeasible due to ethical reasons, but can be considered in theory. 
This amounts to removing all incoming arrows of $S$, and fixing the RV to a particular value. This captures the basic assumption of such experiments, that the intervention choice is independent of prior causes of $S$, such as the common genetic factor $G$.
The result is a post-intervention causal model, built from the post-intervention causal structure $\cG_{\doo(S)}$ and the associated post-intervention distribution $P_{\cG_{\doo(S)}}$. Then, by comparing if the post-intervention distribution $P_{\cG_{\doo(S)}}(C|S)$ is distinct from the original distribution $P_{\cG}(C)$, we can infer whether or not $S$ is a cause of $C$. We would like to highlight that $P_{\cG_{\doo(S)}}$ is related to but can generally not be fully inferred from the pre-intervention distribution $P_{\cG}$ alone \cite{Pearl2009, VVC}.

This notion of interventions is independent of the theory under consideration and its respective causal mechanisms.
Operationally, this allows to capture a notion of \textit{signalling}: If an operation performed at node $X$, encoded in the choice of intervention $\doo(X)$, yields an observably distinct distribution at another node $Y$, we can understand this as $X$ \textit{signalling to} $Y$. More generally, we can consider an additional set of RVs $Z$ and ask whether $X$ can signal to $Y$ given that some interventions have been performed on $Z$.
This idea is formalised through \emph{affects relations}, where $Z$ being trivial or non-trivial delineates 0$^\text{th}$-order and higher-order affects relations.

\begin{definition}[(Unconditional Higher-Order) Affects Relations]
    \label{def:affects}
    Consider a causal model over a set $S$ of observed nodes, associated with a causal structure $\mathcal{G}$.
    For pairwise disjoint subsets $X, Y, Z \subset S$, with $X, Y$ non-empty, we say
    \begin{equation}
        X \, \text{affects} \,\, Y \, \text{given} \, \doo(Z) \, ,
    \end{equation}
    which we alternatively denote as
    \begin{equation}
        X \vDash Y \,|\, \doo(Z) \, ,
    \end{equation}
    if there exist values $x$ of $X$ and $z$ of $Z$  such that
    \begin{equation}
        P_{\mathcal{G}_{\doo(XZ)}} (Y | X=x, Z=z) \neq
        P_{\mathcal{G}_{\doo(Z)}} (Y | Z=z)
    \end{equation}
    If $Z \neq \emptyset$, we have a \emph{higher-order (HO) affects relation}.
    More specifically, it is also called a \emph{$\abs{Z}^\text{th}$-order affects relation}.
\end{definition}
We understand affects relations as the \emph{formalization of signalling} to be used in the remainder of this work.
We will contrast this definition with notions from previous literature in \cref{sec:uncorrelated-signalling}, exploring possibilities for signalling beyond correlations, while stating a variety of examples for affects relations over the course of \cref{sec:affects-new}.

Note that for $X:=S$, $Y:=C$ and $Z=\emptyset$, this definition recovers the intuition for signalling explained for the particular example of smoking and cancer.

Even more generally, one can consider conditional affects relations $X \affects Y \given \doo(Z), W$, which are conditioned on a set of RVs $W$ that are not intervened upon.
A definition of these will be provided in \cref{sec:conditional}.
However, in \cite{VVC} it is shown that each such conditional affects relation implies an unconditional affects relation $X \affects YW$ (see also \cref{thm:decondition}), and we will therefore disregard conditional affects relations and only work with unconditional relations for the bulk of this work.

The next lemma highlights how affects relations can be used to infer information about the causal structure,
and uses the following definition of \emph{cause}.

\begin{definition}[Cause]
\label{def:cause}
    We say that a set $X$ of nodes is a cause of a set $Y$ of nodes in a causal model if $\exists$ $e_X\in X$ and $e_Y\in Y$ with a directed path $e_X\dircause \ldots \dircause e_Y$ between them in the graph $\mathcal{G}$ representing the causal structure of the model. 
\end{definition}

\begin{lemma}
    \label{thm:affects-to-cause-first}
    \label{thm:affects-to-cause-new}
    Let $S$ be a set of RVs in a causal model and $X, Y, Z \subset S$ disjoint. Then
    $X \affects Y \given \doo(Z) \implies X$ is a cause of $Y$.
\end{lemma}

We have seen how a given causal model on a graph $\cG$ can be associated with a family of post-intervention causal models associated with the derived graphs $\cG_{\mathrm{do}(S)}$. Accordingly, 
we can extend the definition of fine-tuning from \cref{eq:fine-tuning} and in this work, we will say that a causal model is fine-tuned if the non-implication of \cref{eq:fine-tuning} holds in any post-intervention causal model obtained from the original one. Further details are found in \cref{sec:d-separation}.

\subsection{Affects relations capture signalling beyond correlations}
\label{sec:uncorrelated-signalling}

In the quantum information community, the notion of signalling is most commonly discussed in the context of protocols where the measurement settings of parties are freely chosen and modelled as parentless variables. A parentless variable, say $X$ (a setting) is said to signal to some outcome variable, say $A$ if there exist distinct values $x$ and $x'$ of $X$ and a value $a$ of $A$ such that $P(A=a|X=x)\neq P(A=a|X=x')$ (which we denote as $P(A|X=x)\neq P(A|X=x')$).
This condition encompasses precisely what it means for $A$ to be correlated with $X$, i.e., for parentless variables (or sets of parentless variables) $X$, correlation with another variable (or a set of variables) $A$ is equivalent to signalling. 

Based on this, we consider the case where (some of) the variables in $X$ have non-trivial parents. 
In this case, by means of interventions on $X$, a natural definition of signalling one might consider is the following. 
\begin{equation}
    \label{eq:correlated-signalling}
    \exists x, x' : P_{\mathcal{G}_{\doo(X)}}(Y|X=x)\neq P_{\mathcal{G}_{\doo(X)}}(Y|X=x')
\end{equation}
This is equivalent to saying that $X$ and $Y$ are correlated in the post-intervention causal model on $\mathcal{G}_{\doo(X)}$. Concepts such as average causal effect used in operationally identifying causation in classical and non-classical causal models follow a similar definition \cite{Pearl2009, Gachechiladze_2020}. It is shown in \cite{VVC} (Lemma IV.2) that this implies $X\affects Y$, but the converse is not true (Example IV.5).
This illustrates that affects relations strictly capture more general ways of signalling than the alternative definition for signalling given in \cref{eq:correlated-signalling} informed by the previous literature.

This generality of affects relations has a clear operational significance: the alternative definition exclusively captures the possibility of signalling by means of two distinct interventions on $X$, while affects relations additionally capture the additional possibility of signalling through the very fact that some non-trivial intervention was performed as opposed to passive observation (even when the exact choice of intervention does not matter).

These two cases can only diverge when $X$ is not parentless. In this case, we can have
\begin{equation}
    \forall \, x \ \text{of} \ X : \quad\!
    P_{\mathcal{G}_{\doo(X)}} (Y | X=x) = 
    P_{\mathcal{G}_{\doo(X)}} (Y) \neq
    P_{\mathcal{G}} (Y),
\end{equation}
yielding $X\affects Y$ although $X$ and $Y$ are uncorrelated in $\mathcal{G}_{\doo(X)}$. The same is true for higher-order affects relations, as illustrated in \cref{eg: red3_2} where we will have $X \affects Y \given \doo(Z)$ even though $P_{\cG_{\doo(XZ)}} (Y|XZ) = P_{\cG_{\doo(XZ)}} (Y|Z)$.

\section{Characterization of signalling in the affects framework}
\label{sec:affects-new}
As discussed in the previous section, affects relations formalize general possibilities for signalling in a theory-independent manner, while accounting for signalling via interventions on variables that have parents in the causal structure, and capturing operational possibilities for signalling that are not encoded in correlations alone.

However, affects relations can still contain certain redundancies. In this section, we address this issue by providing an operationally motivated notion of irreducibility for each of the four arguments of an affects relation, generalizing on the notion for the first argument introduced in \cite{VVC}.
Irreducibility of an affects relation in a given argument will capture the absence of any reduced affects relation (where that argument is replaced by a strict subset thereof) that is operationally equivalent to the original relation. We also introduce a related notion called clustering for each argument of an affects relation, which formalizes the absence of such reduced affects relations altogether.
For instance, $X_1X_2\affects Y$ but $X_1\not\affects Y$, $X_2\not\affects Y$ is an example of clustering in the first argument.

Having defined these new and operationally motivated concepts, we link irreducibility and clustering, showing the application of these concepts for detecting causal fine-tuning using affects relations, and we conclude by exploring implications of these properties for causal inference.
In doing so, we will employ the relations between conditional dependencies and affects relations encapsulated in the following lemma.

\begin{restatable}{lemma}{affectsCorr}
    \label{lemma: aff_corr}
    $X\not\vDash Y|\mathrm{do}(Z)  \implies (X\upmodels Y|Z)_{\cG_{\mathrm{do}}(XZ)}$.
\end{restatable}
\begin{proof}
    See \cref{proof:affectsCorr}.
\end{proof}

\medskip
Going forward, we will focus on the case of unconditional affects relations, and accordingly, on the first three arguments ($X$, $Y$ and $Z$), as these admit the most useful intuition.
This is due the fact that (by \cref{thm:decondition}) for each conditional affects relation $X \affects Y \given \doo(Z), W$, we can infer an unconditional affects relation $X \affects YW \given \doo(Z)$, which is equivalent in terms of causal inference and will allow the application of the concepts of irreducibility and clustering in the second argument to the full set $YW$.
However, most results of this section fully generalize when applied to conditional affects relations, as explicitly carried out in \cref{sec:conditional}.
In the same appendix, a detailed treatment of the properties of clustering and irreducibility in the additional fourth argument of such conditional affects relations is also provided.

\subsection{Reducibility of affects relations in different arguments}
\label{sec:affects-red}

In this section, we will define multiple concepts relating to the absence of certain affects relations and the presence of certain others.
It is to be understood that these concepts are always defined relative to some given set $\mathscr{A}$ of affects relations.
If a causal model is specified, then $\mathscr{A}$ is the set of all affects relations in the model, otherwise the set $\mathscr{A}$ must be explicitly specified when applying these concepts. 

\begin{definition}[Reducibility in the first argument \cite{VVC}]
\label{def: red1}
\label{def:reducible}
 We say that an affects relation $X \vDash Y | \mathrm{do}(Z)$  is reducible in the first argument (or Red$_1$) if there exists a non-empty subset $s_X\subsetneq X$ such that $s_X \not\vDash Y | \mathrm{do}(\tilde{s}_XZ)$, where $\tilde{s}_X:=X\backslash s_X$. Otherwise, we say that $X \vDash Y | \mathrm{do}(Z)$ is irreducible in the first argument and denote it as Irred$_1$.
\end{definition}

Using this definition, the following lemma is proven in \cite{VVC}. 
\begin{lemma}
\label{lemma: red1}
    If $X \vDash Y | \mathrm{do}(Z)$ is a Red$_1$ affects relation, then there exists $\tilde{s}_X\subsetneq X$ such that  $\tilde{s}_X \vDash Y | \mathrm{do}(Z)$ holds. 
\end{lemma}

\textbf{Operational motivation for definition} Red$_1$ captures the idea that the original affects relation $X \vDash Y | \mathrm{do}(Z)$ and the reduced affects relation $\tilde{s}_X \vDash Y | \mathrm{do}(Z)$ (for $\tilde{s}_X\subsetneq X$) carry the same information. Writing out these affects relations, we have
\begin{align}
    \begin{split}
  P_{\cG_{\mathrm{do}(XZ)}}(Y|XZ)&\neq  P_{\cG_{\mathrm{do}(Z)}}(Y|Z)\\
        P_{\cG_{\mathrm{do}(\tilde{s}_XZ)}}(Y|\tilde{s}_XZ)&\neq  P_{\cG_{\mathrm{do}(Z)}}(Y|Z)   
    \end{split}
\end{align}
Notice that the right hand side of the two expressions are the same, and Red$_1$ requires the left hand sides to be identical (which is equivalent to saying that $s_X \not\vDash Y | \mathrm{do}(\tilde{s}_XZ)$). Once this is imposed, the two affects relations carry the same information, and are expressed by equivalent expressions.

\begin{example}{(Motivating example for Red$_1$)}
    Consider the simple causal model where $X_1$ is a cause of $Y$ with $Y=X_1$, while $X_2$ is an additional causally disconnected node. Here we would expect $X_1X_2\vDash Y$ to be reducible to $X_1\vDash Y$ as $X_2$ is clearly redundant. Indeed this is the case using the above definition, as we have $X_2\not\vDash Y |\mathrm{do}(X_1)$ in this example. 
\end{example}

We apply a similar logic to define reducibility in the remaining arguments. We start with the third argument as its operational motivation is closer to Red$_1$, since both the first and third arguments correspond to nodes on which active interventions have been performed.

\begin{definition}[Reducibility in the third argument]
\label{def: red3}
 We say that an affects relation $X \vDash Y | \mathrm{do}(Z)$  is reducible in the third argument (or Red$_3$) if there exists a non-empty subset $s_Z\subseteq Z$ such that both the following conditions hold, where $\tilde{s}_Z:=Z\backslash s_Z$
 \begin{itemize}
     \item  $s_Z \not\vDash Y | \mathrm{do}(X\tilde{s}_Z)$
     \item $s_Z \not\vDash Y | \mathrm{do}(\tilde{s}_Z)$
 \end{itemize}
Otherwise, we say that $X \vDash Y | \mathrm{do}(Z)$ is irreducible in the third argument and denote it as Irred$_3$.
\end{definition}

\begin{restatable}{lemma}{redThree}
\label{lemma: red3}
    If $X \vDash Y | \mathrm{do}(Z)$ is a Red$_3$ affects relation, then there exists $\tilde{s}_Z\subsetneq Z$ such that  $X \vDash Y | \mathrm{do}(\tilde{s}_Z)$ holds. 
\end{restatable}
\begin{proof}
	Suppose $X \vDash Y | \mathrm{do}(Z)$ holds and is a Red$_3$ affects relation. Then writing out this affects relation along with the two non-affects relations implied by the reducibility (while recalling that $s_Z\cup\tilde{s}_Z=Z$), we have
	\begin{align}
	\label{eq: red3_proof}
		\begin{split}
		  P_{\cG_{\mathrm{do}(XZ)}}(Y|XZ)&\neq  P_{\cG_{\mathrm{do}(Z)}}(Y|Z)\\
			P_{\cG_{\mathrm{do}(XZ)}}(Y|XZ)&=  P_{\cG_{\mathrm{do}(X\tilde{s}_Z)}}(Y|X\tilde{s}_Z)\\
			 P_{\cG_{\mathrm{do}(Z)}}(Y|Z)&=  P_{\cG_{\mathrm{do}(\tilde{s}_Z)}}(Y|\tilde{s}_Z)\\
		\end{split}
	\end{align}
	Taken together, these imply that $P_{\cG_{\mathrm{do}(X\tilde{s}_Z)}}(Y|X\tilde{s}_Z)\neq P_{\cG_{\mathrm{do}(\tilde{s}_Z)}}(Y|\tilde{s}_Z)$ which is equivalent to $X \vDash Y|\mathrm{do}(\tilde{s}_Z)$.
\end{proof}

\medskip
\textbf{Operational motivation for the definition} As with the case of Red$_1$, Red$_3$ captures that the original affects relation $X \vDash Y|\mathrm{do}(Z)$ and the reduced (in the third argument) affects relation $X \vDash Y|\mathrm{do}(\tilde{s}_Z)$, for $\tilde{s}_Z\subsetneq Z$ carry the same information, while ensuring that it recovers the expected notion of redundancy in simple test examples. 
Noting that the original and reduced relations are equivalent to the following two conditions respectively,
\begin{align}
\label{eq: red3_motivation}
    \begin{split}
        P_{\cG_{\mathrm{do}(XZ)}}(Y|XZ)&\neq  P_{\cG_{\mathrm{do}(Z)}}(Y|Z),\\
       P_{\cG_{\mathrm{do}(X\tilde{s}_Z)}}(Y|X\tilde{s}_Z)&\neq P_{\cG_{\mathrm{do}(\tilde{s}_Z)}}(Y|\tilde{s}_Z).
    \end{split}
\end{align}

We observe that the two conditions in \cref{def: red3} ensure that the left-hand sides and right-hand sides of the expressions for the original and reduced affects relations are identical. Specifically, this corresponds to $s_Z \not\vDash Y \given \mathrm{do}(\tilde{s}_Z)$ and $s_Z \not\vDash Y \given \mathrm{do}(\tilde{s}_Z X)$. Unlike the case of Red$_1$, where the right-hand sides of the expressions for the original and reduced affects relations are equal by default, here we require two conditions because this equality does not hold by default in Red$_3$. Consequently, an additional level of subtlety arises. In the above, we equate the two left-hand sides and the two right-hand sides of the relevant expressions. This corresponds to equating the respective pre- and post-intervention distributions of the two affects relations.
One might consider an alternative approach, equating the left-hand side of one expression with the right-hand side of another, to derive a different condition for reducibility in the third argument, which would also imply the equivalence of the expressions in \cref{eq: red3_motivation}.
However, in doing so, we would disregard the information specifying which distribution is pre-/post-intervention, even though it is operationally accessible via the set of nodes intervened upon.
Therefore, a more natural definition of operational equivalence of the original and reduced affects relation is the former one which account for this additional, operationally accessible information.
We can further confirm this argument with a simple test example for ruling out the alternative definition.

\begin{example}{(Motivating example for Red$_3$)}
    Let $X$, $Y$, and $Z$ be individual random variables where $X$ is a cause of $Y$ and $Z$ is a causally disconnected node. In this scenario, we have $X \vDash Y | \mathrm{do}(Z)$ and $X \vDash Y$ in general, but we would expect the former to be reducible to the latter since $Z$ is entirely redundant. This is indeed the case by \cref{def: red3} because $Z \not\vDash Y$ and $Z \not\vDash Y | \mathrm{do}(X)$. However, the alternative definition mentioned above, which mixes the pre- and post-intervention distributions, would incorrectly classify $X \vDash Y | \mathrm{do}(Z)$ as irreducible despite $Z$'s redundancy.\footnote{One could nonetheless consider a definition of reducibility based on the logical OR of the two conditions discussed, which are two ways by which the expressions in \cref{eq: red3_motivation} can be equivalent. This alternative would classify a smaller subset of affects relations as irreducible compared to \cref{def: red3}, thus limiting the generality of our main results. The current results, which use \cref{def: red3} to explore the consequences of Irred$_3$ affects relations, are more general and would imply the same for this alternative definition.}
\end{example}

The following examples both illustrate affects relations $X\vDash Y|\mathrm{do}(Z)$ which are Irred$_3$, but show that the two conditions of \cref{def: red3} can be independently violated.\footnote{For general causal models, that is. For faithful causal models, we demonstrate in \cref{sec:clus-prop} that the first condition implies the second one.}
\begin{example}{(Violating the first Red$_3$ condition)}
\label{eg: red3_1}
Consider a one-time pad over binary variables where $X\longrsquigarrow Y$ and $Z\longrsquigarrow Y$, with $X$ and $Z$ uniformly distributed and $Y=X\oplus Z$. Here, we have $X\vDash Y|\mathrm{do}(Z)$, $Z\not\vDash Y$ and $Z\vDash Y|\mathrm{do}(X)$. $Z$ is a singleton and hence the only possible non-empty subset $s_Z$ here is $Z$ itself.
Thus $X\vDash Y|\mathrm{do}(Z)$ is an Irred$_3$ affects relation, it violates the the first condition of \cref{def: red3} for $s_Z=Z$ (since $Z\vDash Y|\mathrm{do}(X)$). However it satisfies the second condition for the same $s_Z$, (since $Z\not\vDash Y$).
\end{example}

\begin{example}{(Violating the second Red$_3$ condition)}
\label{eg: red3_2}
 This is identical to Example IV.4 from \cite{VVC} and illustrated in \cref{fig: eg_red3_2}. Consider the causal structure $\mathcal{G}$ of \cref{fig: eg_red3_2a} with all nodes being binary variables, and the causal model where $W$ is uniformly distributed, $Y=X\oplus Z\oplus W$, $Z=X$ and $X=W$. This gives us $Y=W$ in $\mathcal{G}$ with $P_{\mathcal{G}}(Y)=P_{\mathcal{G}}(W)$ being uniform. Consider an intervention on $X$ associated with the post-intervention graph $\mathcal{G}_{\mathrm{do}(X)}$ of \cref{fig: eg_red3_2b}. Here, we no longer have $X=W$ (but the remaining functional dependences hold). We still have $Y=W$ (since $Z=X$) and this tells us that $P_{\mathcal{G}_{\mathrm{do}(X)}}(Y|X)$ is also uniform, independently of the value of $X$. Under interventions on $Z$, we obtain the graph $\mathcal{G}_{\mathrm{do}(Z)}$ of \cref{fig: eg_red3_2c} where $Z=X$ no longer holds. We obtain $Y=Z$ in this graph since $X=W$ which implies that $P_{\mathcal{G}_{\mathrm{do}(Z)}}(Y|Z)$ is deterministic. Since $Y=X\oplus Z\oplus W$ still holds, this tells us that $Z\vDash Y$. Finally, under joint intervention on $X$ and $Z$, we obtain the graph $\mathcal{G}_{\mathrm{do}(XZ)}$ of \cref{fig: eg_red3_2d} where we neither have $Z=X$ nor $X=W$. It is easy to see that $P_{\mathcal{G}_{\mathrm{do}(XZ)}}(Y|XZ)$ is then uniform, independent of the values of $X$ and $Z$, since $W$ is uniform. Since $P_{\mathcal{G}_{\mathrm{do}(XZ)}}(Y|XZ)$ differs from $P_{\mathcal{G}_{\mathrm{do}(Z)}}(Y|Z)$ but not from $P_{\mathcal{G}_{\mathrm{do}(X)}}(Y|X)$, we have $X\vDash Y|\mathrm{do}(Z)$ and $Z\not\vDash Y|\mathrm{do}(X)$. Again, as in \cref{eg: red3_1}, the only possible non-empty subset $s_Z$ of $Z$ is $Z$ itself, and the affects relation $X\vDash Y|\mathrm{do}(Z)$ of this example is also Irred$_3$, but in this case, it satisfies the first condition (since $Z\not\vDash Y|\mathrm{do}(X)$) but violates the second condition (since $Z\vDash Y$) of \cref{def: red3}.
\end{example}

Clearly $Z$ is non-redundant in the affects relation $X\vDash Y|\mathrm{do}(Z)$ in both of the above examples, and these examples independently motivate the relevance of both conditions of \cref{def: red3}.

We conclude by introducing the notion of affects relations that are reducible in their second argument.
In analogy to the fact that higher-order affects relations are required for defining Red$_1$ even for a 0$^\text{th}$-order relation $X\affects Y$ \cite{VVC}, we require conditional affects relations (\cref{def:affects-cond}) to define Red$_2$ for unconditional relations. Moreover, as the second argument of an affects relation emerges in the main argument of the associated probability distributions (as opposed to other arguments which appear as the conditionals), we additionally gain a condition on conditional independence (denoted by $\indep$, cf.\ \cref{def:cond_indep}). We define Red$_2$ below and subsequently motivate the definition further.

\begin{figure}
    \centering
    \begin{subfigure}[b]{0.35\textwidth}
        \centering
        \begin{tikzpicture}[scale=0.9]
            \node[shape=circle,draw=black] (X) at (0,0.5) {$X$};
            \node[shape=circle,draw=black] (Z) at (-2,2) {$Z$};
            \node[shape=circle,draw=black] (W) at (2,2) {$W$};
            \node[shape=circle,draw=black] (Y) at (0,3.5) {$Y$};
            \draw[vvarrow, arrows={-Stealth}] (X) -- (Y);  \draw[vvarrow, arrows={-Stealth}] (Z) -- (Y);
            \draw[vvarrow, arrows={-Stealth}] (W) -- (Y);
            \draw[vvarrow, arrows={-Stealth}] (W) -- (X);
            \draw[vvarrow, arrows={-Stealth}] (X) -- (Z);
        \end{tikzpicture}
        \caption{A causal structure $\mathcal{G}$}
        \label{fig: eg_red3_2a}
    \end{subfigure}%
    \hspace{0.05\textwidth}
    \begin{subfigure}[b]{0.35\textwidth}
        \centering
        \begin{tikzpicture}[scale=0.9]
            \node[shape=circle,draw=black] (X) at (0,0.5) {$X$};
            \node[shape=circle,draw=black] (Z) at (-2,2) {$Z$};
            \node[shape=circle,draw=black] (W) at (2,2) {$W$};
            \node[shape=circle,draw=black] (Y) at (0,3.5) {$Y$};
            \draw[vvarrow, arrows={-Stealth}] (X) -- (Y);
            \draw[vvarrow, arrows={-Stealth}] (Z) -- (Y);
            \draw[vvarrow, arrows={-Stealth}] (W) -- (Y);
            \draw[vvarrow, arrows={-Stealth}] (X) -- (Z);
        \end{tikzpicture}
        \caption{A causal structure $\mathcal{G}_{\doo(X)}$}
        \label{fig: eg_red3_2b}
    \end{subfigure}

    \vspace{.5cm}
    \begin{subfigure}[b]{0.35\textwidth}
        \centering
        \begin{tikzpicture}[scale=0.9]
            \node[shape=circle,draw=black] (X) at (0,0.5) {$X$};
            \node[shape=circle,draw=black] (Z) at (-2,2) {$Z$};
            \node[shape=circle,draw=black] (W) at (2,2) {$W$};
            \node[shape=circle,draw=black] (Y) at (0,3.5) {$Y$};
            \draw[vvarrow, arrows={-Stealth}] (X) -- (Y);
            \draw[vvarrow, arrows={-Stealth}] (Z) -- (Y);
            \draw[vvarrow, arrows={-Stealth}] (W) -- (Y);
            \draw[vvarrow, arrows={-Stealth}] (W) -- (X);
        \end{tikzpicture}
        \caption{A causal structure $\mathcal{G}_{\doo(Z)}$}
        \label{fig: eg_red3_2c}
    \end{subfigure}%
    \hspace{0.05\textwidth}
    \begin{subfigure}[b]{0.35\textwidth}
        \centering
        \begin{tikzpicture}[scale=0.9]
            \node[shape=circle,draw=black] (X) at (0,0.5) {$X$};
            \node[shape=circle,draw=black] (Z) at (-2,2) {$Z$};
            \node[shape=circle,draw=black] (W) at (2,2) {$W$};
            \node[shape=circle,draw=black] (Y) at (0,3.5) {$Y$};
            \draw[vvarrow, arrows={-Stealth}] (X) -- (Y);
            \draw[vvarrow, arrows={-Stealth}] (Z) -- (Y);
            \draw[vvarrow, arrows={-Stealth}] (W) -- (Y);
        \end{tikzpicture}
        \caption{A causal structure $\mathcal{G}_{\doo(XZ)}$}
        \label{fig: eg_red3_2d}
    \end{subfigure}
    \caption{Pre- and post-intervention causal structures for \cref{eg: red3_2}. Figure derived from \cite{VVC}.}
    \label{fig: eg_red3_2}
\end{figure}

\begin{restatable}[Reducibility in the second argument]{definition}{redTwoDef}
\label{def: red2}
 We say that an affects relation $X \vDash Y | \mathrm{do}(Z)$  is reducible in the second argument (or Red$_2$) if there exists a non-empty subset $s_Y\subsetneq Y$ such that both the following conditions hold, where $\tilde{s}_Y:=Y\backslash s_Y$.
 \begin{itemize}
     \item $X \not\vDash s_Y | \mathrm{do}(Z), \tilde{s}_Y$ (cf.\ \cref{def:affects-cond})
     \item $(s_Y\upmodels \tilde{s}_Y|XZ)_{\cG_{\mathrm{do}(XZ)}}$ or $(s_Y\upmodels \tilde{s}_Y|Z)_{\cG_{\mathrm{do}(Z)}}$
 \end{itemize}
Otherwise, we say that $X \vDash Y | \mathrm{do}(Z)$ is irreducible in the second argument and denote it as Irred$_2$.
\end{restatable}

\begin{restatable}{lemma}{redTwo}
\label{lemma: red2}
    If $X \vDash Y | \mathrm{do}(Z)$ is a Red$_2$ affects relation, then there exists $\tilde{s}_Y\subsetneq Y$ such that $X \vDash \tilde{s}_Y | \mathrm{do}(Z)$ holds. 
\end{restatable}
\begin{proof}
    See \cref{sec:red2}.
\end{proof}

\medskip

\textbf{Operational motivation for the definition} As before, the Red$_2$ property of an affects relation $X \vDash Y | \mathrm{do}(Z)$ captures that it encodes the same information as an affects relation $X \vDash \tilde{s}_Y | \mathrm{do}(Z)$ for a strictly smaller second argument $\tilde{s}_Y \subsetneq Y$. For example, consider a causal model over $X$, $Y_1$ and $Y_2$ where $X\longrsquigarrow Y_1$ is the only edge and we have $Y_1=X$. Then for any distribution over $Y_2$ and $X$, we have $X\vDash Y_1$ and $X\vDash Y_1Y_2$. However, we know that $Y_2$ is entirely superficial in this example, and would expect that the latter affects relation should be reducible to the former one. Expressing the two relations explicitly, we have 
\begin{align}
\label{eq: red2_op1}
    \begin{split}
    P_{\cG_{\mathrm{do}(X)}}(Y_1|X)&\neq P_{\cG}(Y_1),\\
    P_{\cG_{\mathrm{do}(X)}}(Y_1Y_2|X)&\neq P_{\cG}(Y_1Y_2)
    \end{split}
\end{align}
Notice however that in this case, it does not make sense to simply equate the left and right hand sides of the expressions for the original and reduced affects relations, as we did for Red$_1$ and Red$_3$. Even in our simple example, $P_{\cG_{\mathrm{do}(X)}}(Y_1|X)\neq P_{\cG_{\mathrm{do}(X)}}(Y_1Y_2|X)$ and $P_{\cG}(Y_1Y_2)\neq P_{\cG}(Y_1)$.
Therefore, the motivation for the definition of Irred$_2$ is more complex than for the other cases. Moreover, Irred$_2$ will not be as central to our results as Irred$_1$ and Irred$_3$. Therefore, we conclude with a final example and defer further details regarding the rationale and examples for Irred$_2$ to \cref{sec:red2}.

\begin{example}{(Motivating example for Red$_2$)}
    \label{ex:jamming}
    The concept of jamming was first introduced in \cite{PhysRevA.53.3781} as a new form of post-quantum correlation in the context of a tripartite Bell scenario.
    Suppose that $Y$ is the (freely chosen) setting of a party Bob while $A$ and $C$ are outcomes of the remaining two parties Alice and Charlie.
    Jamming correlations allow for $Y$ to jointly signal to $A$ and $C$ (i.e., $P(AC|Y)\neq P(AC)$), without signalling individually i.e., $P(A|Y)=P(A)$, $P(C|Y)=P(C)$.\footnote{The motivation for considering these correlations is the following. Suppose we embed the tripartite scenario in space-time such that the three measurements are space-like separated and $Y$'s future light cone contains the intersection of the future light cones (i.e., the joint future) of $A$ and $C$. Then the joint signalling from $Y$ to $AC$ can only be verified within $Y$'s future light cone.
    Furthermore, since $Y$ does not individually signal to the space-like-separated variables $A$ or $C$, the authors of \cite{PhysRevA.53.3781} argue that joint signalling from $Y$ to $AC$ in this space-time configuration cannot facilitate superluminal signalling.
    Such space-time configurations, which allow for jamming correlations in a Bell-type scenario without enabling superluminal signalling, are indeed possible in Minkowski space-time — specifically, in both 1- and 2-dimensional spatial configurations, as explicitly demonstrated in Appendix C of \cite{VVCJ}.}

    In \cite{VVC_Letter, VVCJ, VVCPR} explicit causal models for such jamming scenarios have been proposed and analyzed in the affects framework. In the simplest case using the three variables mentioned here, the following causal model exhibits jamming, where the signalling is now expressed through affects relations.
    We have $Y\longrsquigarrow A$, $\Lambda\longrsquigarrow A$  and $\Lambda\longrsquigarrow C$ with $P(\Lambda)$ and $P(Y)$ being uniform, $A=\Lambda\oplus Y$ and $C=\Lambda$. Here, $Y\vDash AC$ but $Y\not\affects A$ and $Y\not\affects C$. However both $Y \vDash A \given C$ and $Y \vDash C \given A$ hold.
    Therefore, $Y \affects AC$ is Irred$_2$.
\end{example}

\subsection{Clustering of affects relations: detecting and classifying fine-tuning}
\label{sec:affects-clus}

So far we have focused on the reducibility and irreducibility of affects relations, and seen that an affects relation $X \vDash Y|\mathrm{do}(Z)$ which is reducible in any of the arguments $X$, $Y$ or $Z$ implies a reduced affects relation of the same form where the corresponding argument is replaced by a strict subset of it.
Conversely, how does the irreducibility of $X \vDash Y|\mathrm{do}(Z)$ in a given argument relate to the existence/non-existence of affects relations of the same form involving strict subsets of that argument?
To study this, we introduce a related property to irreducibility, which also applies to different arguments and which we call \emph{clustering}.
It captures the property of an affects relation to hold without the presence of any reduced affects relations (involving strict subsets of a given argument).

\begin{definition}[Clustering in the first argument]
\label{def: clus1}
  An affects relation $X\vDash Y|\mathrm{do}(Z)$ is called clustered in the first argument (denoted Clus$_1$) if $|X|\geq 2$ and there exists no $s_X\subsetneq X$ such that   $s_X\vDash Y|\mathrm{do}(Z)$. 
\end{definition}

\begin{definition}[Clustering in the second argument]
\label{def: clus2}
  An affects relation $X\vDash Y|\mathrm{do}(Z)$ is called clustered in the second argument (denoted Clus$_2$) if $|Y|\geq 2$ and there exists no $s_Y\subsetneq Y$ such that   $X\vDash s_Y|\mathrm{do}(Z)$. 
\end{definition}

\begin{definition}[Clustering in the third argument]
\label{def: clus3}
  An affects relation $X\vDash Y|\mathrm{do}(Z)$ is called clustered in the third argument (denoted Clus$_3$) if $|Z|\geq 1$ and there exists no $s_Z\subsetneq Z$ such that   $X\vDash Y|\mathrm{do}(s_Z)$. 
\end{definition}

For each argument, these definitions capture the intuition that an affects relation can not be ascribed to individual elements of that argument, but emerges only when jointly considering multiple elements.

We note that the third argument is treated slightly differently from the first two arguments, because for any affects relation the latter must necessarily be non-empty while the former need not. The following theorem then shows that clustered affects relations constitute a special case of irreducibility.
\begin{restatable}{theorem}{clusIrreducible}[Clustering implies irreducibility]
\label{theorem: clus_irred}
    For any affects relation  $X\vDash Y|\mathrm{do}(Z)$, Clus$_i$ $\Rightarrow$ Irred$_i$ for all $i\in\{1,2,3\}$.
\end{restatable}
\begin{proof}
    See \cref{proof:clusIrreducible}.
\end{proof}

\medskip
While clustering implies irreducibility, the following examples show that the converse is not true. The conclusions of these examples hold for all choices of distributions over the parentless nodes, and we therefore do not specify them.

\begin{example}{(Irred$_1$ $\not\Rightarrow$ Clus$_1$)}
    \label{ex:1}
Consider a causal model over the nodes $X_1$, $X_2$ and $Y$ with edges $X_1\longrsquigarrow Y$ and $X_2\longrsquigarrow Y$ with $Y=X_1 \cdot X_2$, where $\cdot$ represents binary \texttt{AND}. Let $X_1$ and $X_2$ be non-deterministically distributed.
We then have $X_1\vDash Y$ and $X_1X_2\vDash Y$, which implies that $X_1X_2\vDash Y$ does not satisfy the Clus$_1$ property. However $X_1X_2\vDash Y$ is an Irred$_1$ affects relation since $X_1\vDash Y|\mathrm{do}(X_2)$ and $X_2\vDash Y|\mathrm{do}(X_1)$ both hold.
\end{example}
\begin{example}{(Irred$_2$ $\not\Rightarrow$ Clus$_2$)}
    \label{ex:2}
Consider a causal model over the nodes $X$, $Y_1$ and $Y_2$ which are associated binary variables, with edges $X\longrsquigarrow Y_1 \longrsquigarrow Y_2$ and $X\longrsquigarrow Y_2$ with $Y_1=X$ and $Y_2=Y_1\cdot X$, where $\cdot$ represents binary \texttt{AND}. Then we have $X \affects Y_1$ which implies that $X\vDash Y_1Y_2$ does not satisfy the Clus$_2$ property. The only non-empty and strict subsets of $\{Y_1,Y_2\}$ are $\{Y_1\}$ and $\{Y_2\}$, but we have $X\vDash Y_1$ and $X\vDash Y_2$. Additionally, we observe that both $(Y_1 \not\upmodels Y_2)_{\cG}$ and $(Y_1 \not\upmodels Y_2|X)_{\cG_{do(X)}}$, which means that the second condition of \cref{def: red2} is always violated. Therefore $X\vDash Y_1Y_2$ is an Irred$_2$ affects relation.
\end{example}
\begin{example}{(Irred$_3$ $\not\Rightarrow$ Clus$_3$)}
    \label{ex:3}
Consider the causal model of \cref{ex:1}. We have $X_1\vDash Y$ and $X_1\vDash Y|\mathrm{do}(X_2)$ which implies that $X_1\vDash Y|\mathrm{do}(X_2)$ does not satisfy the Clus$_3$ property. However, $X_1\vDash Y|\mathrm{do}(X_2)$ does satisfy Irred$_3$ since $X_2\vDash Y|\mathrm{do}(X_1)$ and $X_2\vDash Y$ both hold, which violate both conditions of \cref{def: red3} for $s_{X_2}=X_2$ (which is the only possible non-empty $s_{X_2}\subseteq X_2$ in this case).
\end{example}
In \cref{sec:clus-prop}, we will point out some further strong implications of affects relations being clustered. These will not be relevant for our main points on the fine-tuning of causal models, but will be used in some proofs later on. The next connection we show is between clustering and fine-tuning.

\begin{figure}
    \centering
    \begin{subfigure}[b]{0.32\textwidth}
    	\begin{tikzpicture}
    		\begin{scope}[every node/.style={inner sep=8pt}]
    			\node (Irred) at (0,0) {Irred$_i$};
    			\node (Clus) at (2,3) {Clus$_i$};
    			\node (Fine) at (4,0) {Fine-tuning};
    		\end{scope}
            \node (placeholder) at (0,-1.25) {};
    		%\begin{scope}[every edge/.style={double equal sign distance}]
    			\path[-implies,draw] (Clus)
                    edge[double equal sign distance,teal]
                    node[below right,black] {\footnotesize \ref{theorem: clus_irred}}
                    (Irred);
    			\path[-implies,draw] (Irred)
                    edge[strike through, double equal sign distance, bend left,red]
                    node[above left,black] {\footnotesize Ex.~\ref{ex:1}--\ref{ex:3}}
                    (Clus);
    			\path[-implies,draw] (Clus)
                    edge[double equal sign distance,teal]
                    node[below left,black] {\footnotesize \ref{lemma: clus_finetune}}
                    (Fine);
    			\path[-implies,draw] (Fine)
                    edge[strike through, double equal sign distance, bend right,red]
                    node[above right,black] {\footnotesize Ex.~\ref{ex:sorkin}}
                    (Clus);
    			\path[-implies,draw] (Irred)
                    edge[strike through, double equal sign distance,red]
                    node[above=.15cm,black] {\footnotesize Ex.~\ref{ex:1}--\ref{ex:3}}
                    (Fine);
    		%\end{scope}
    	\end{tikzpicture}
        \caption[b]{High-level view of main concepts}
    \end{subfigure}%
    \hspace{0.03\textwidth}
    \begin{subfigure}[b]{0.64\textwidth}
        \centering
        \begin{tikzpicture}
            \begin{scope}[every node/.style={align=center, inner sep=12pt}]
                \node (Irred) at (-0.5,4) {\footnotesize Operational irreducibility\\ \footnotesize of an affects relation};
                \node (colon) at (1.5,4) {:};
                \node (Irred2) at (3,4) {Irred$_2$};
                \node (Irred3) at (5.5,4) {Irred$_3$};
                \node (Irred1) at (7,4) {Irred$_1$};

                \node (Fine) at (-0.5,2) {\footnotesize Absence of affects relations\\ \footnotesize with lower cardinality};
                \node (colonc) at (1.5,2) {:};
                \node (Clus2) at (3,2) {Clus$_2$};
                \node (Clus3) at (5.5,2) {Clus$_3$};
                \node (Clus1) at (7,2) {Clus$_1$};

           \node (Clus3A) at (5.7,2) {\phantom{Clus$_3$}};
            \node (Clus3B) at (5.3,2) {\phantom{Clus$_3$}};
                \node (Fine) at (-0.5,0) {\footnotesize Fine-tuning\\ \footnotesize of causal model};
                \node (supset) at (1.35,0) {$\supsetneq$};
                \node (Obs) at (3,0) {\footnotesize Observational\\ \footnotesize fine-tuning};
                \node (cup) at (4.625,0) {$\cup$};
                \node (Intl) at (5.7,0) {\phantom{\footnotesize Observational}\\ \phantom{\footnotesize fine-tuning}};

                 \node (Intl2) at (5.3,0) {\phantom{\footnotesize Observational}\\ \phantom{\footnotesize fine-tuning}};
                \node (Intr) at (7.0,0) {\phantom{\footnotesize Observational}\\ \phantom{\footnotesize fine-tuning}};
                \node (Int) at (6.25,0) {\footnotesize Interventional\\ \footnotesize fine-tuning};
            \end{scope}
          
            \path[-implies,draw] (Clus1) edge[double equal sign distance,teal] node[left=.05cm,black] {\footnotesize  \ref{theorem: clus_irred}}(Irred1);
            \path[-implies,draw] (Clus2) edge[double equal sign distance,teal] node[left=.05cm,black] {\footnotesize  \ref{theorem: clus_irred}}(Irred2);
            \path[-implies,draw] (Clus3) edge[double equal sign distance,teal] node[left=.05cm,black] {\footnotesize \ref{theorem: clus_irred}}(Irred3);
          
            \path[implies-implies,draw] (Clus2) edge[double equal sign distance,teal]   node[left=.05cm,black] {\scriptsize \textcolor{teal}{ :=}} (Obs);
            \path[implies-,draw] (Clus3A) edge[double equal sign distance,teal] node[right=.05cm,black] {\footnotesize \ref{thm:clusThreeOne}}(Intl);
             \path[-implies,draw] (Clus3B) edge[double equal sign distance,teal] node[left=.05cm,black] {\scriptsize \textcolor{teal}{:=}}(Intl2);
            \path[-implies,draw] (Clus1) edge[double equal sign distance,teal]  node[left=.05cm,black] {\scriptsize \textcolor{teal}{:=}} (Intr);
          
        \end{tikzpicture}
        \caption[b]{Relation between types of clustering and fine-tuning}
    \end{subfigure}
    \caption{Overview of causal inference concepts and relationships given by our results. Generally, an affects relation $X \affects Y \given \doo(Z)$ may contain elements in $X,Y,Z$ which are operationally redundant. While irreducibility (Irred$_i$) signifies that this is not the case (per $i^\text{th}$ argument), clustering (Clus$_i$) additionally certifies that it is impossible to recover signalling by considering a subset of variables (per $i^\text{th}$ argument). The latter is a signature of fine-tuning, which we furthermore use to propose a distinction into different types of fine-tuning.
    In (b), for every one-way implication shown, inequivalence holds.
  Most of the implications between clustering and the specific types of fine-tuning are by definition, as indicated by the annotation $:=$.    
    }
    \label{fig:relations-coarse}
\end{figure}
\begin{restatable}{theorem}{clusFineTuning}[Clustering implies fine-tuning]
    \label{lemma: clus_finetune}
    Any affects relation  $X\vDash Y|\mathrm{do}(Z)$  that satisfies Clus$_i$ for any $i\in\{1,2,3\}$ necessarily arises from a fine-tuned causal model.  
\end{restatable}
\begin{proof}
    See \cref{proof:clusFinetuning}.
\end{proof}

\medskip
Therefore, together with the aforementioned examples this theorem demonstrates that in contrast to clustering, irreducibility applies to both faithful and fine-tuned models. 
While we show that all other statements do, we will see in \cref{sec:affects-conditional} that this theorem does not fully generalize to the case of conditional affects relations (which include the additional argument $W$). There, we prove the generalization under an additional assumption (that $X$ is not a cause of $W$ in the underlying model) and conjecture that the full generalization holds. 

The connection between clustering and fine-tuning established above allows us to differentiate between multiple independent types of fine-tuning.
For instance, we will demonstrate a distinction between \textit{observational} fine-tuning, as indicated by Clus$_2$, and \textit{interventional} fine-tuning, as indicated by Clus$_1$ and Clus$_3$, which could potentially be relevant to characterize information processing tasks.

If it is necessary to intervene on multiple nodes to demonstrate the presence of signalling, we encounter \textit{interventional} fine-tuning, as encoded by Clus$_1$ or Clus$_3$.
In \cref{sec:clus-prop} we show that we can determine the presence of interventional fine-tuning entirely by the presence of affects relations with Clus$_3$.

Regarding \textit{observational} fine-tuning, we may consider the jamming causal model \cite{VVC_Letter, VVCJ, VVCPR} presented in \cref{ex:jamming} where $Y$ jams the correlations between $A$ and $C$, yielding $Y\affects AC$. This model exhibits clustering (only) in the second argument as it has $Y\naffects A$ and $Y\naffects C$.\footnote{Since this is the only affects relation, and it emanates from a single RV while having no third argument, it can be neither Clus$_1$ nor Clus$_3$.} Therefore, in this case we need to observe multiple nodes in the pre- and post-intervention distribution to register the presence of signalling.
In \cite{VVCJ}, it is shown that any causal model where the jamming variable $Y$ (in \cite{VVCJ}, this variable is called $B$) is parentless, and which exhibits jamming must necessarily be fine-tuned, independently of the theory describing hidden common causes.
Our \cref{lemma: clus_finetune} generalizes this previous result, since jamming is characterized by a Clus$_2$ relation (signalling to a set of RVs without signalling to a subset).
Moreover, it also applies to situations beyond Bell-type scenarios\footnote{Note that jamming was originally formulated in a Bell-type scenario, see \cref{ex:jamming}.} where $Y$ is not a freely chosen parentless variable, and to affects relations clustered in any argument. It provides a way to identify the fine-tuning of possible hidden causal parameters in general physical theories, solely from the observable affects relations.

Here, we note that the previous work considered only the original causal model in the definition of fine-tuning, while we also consider post-intervention causal models obtained from the original one. This means that potentially, a larger set of causal models are regarded as fine-tuned according to our definition as compared to the previous work, since one might in-principle have fine-tuning at the level of a post-intervention causal model even when the original pre-intervention model is faithful. Whether this is possible, however, remains an open question. See \cref{def:faithful} for further details on fine-tuning.

Conversely, examples with all three types of clustering also exist: The one-time pad causal model given in \cref{eg: red3_1} involves all three types of clustering. $XZ\affects Y$ is a Clus$_1$ relation (since $X\naffects Y$ and $Z\naffects Y$), we also have $X \affects YZ$ which is Clus$_2$ (since $X\naffects Y$ and $X\naffects Z$) and $X \affects Y|\mathrm{do}(Z)$ is Clus$_3$ (since $X\naffects Y$). Nonetheless, not every type of fine-tuning can be related to clustering in one way or another.

\begin{figure}[t]
	\centering
    \begin{subfigure}[b]{0.28\textwidth}
        \centering
        \begin{tikzpicture}[scale=0.9]
            \begin{scope}[every node/.style={circle,thick,draw,inner sep=0pt,minimum size=0.8cm}]
                \node (X) at (0,1) {$A$};
                \node (Y) at (1.75,1) {$B$};
                \node (Z) at (3.50,1) {$C$};
            \end{scope}
            \begin{scope}[every node/.style={inner sep=0pt,minimum size=0.8cm}]
                \node (P) at (0,0.5) {};
            \end{scope}

            \begin{scope}[>={Stealth[black]},
                          every node/.style={fill=white,circle},
                          every edge/.style=vvarrow]
                \path [->] (X) edge (Y);
                \path [->] (Y) edge (Z);
            \end{scope}
        \end{tikzpicture}
        \caption{Original causal structure $\mathcal{G}$}
        \label{fig:fine-grainining-original}
    \end{subfigure}%
    \hspace{0.04\textwidth}
    \begin{subfigure}[b]{0.28\textwidth}
        \centering
        \begin{tikzpicture}[scale=0.9]
            \begin{scope}[every node/.style={circle,thick,draw,inner sep=0pt,minimum size=0.8cm}]
                \node (X_2) at (0,1) {$A$};
                \node (Y_1) at (1.75,0.45) {$B_1$};
                \node (Y_2) at (1.75,1.55) {$B_2$};
                \node (Z_1) at (3.5,1) {$C$};
            \end{scope}

            \begin{scope}[>={Stealth[black]},
                          every edge/.style=vvarrow]
                \path [->] (X_2) edge (Y_1);
                \path [->] (Y_2) edge (Z_1);
            \end{scope}
        \end{tikzpicture}
        \caption{Fine-grained causal structure $\mathcal{G}'$}
        \label{fig:fine-grainining-after}
    \end{subfigure}%
    \hspace{0.04\textwidth}
    \begin{subfigure}[b]{0.28\textwidth}
        \centering
        \begin{tikzpicture}[scale=0.9]
        \draw [thick, blue, draw opacity=0.6, fill=blue!80!white, opacity=0.5] plot [smooth cycle] coordinates {(-0.5,0) (-0.2,-0.2) (0,-0.2)  (0.5,0) (0.3,0.3) (0.1,0.2) (-0.3,0.3)};
        \node[thick, white, draw opacity=1] at (0,0.05) {$\mathcal{D}_\cA$};
        \draw[thick, blue, draw opacity=0.6] (-0.5,0)--(-1.2,0.7);
        \draw[thick, blue, draw opacity=0.6] (0.5,0)--(2.4,1.9);  
            
        \begin{scope}[shift={(1.3,0.6)}]
            \draw [thick, red, draw opacity=0.6,fill=red!80!white, opacity=0.5] plot [smooth cycle] coordinates {(-0.5,0) (-0.2,-0.2) (0,-0.2)  (0.5,0) (0.3,0.3) (0.1,0.2) (-0.3,0.3)};
            \node[thick, white, draw opacity=0.6] at (0,0.05) {$\mathcal{D}_\cB$};
            \draw[thick, red, draw opacity=0.6] (-0.5,0)--(-1.2,0.7);
            \draw[thick, red, draw opacity=0.6] (0.5,0)--(1.8,1.3);
        \end{scope}
        
        \begin{scope}[shift={(2.6,1.2)}]
            \draw [thick, olive, draw opacity=0.6,fill=olive!80!white, opacity=0.5] plot [smooth cycle] coordinates {(-0.5,0) (-0.2,-0.2) (0,-0.2)  (0.5,0) (0.3,0.3) (0.1,0.2) (-0.3,0.3)};
            \node[thick, white, draw opacity=0.6] at (0,0.05) {$\mathcal{D}_\cC$};
            \draw[thick, olive, draw opacity=0.6] (-0.5,0)--(-1.2,0.7);
            \draw[thick, olive, draw opacity=0.6] (0.5,0)--(1.2,0.7);
        \end{scope}
    \end{tikzpicture}
    \caption{Space-time embedding}
    \label{fig:regions}
    \end{subfigure}
	\caption[Fine-graining of a cyclic causal structure to an acyclic causal structure and its space-time embedding.]{
        Example for a fine-graining of a causal structure, as given in \cref{ex:sorkin}, as well as its space-time embedding that corresponds to the configuration of Sorkin's problem \cite{Sorkin1993}
        In contrast to our approach within \cref{sec:poset}, just for this illustration, we consider an embedding that associates each RV, say $X$, with a space-time region $\mathcal{D}_{\XX}$ which has the relativistic future $\Fut_\text{reg} (\XX):=\{p\in \TT|\exists q\in \mathcal{D}_{\XX}, q\preceq p\}$. We can order regions $\mathcal{D}_{\XX}$ and $\mathcal{D}_{\YY}$ by considering whether the future of $\mathcal{D}_{\XX}$ has an overlap with $\mathcal{D}_{\YY}$ \cite{VVR}.
However, this will not be a partial order, as it lacks transitivity:
        Even though we have $\Fut_\text{reg} (\AA) \cap \mathcal{D}_\BB \neq \emptyset$ and $\Fut_\text{reg} (\BB) \cap \mathcal{D}_\CC \neq \emptyset$, we still have $\Fut_\text{reg} (\AA) \cap \mathcal{D}_\CC = \emptyset$. Note that in the spatio-temporal case as well, analogous to the information-theoretic fine-graining, the ``node'' corresponding to the region $\mathcal{D}_{\BB}$ can be fine-grained into two regions $\mathcal{D}_{\BB}=\mathcal{D}_{\BB_1}\cup \mathcal{D}_{\BB_2}$, where $\mathcal{D}_{\BB_1}$ includes all points in $\mathcal{D}_{\BB}$ that lie the future of $\mathcal{D}_{\AA}$ and $\mathcal{D}_{\BB_2}$ includes the remaining points in $\mathcal{D}_{\BB}$, such that $\Fut_\text{reg} (\BB_1) \cap \mathcal{D}_\CC = \emptyset$. 
        }
	\label{fig:fine-graining}
\end{figure}

\begin{example}{(Non-transitivity: fine-tuning does not imply clustering.)}
    \label{ex:sorkin} Consider a causal model over $S = \{ A, B, C \}$ associated with the causal structure of \cref{fig:fine-grainining-original}, where $A,C\in \{0,1\}$  and $B\in \{00,01,10,11\}$. The model is specified by an arbitrary exogenous distribution $P(A)$ and the conditional distributions $P(B|A)$, $P(C|B)$, where $P(B=i0|A=i)=P(B=i1|A=i)=\frac{1}{2}$ for $i\in \{0,1\}$ (and remaining entries 0), and $P(C=j|B=0j)=P(C=j|B=1j)=1$ for $j\in\{0,1\}$ (and remaining entries 0). Notice that $B$ can be fine-grained into two nodes $B = B_1 \times B_2$ with $B_1,B_2\in \{0,1\}$.
    Then the given causal model can be equivalently expressed as a fine-grained model on the causal structure of \cref{fig:fine-grainining-after} with $B_1 = A$ and $C = B_2$, with $B_2$ uniformly distributed and $A$ distributed as $P(A)$.
    Then, relative to the coarse-grained scenario of \cref{fig:fine-grainining-original}, we have $A \not\perp^d C$, but $A \indep C$, and hence a fine-tuned\footnote{
        However, it satisfies a relaxed causal faithfulness assumption, known as \emph{adjacency faithfulness} \cite{Zhang2008, Zhalama2017}.
    } causal model. 
    Accordingly, we have $A \affects B$ and $B \affects C$, but $A \naffects C$ (and $A \naffects C \given \doo(B)$).
    By doing so, we use a similar setup to \cite{VVR}, where they fine-grain causal structures to unravel cyclic causal structures into acyclic ones. Note that here, the fine-graining removes the fine-tuning of the model.
    
    Furthermore, this example bears some resemblance with the setup of the Sorkin problem \cite{Sorkin1993, PhysRevD.103.025017} in algebraic QFT, where signalling can not be transitive, as otherwise superluminal signalling would be possible.
    In this scenario, RVs are associated to regions rather than points of space-time, with access to the entire region generally required to extract the entire information associated with a RV. 
\end{example}

\begin{example}{(Lack of full support: fine-tuning does not imply clustering. \cite{arxiv.1906.10726})}
    Consider a causal model over two binary RVs $Y$ and $Z$, with $Y \longrsquigarrow Z$, $Z = Y$, and $Y$ deterministically distributed, e.g. $P(Y=0) = 1$ and $P(Y=1) = 0.$
    Then, we have $Y \not\perp^d Z$, yet $Y \indep Z$, and $P_\cG (X,Y)$ is fine-tuned.
    However, the only affects relation in this model, $Y \affects Z$, is not clustered.
    Note that choosing any distribution with full support for $Y$, realizing each outcome with non-zero probability, will remove the fine-tuning of the model.
\end{example}

This yields the aforementioned remaining type of fine-tuning, which is not captured by interventional or observational clustering of affects relations.
In the first example, the fine-tuning is detectable from the affects relations alone, as it is certified through their non-transitivity\footnote{Specifically, the non-transitivity of affects relations in our example certifies fine-tuning because $A\affects B$ and $B\affects C$ tells us that $A$ is a cause of $C$ (\cref{cor: causal_inference}), which gives $A\not\perp^d C$. However,  $A\naffects C$ here implies $A\indep C$ (\cref{lemma: aff_corr}), which implies fine-tuning (\cref{eq:fine-tuning}).}, Here, this is due to the causal model being fine-grainable.
By contrast, in the second example, the set of affects relations alone does not indicate that the causal model is fine-tuned.

For pre-intervention distributions with full support, do clustering and fine-grainability cover all operationally detectable fine-tunings? We leave these questions for future work. 

\subsection{Causal inference using irreducibility and clustering}
\label{sec:affects-to-cause}

Using irreducibility, one can obtain stronger causal inference results than given in \cref{thm:affects-to-cause-first}. For Irred$_1$ the following lemma was shown in \cite{VVC}. Here we use \cref{def:cause} of \emph{cause}.
\begin{lemma}[{\cite{VVC}}]
    \label{thm:affects-to-cause}
    For any disjoint sets $X$, $Y$ and $Z$ of observed nodes $X\vDash Y|\mathrm{do}(Z)$ is Irred$_1$ $\Rightarrow $ each $e_X\in X$ is a cause of at least one element $e_Y\in Y$.
\end{lemma}

Analogously, we obtain the following results for Irred$_3$.
\begin{restatable}{lemma}{causalThree}
    \label{lemma: causal_inference_irred3}
    For any disjoint sets $X$, $Y$ and $Z$ of observed nodes, $X\vDash Y|\mathrm{do}(Z)$ is Irred$_3$ $\Rightarrow $ each $e_Z\in Z$ is a cause of at least one element $e_Y\in Y$.
\end{restatable}
\begin{proof}
    $X\vDash Y|\mathrm{do}(Z)$ satisfies Irred$_3$ implies in particular that for each $e_Z\in Z$, either $e_Z\vDash Y|\mathrm{do}(XZ\backslash e_Z)$ or $e_Z\vDash Y|\mathrm{do}(Z\backslash e_Z)$. As $e_Z$ is a singleton, each of these implies by \cref{thm:affects-to-cause} that $e_Z$ is a cause of at least one $e_Y\in Y$, which establishes the claim.
\end{proof}

\medskip
Therefore, Irred$_1$ and Irred$_3$ complement each other, allowing to perform causal inference for individual nodes from all nodes that are intervened upon.
For affects relations which are irreducible in both arguments we can summarize this to a single expression.

\begin{corollary}
    \label{cor: causal_inference}
    \label{thm:irr-indec}
    For any disjoint sets $X$, $Y$ and $Z$ of observed nodes, if  $X\vDash Y|\mathrm{do}(Z)$ satisfies Irred$_1$ and Irred$_3$ then $\forall e_{XZ}\in XZ$, $\exists e_{Y}\in Y$ such that $e_{XZ}$ is a cause of $e_{Y}$.   
\end{corollary}

Combining this and \cref{theorem: clus_irred} we obtain the following corollaries. 
\begin{corollary}
    \label{cor: causal_inference_clus1}
    For any disjoint sets $X$, $Y$ and $Z$ of observed nodes, $X\vDash Y|\mathrm{do}(Z)$ is Clus$_1$ implies that each $e_X \in X$ is a cause of at least one element $e_Y\in Y$.
\end{corollary}

\begin{corollary}
    \label{cor: causal_inference_clus3}
    For any disjoint sets $X$, $Y$ and $Z$ of observed nodes, $X\vDash Y|\mathrm{do}(Z)$ is Clus$_3$ implies that each $e_Z \in Z$ is a cause of at least one element $e_Y\in Y$.
\end{corollary}

As clustering properties necessitate the absence of certain affects relations and also indicate fine-tuning (\cref{lemma: clus_finetune}), this yields the interesting result that in certain fine-tuned causal models, we can use the \textit{absence} of signalling between nodes for successful causal inference.

Furthermore, this raises the question whether we can use Irred$_2$ in a similar fashion to pinpoint individual nodes $e_Y \in Y$ to be caused by an affects relation $X \affects Y \given \doo(Z)$.
However, the jamming scenario given in \cref{ex:jamming} shows that this is not the case: There, $Y \affects AC$ and is Irred$_2$ (and Clus$_2$), yet $Y$ is not a cause of $A$. However, $A$ and $C$ share a common cause distinct from $Y$.
Therefore, for an Irred$_2$ affects relation $X \affects Y \given \doo(Z)$, while we cannot infer that every element of $Y$ is an effect of some element of $X$, we can infer some $d$-connections between these elements, as shown below.

\begin{restatable}{lemma}{causalTwo}
    \label{lemma: causal_inference_irred2}
    For any disjoint sets $X$, $Y$ and $Z$ of observed nodes, $X\vDash Y|\mathrm{do}(Z)$ is Irred$_2$ implies that both $(s_Y \not\perp^d \tilde{s}_Y|Z)_{\cG_{\mathrm{do}(Z)}}$ and $(X \not\perp^d s_Y | Z \tilde{s}_Y)_{\cG_{\mathrm{do}(X (\tilde{s}_Y) Z)}}$ for any partition of $Y$ into subsets $s_Y, \tilde{s}_Y$.
    Here, $X(W):=X\backslash (X \cap \mathrm{anc}(W))$ for $\mathrm{anc}(W)$ denoting the set of all ancestors of $W$ in the original graph $\cG$.
\end{restatable}
\begin{proof}
    If $X\vDash Y|\mathrm{do}(Z)$ is Irred$_2$, then $X \vDash s_Y | \mathrm{do}(Z), \tilde{s}_Y$ must hold and
    either $(s_Y \not\upmodels \tilde{s}_Y|XZ)_{\cG_{\mathrm{do}(XZ)}}$
    or $(s_Y \not\upmodels \tilde{s}_Y|Z)_{\cG_{\mathrm{do}(Z)}}$ must hold,
    for each partition of $Y=s_Y\tilde{s}_Y$. Using the d-separation property, this implies that either $(s_Y \not\perp^d \tilde{s}_Y|XZ)_{\cG_{\mathrm{do}(XZ)}}$ or  $(s_Y \not\perp^d \tilde{s}_Y|Z)_{\cG_{\mathrm{do}(Z)}}$ must hold.
    Noting that in $\cG_{\mathrm{do}(XZ)}$, $X$ only consist of parentless nodes, it cannot act as a collider and therefore removing it from the conditioning set cannot remove d-connection. Moreover, in going from $\cG_{\mathrm{do}(XZ)}$ to $\cG_{\mathrm{do}(XZ)}$ we would introduce additional edges, which cannot remove d-connection.
    Thus in both cases we have $(s_Y \not\perp^d \tilde{s}_Y|Z)_{\cG_{\mathrm{do}(Z)}}$.
    The first condition, i.e.
    $X \vDash s_Y | \mathrm{do}(Z), \tilde{s}_Y$ further yields
    by \cref{lemma: dsep_aff_cond}, 
    $(X \not\perp^d s_Y | Z \tilde{s}_Y)_{\cG_{\mathrm{do}(X(\tilde{s}_Y) Z)}}$.
\end{proof}

\section{Space-time as a partial order}

\label{sec:poset}
To model the causal properties of space-time in this framework, we aim to be as general as possible. Therefore, we will model the causal structure of space-time as a partially ordered set (poset) $\TT$ of its points, as suggested by \cite{Kronheimer1967}.
Thereby, we match the fundamental approach taken in \cite{VVC}, albeit analyzing the properties of the respective poset in more detail.

\begin{figure}[t]
    \centering
    \tikzset{surface/.style={draw=blue!70!black, fill=blue!20!white, fill opacity=.6}}

    \newcommand{\coneback}[4][]{
    \draw[canvas is xy plane at z=#2, #1] (45-#4:#3) arc (45-#4:225+#4:#3) -- (O) --cycle;
    }
    \newcommand{\conefront}[4][]{
    \draw[canvas is xy plane at z=#2, #1] (45-#4:#3) arc (45-#4:-135+#4:#3) -- (O) --cycle;
    }
    \scalebox{0.9}{
    \begin{tikzpicture}[tdplot_main_coords, grid/.style={help lines,violet!40!white,opacity=0.5},scale=1]
        \coordinate (O) at (0,0,0);

        \coneback[surface]{-3}{2}{-12}
        \conefront[surface]{-3}{2}{-12}

        \fill[violet!40!white,opacity=0.5] (-4,-4,0) -- (-4,4,0) -- (4,4,0) -- (4,-4,0) -- cycle;

        \foreach \x in {-4,...,4}
        \foreach \y in {-4,...,4}
        {
         \draw[grid] (\x,-4) -- (\x,4);
         \draw[grid] (-4,\y) -- (4,\y);
         \draw[violet] (-4,4)--(-4,-4)--(4,-4)--(4,4)--cycle;
        }

        \draw[->] (-4,0,0) -- (4,0,0) {};
        \draw[->] (0,-4,0) -- (0,4,0) {};
        \coneback[surface]{3}{2}{12}
        \draw[-,dashed] (0,0,-2.65) -- (0,0,2.65) node[above] {};
        \draw[-,dashed] (0,0,-4) -- (0,0,-3.35) node[above] {};
        \draw[->,dashed] (0,0,3.35) -- (0,0,4) node[above] {$time$};
        \conefront[surface]{3}{2}{12}
        \fill (4,0,2) circle (2pt) node[above right] {$c$};
        \fill (0,0,0) circle (2pt) {};
        \fill (-0.5,-0.85,2.2) circle (2pt) node[above left] {$a$};
        \fill (1.3,0.5,2) circle (2pt) node[above left] {$b$};
        \draw[->,red] (0,0,0) -- (4,0,2) node[below, pos=0.6, rotate=26.5651,scale=0.70,black] {$\textbf{space-like vector}$};
        \draw[->,red] (0,0,0) -- (1.3,0.5,2) node[below, pos=0.65, rotate=55.1459,scale=0.70,black] {$\textbf{light-like vector}$};
        \draw[->,red] (0,0,0) -- (-0.5,-0.85,2.2) node[above, pos=0.57, rotate=-65.8557,scale=0.70,black] {$\textbf{time-like vector}$};
        \node[black] at (0,0,3) {$Future\,\,Light\,\,Cone$};
        \node[black] at (0,0,-3) {$Past\,\,Light\,\,Cone$};
        \node[black] at (0,0.05,0.3) {$o$};
        \node[black] at (0,4.7,0) {$space$};
        \node[black] at (5,-0.3,0) {$space$};
    \end{tikzpicture}
    }
	\caption[A light cone in 2+1-Minkowski space-time with space-like and time-like region.]{
        A light cone in 2+1-Minkowski space-time. For each point $o$, there exist time-like separated points $a$, in this case with $a \succ o$, space-like separated points $c$ with $c \unord o$, which are unordered with regard to $o$, and light-like separated points, which are again ordered with respect to $o$: Here, $b \succ o$.
        Therefore, relative to $o$, there exist a region of space-like separated points and two regions of time-like separated points, which are separated by a surface of codimension 1 of points which are separated in a light-like way.
        Adapted from picture published by SandyG and Nick at \href{https://tex.stackexchange.com/questions/640441/}{TeX Stack Exchange} and licensed under \href{https://creativecommons.org/licenses/by-sa/4.0/}{CC BY-SA 4.0}.}
	\label{fig:minkowski}
\end{figure}

Modelling space-time as a poset allows us to study its causal structure in order-theoretic terms, while setting aside the additional mathematical structure which is usually assumed when studying space-time, such as the differential manifold structure and symmetries.
However, Malament's theorem \cite{Malament1977} shows that most of this structure, including topological properties, can be captured using partial order relations alone. The partial order relations between space-time points $a, b \in \TT$ can be one of the following:

\begin{equation}
    a = b \, , \quad
    a \prec b \, , \quad
    a \succ b \, , \quad
    a \unord b \, .
\end{equation}
These correspond respectively to $a$ and $b$ being identical, $a$ being in the causal past/future of $b$ (and therefore time-like or light-like relative to each other), or causally unordered (and therefore space-like with regard to each other).

By allowing to order events in space-time in a transitive way, this approach is generic enough to model causality for an arbitrary space-time manifold without closed time-like curves (CTC), an example being Minkowski space-time depicted in \cref{fig:minkowski}.
Allowing for CTCs would allow for orderings of the form $b \prec a \prec b$, which would be incompatible with a partial order.
Such space-times would be modelled as a \emph{preorder} instead.
However, this generality is not of interest here as it allows the causal future and the causal past of a point to be non-disjoint, and would thereby trivially allow the embedding of arbitrary causal models, as we will appreciate later on.

Further, it allows to study discrete generalizations of space-time \cite{PhysRevD.87.064022}, since we make no further assumptions on the properties of $\TT$. Here, the most notable representative of the latter group are \emph{causal sets} (causets) \cite{PhysRevLett.59.521, Surya2019}, which are locally finite in addition to being posets.

The following definition captures the notion of causal future, matching literature conventions of the field of general relativity.
There, it is commonly used to study causal structures on space-time manifolds \cite{Penrose1972}.
\begin{definition}[Causal Future and Causal Past]
    \label{def:futurej}
    Let $\TT$ be a poset and $a \in \TT$. Then the \emph{exclusive causal future} of $a$ is given by $J^+ (a) := \{ b \in \TT : b \succ a \}$, while the \emph{(inclusive) causal future} is given by $\bar{J}^+ (a) := \{ b \in \TT : b \succeq a \}$.
    Dually, the \emph{(inclusive) causal past} of $a$ is given by $\bar{J}^- (a) :=  \{ b \in \TT : b \preceq a \}$.
\end{definition}
Deviating slightly from usual conventions, we will refer to the \emph{inclusive} causal future whenever we just refer to the causal future for short.
Of course, these definitions apply equally well to posets $\TT$ which do not form a manifold.
Informally, we will nonetheless generally refer to the causal future as a light cone.

While there is a wide variety of well-studied properties a poset may implement, they will not be of vital importance for the arguments of this work.
We give a brief formal overview of posets and some of their properties in \cref{sec:poset-props}, where we will also relate them to Minkowski space-time.

\subsection{An order-theoretic property of physical space-times: conicality}
\label{sec:conicality}
In this section, we introduce the order-theoretic property of conicality, which, to the best of our knowledge, has not been studied so far.
This property is in particular satisfied by higher-dimensional Minkowski space-time, as we will show.
As a prerequisite to express conicality, we introduce a notion of \textit{spanning elements} for subsets of the poset. 
For each set of points, these denote the subset that is contributing to the shape of their joint causal future.
\begin{figure}
\centering
\begin{subfigure}[b]{0.6\textwidth}
	\centering
	\begin{tikzpicture}
		\begin{scope}
			\clip (3.5,0) circle (2.5);
			\fill[fill=blue!20] (0,0) circle (2.5);
		\end{scope}
		
		\draw[fill=none](0.00,0) circle (2.5);
		\draw[fill=none](1.75,0) circle (2.5);
		\draw[fill=none](3.50,0) circle (2.5);
		
		\draw[fill=black](0.00,0) circle (1pt) node [above] {\small $a$};
		\draw[fill=black](1.75,0) circle (1pt) node [above] {\small $b$};
		\draw[fill=black](3.50,0) circle (1pt) node [above] {\small $c$};
	\end{tikzpicture}
	\caption{Temporal slice}
    \label{fig:late-2}
\end{subfigure}%

\bigskip
\begin{subfigure}[b]{0.6\textwidth}
	\centering
	\begin{tikzpicture}[dot/.style={circle,inner sep=1pt,fill,name=#1}]
		\node [dot=A,label=$a$] at (0,0) {};
		\node [dot=B,label=$b$] at (1.75,0) {};
		\node [dot=C,label=$c$] at (3.5,0) {};
		
		\fill[fill=blue!20] (1.75,1.75) -- (2.5,2.5) -- (1,2.5) -- (1.75,1.75);
		
		\draw (A) -- ++(-2.5,2.5);
		\draw (A) -- ++(2.5,2.5);
		\draw (B) -- ++(-2.5,2.5);
		\draw (B) -- ++(2.5,2.5);
		\draw (C) -- ++(-2.5,2.5);
		\draw (C) -- ++(2.5,2.5);
	\end{tikzpicture}
	\caption{Spatio-temporal slice}
    \label{fig:late-1}
\end{subfigure}%
\caption{Light cones originating from three points located on a space-like line in Minkowski space-time with $d = 2$ spatial dimensions.
Even though $a, b, c \in \TT$ are unordered with regard to one another and therefore
$\{ a, b, c \} = \late(\{ a, b, c \})$, the point $b$ does not contribute to the shape of their joint future. We will denote this fact as $b \protect\not\in \protect\spann(\{ a, b, c \}) = \{ a, c \}$.
This property holds even if we move $b$ earlier in time.
However, it would fail if we move it later in time, as the joint future of $a$ and $c$ would no longer be contained in the future of $b$ in that case (for $d=2$ spatial dimensions) \cite{VVCJ}.}

\label{fig:late}
\end{figure}

\begin{definition}[Spanning Points]
    \label{def:span}
    Let $\TT$ be a poset and $L \subset \TT$ be finite. Then, the set of \emph{spanning elements}, denoted by $\spann(L)$, is given by the union of all sets $L' \subseteq L$ that satisfy
    \begin{equation}
        f(L') = f(L) \quad \text{and} \quad
        \not \exists L'' \subsetneq L' \colon  f(L'') = f(L)
    \end{equation}
    where $f(L) := \bigcap_{x \in L} \bar{J}^+ (x)$.
\end{definition}

Hence it forms the union of all possible minimal subsets of $L$ which share the same joint future light cone as $L$.
Generally, $\spann(L)$ is different from the set of \enquote{latest} elements of $L$ in the poset, given by
\begin{equation}
    \late(L) := \{ x \in L \,|\, \not\exists y \in L \ \colon \ y \succ x \} \, .
\end{equation}
For Minkowski space-time in any dimension, this is illustrated in \cref{fig:late} for the example of three equidistant points $a, b, c \in \TT$, located on a space-like line next to each other.\footnote{
    Equidistance is assumed to ensure $\bar{J}^+ (b) \subseteq \bar{J}^+ (a) \cap \bar{J}^+ (c)$.
    Alternatively, positioning $b$ earlier than the space-like line connecting $a$ and $c$, yet space-like to the individual points, relaxes the equidistance requirement.}
\begin{equation}
    \bar{J}^+ (a) \cap \bar{J}^+ (b) \cap \bar{J}^+ (c) = \bar{J}^+ (a) \cap \bar{J}^+ (c).
\end{equation}
Therefore, in this case, $\spann(\{a,b,c\}) = \{a,c\}\neq \late(\{a,b,c\}) = \{a,b,c\}$.
Vice-versa, for the example shown in \cref{fig:d-1}, we have $\spann(\{x,a\}) = \{x,a\}$, while  $\late(\{x,a\}) = \{x\}$.

\begin{definition}[Conicality]
    \label{def:conical}
    Let $\TT$ be a poset. We say that $\TT$ is a \emph{conical} poset if for any two finite subsets $L_i, L_j \subset \TT$,
    \begin{equation}
        \label{eq:conical}
        f(L_i) = f(L_j)
        \quad \implies \quad 
        \spann(L_i) = \spann(L_j) 
    \end{equation}
    holds. 
\end{definition}

\begin{figure}
    \centering
    \begin{subfigure}[b]{0.4\textwidth}
        \centering
		\begin{tikzpicture}[dot/.style={circle,inner sep=1pt,fill,name=#1}]

			\node [dot=A,label=$a$] at (1,1) {};
			\node [dot=X,label=$x$] at (2,2) {};
			\node [dot=Z,label=$y$] at (3.7,1.8) {};
			\node [dot=B,label=$b$] at (4.5,1) {};

			\node (left) at (0.5,0.5) {};
			\node (right) at (5,0.5) {};
			\draw (left) -- ++(3.5,3.5);
			\draw (right) -- ++(-3.5,3.5);
			\fill[fill=blue!20] (2.75,2.75) -- (4,4) -- (1.5,4) -- (2.75,2.75);
		\end{tikzpicture}
        \caption{d=1}
        \label{fig:d-1}
    \end{subfigure}%
    \begin{subfigure}[b]{0.6\textwidth}
        \centering
        \begin{tikzpicture}
            \begin{scope}
                \clip (0.0,0) circle (2.5);
                \fill[fill=red!20] (3.0,0) circle (2.5);
            \end{scope}

            \begin{scope}
                \clip (0.75,0) circle (1.75);
                \fill[fill=blue!20] (2.25,0) circle (1.75);
            \end{scope}

            \draw[fill=none](0.75,0) circle (1.75);
            \draw[fill=none](0.00,0) circle (2.50);
            \draw[fill=none](2.25,0) circle (1.75);
            \draw[fill=none](3.00,0) circle (2.50);

            \draw[fill=black](0.75,0) circle (1pt) node [above] {\small $x$};
            \draw[fill=black](0.00,0) circle (1pt) node [above] {\small $a$};
            \draw[fill=black](2.25,0) circle (1pt) node [above] {\small $y$};
            \draw[fill=black](3.00,0) circle (1pt) node [above] {\small $b$};
        \end{tikzpicture}
        \caption{d=2}
        \label{fig:d-2}
	\end{subfigure}%
    \caption{Schematic representation of light cones in Minkowski space-time for $d$ spatial dimensions. (a) For $d=1$ Minkowski space-time does not show conicality. This is because $L_1:=\{a,b\}$ and $L_2:=\{x,y\}$ are distinct sets where both satisfy $L_i=\spann(L_i)$ but    $f(L_1)=\highlight{red!20}{\bar{J}^+ (a) \cap \bar{J}^+ (b)} = \highlight{blue!20}{\bar{J}^+ (x) \cap \bar{J}^+ (y)}=f(L_2)$. (b) For $d=2$, Minkowski space-time shows conicality (cf.\ \cref{def:conical}). As can be seen from the figure which represents one particular time slice, the joint futures are distinct between $L_1$ and $L_2$.}
	\label{fig:d}
\end{figure}

Hence, in conical posets, it is possible to recover the origin of each light cone contributing to the shape of the joint future from that shape alone.
In particular, conicality ensures that the set $L' (=\spann(L'))$ featuring in \cref{def:span} is unique,
guaranteeing $L' = \spann(L)$ and eliminating the necessity to consider a union of distinct such sets.
We now consider whether this property is satisfied for a relevant physical models for space-time.

\begin{restatable}{lemma}{conicalMinkowski}
    \label{thm:conical}
    For $d \geq 2$, $d$+1-Minkowski space-time is a conical poset.
\end{restatable}
\begin{proof}
    See \cref{proof:conicalMinkowski}.
\end{proof}

\medskip
\Cref{fig:d} illustrates the difference of the geometric properties of Minkowski space-time in 1+1 dimensions, where it does not satisfy conicality, in contrast to the higher-dimensional case. Finally, we will conjecture that this property generalizes to other space-times than Minkowski. Therefore, all results we show which require conicality would generalize to this class of space-times (if the conjecture holds true).

\begin{conjecture}
    Any $d$+1-dimensional space-time manifold with $d \geq 2$, time orientation and no CTCs which is homotopic to Minkowski space-time is conical.
\end{conjecture}

Here, we refer to \emph{closed time-like curves as used in general relativity}, which we have ruled out by choosing a partial order model for space-time in \cref{sec:poset}. This is distinct from \emph{causal loops in an information-theoretic causal model}. Even when embedding a causal model in partially ordered space-times and forbidding signalling outside the future light cone, the arrows of the causal model need not align with the direction of time \cite{VVC_Letter}.
Thus, we will not refer to cyclic (information-theoretic) causal structures as featuring CTCs, but refer to them as \emph{causal loops}.

The conjecture is based on the observation that conicality only refers to qualitative features of the light cone structure of $d$+1-Minkowski space-time, which is locally upheld by arbitrary space-time manifolds.
However, it does not necessarily carry over if the space-time is finite, is not simply-connected (e.g.\ has singularities or light cones interfering with themselves), as light cones can change their fundamental geometry in such situations. Further, we presume it plausible that this conjecture may require additional technical conditions on the space-time, as for example given by global hyperbolicity \cite{Wald1984}. 
\begin{remark}
    Consider Newtonian space-time, composite of absolute space and totally ordered time, as a poset $\TT_\text{abs}$.\footnote{
        In this poset, points are ordered by their time label, with distinct points of the same time label being unordered (as they are neither identical nor admit a well-defined order).}
    In fact, this poset is \underline{not} conical.
    Consider three simultaneous points $x \unord y \unord z$.
    Then,
    \begin{equation}
        f(\{ x, y \}) = \{ a \in \TT | a \succeq x \} \cap \{ a \in \TT | a \succeq y \}
        = \{ a \in \TT | a \succ x \} = \{ a \in \TT | a \succ y \} \, .
    \end{equation}
    Accordingly, for any $L \subset \TT_\text{abs}$ finite, while $\spann(L) = \late(L)$, we have
    $f(\{ x, y \}) = f(\{ y, z \}) = f(\{ x, z \})$
    and hence the poset is not conical.\footnote{Note that e.g.\ $x \not\in f(\{ x, y \})$, yet $x \in f(\{ x \})$.}
    By contrast, any totally ordered set -- e.g.\ a set of absolute timestamps -- is conical.\footnote{Furthermore, each totally ordered set is an order lattice (cf.\ \cref{def:lattice}), as only trivial joins and meets exists.}
\end{remark}
Moreover, there can be finite posets (not necessarily coming from an underlying manifold structure) which satisfy the conicality property. Further research is needed to determine and characterize the full space of posets satisfying these properties.

\subsection{Ordered random variables} \label{sec:orvs}

To gain a link between the set of RVs $S$ of a causal model as outlined in \cref{sec:causal-model}, and a space-time given by a poset $\mathcal{T}$, in which the respective physical experiments are ultimately performed, the concept of a space-time embedding $\mathcal{E}$ of a causal model was introduced in \cite{VVC}. Each observed random variable in the model will be embedded into a single location in space-time.
In doing so, we will introduce the formal framework to describe an embedding independent of any desired compatibility conditions with the causal model, which will follow in \cref{sec:compat-top}.
While reviewing \cite{VVC}, we will however perform some simplifications and also highlight certain properties of sets of ORVs more explicitly.

\begin{definition}[Ordered Random Variable (ORV)]
    \label{def:orv}
    An \emph{ORV} is defined as the pair $\XX := (X, O(X))$, where $X$ is a random variable (RV) and $O(X) \in \TT$ its assigned location in the partial order $\TT$.
    Here, $O$ is referred to as an \emph{Ordering}.
    Given a set of RVs $S$, the associated set of ORVs is denoted by $\SC$.
\end{definition}

We then have the following definition of a space-time embedding, adapted from \cite{VVC}.

\begin{definition}[Embedding]
    \label{def:embedding}
    An \emph{embedding} $\mathcal{E}$ specifies an ordering $O$ on a set of RVs $S$ in $\TT$, yielding a set of ORVs $\SC = (S, O(S))$.
    \begin{align}
        \mathcal{E}&: S \mapsto \SC :=  \{ (X,O(X)) | X \in S \}
    \end{align}
    If $O$ is injective, its embedding is considered \emph{non-degenerate} ($O(X) \neq O(Y)$ for all distinct $X,Y \in S$).
\end{definition}

\begin{notation}
    We will carry over the notion of a partial order from $\TT$ to $\SC$. Hence, $O(X) \prec O(Y)$ will be expressed as $\XX \prec \YY$. By abuse of notation, we will often denote $O(\XX) := O(X)$.
\end{notation}

\begin{notation}
    Going forward, we will also use the notation $O (S) :=\{O(X): X\in S \}$ for sets $S$ of RVs. Thereby, we obtain a set of ORV locations.
    By abuse of notation, we will further denote $O(\mathcal{S}) := O (S)$.
\end{notation}

\begin{notation}
    From now on, let the (sets of) ORVs $\AA, \BB, \CC, \XX, \YY, \ZZ, \SC_i, e_\XX$ etc.\ always be associated to the (sets of) RVs $A, B, C, X, Y, Z, S_i, e_X$ etc.\ with respect to some embedding $\mathcal{E}$, and vice versa.
\end{notation}

We carry over the notion of causal future $\bar{J}^+ (x)$, introduced for points $x \in \TT$ within \cref{def:futurej}, to ORVs.

\begin{definition}[Future]
    \label{def:future}
    The \emph{(inclusive) future} of an ORV $\XX$ are defined as
    \begin{equation}
        \Fut (\XX) := \{ a \in \TT : a \succeq O(\XX) \} \, .\footnote{
            Additionally, \cite{VVC} introduced a notion of \textit{exclusive} future, given by $\mathcal{F} (\XX) := \{ a \in \TT : a \succ O(\XX) \}$. We will not use it in this work.
        }
    \end{equation}
\end{definition}

Connecting this to the standard notion of causal future introduced in \cref{def:futurej}, we can equivalently write
$\Fut (\XX) := \bar{J}^+ (O(\XX))$.

Having assigned a location to each RV in $S$ to obtain an ORV, this induces a notion of joint future for sets of ORVs \cite{VVC}. We explicitly distinguish the future of a single ORV from the joint future of a set of more than one ORVs, by referring to the latter as the \emph{support future}.

\begin{definition}[Support Future]
    \label{def:supp-future}
    Let $\XX$ be a subset of a set of ORVs $\SC$. We define
    \begin{equation}
        \Fut_s (\XX) \equiv \Fut_s \left( \bigcup_{\XX_i \in \XX} \XX_i \right) := \bigcap_{\XX_i \in \XX} \Fut(\XX_i)
    \end{equation}
    and call it the \emph{support future}.
\end{definition}

Hence, future and support future coincide for individual ORVs $e_\XX \in \SC$, i.e. $\Fut_s (e_\XX) = \Fut (e_\XX)$.
Note further that associating the joint future of a set of locations $L \subset \TT$ to be the \emph{intersection} of the futures of the individual locations is not the usual convention in general relativity. There, the causal future of a set of points is usually understood as the \emph{union} $\bar{J}^+ [L] = \bigcup_{x \in L} \bar{J}^+ (x)$ \cite{Penrose1972, Wald1984, Minguzzi2019}.
The convention adopted here (based on \cite{VVC, VVC_Letter}) is relevant for tightly characterizing the principle of no signalling outside the future light cone, where the relevant aspect to consider is where can a set of variables be jointly accessed (see \cref{sec:compat}).

\begin{remark}
    By this definition, $\Fut_s (\emptyset) = \TT$.
    This follows as for subsets of $\TT$, intersection with $\TT$ is the identity operation.
    Specifically, for any set $\XX$ of ORVs, $\Fut_s(\XX)=\Fut_s(\XX\cup \emptyset)=\Fut_s(\XX)\cap \Fut_s(\emptyset)$.
\end{remark}

We conclude by carrying over the notions of spanning elements and conicality, which were originally formulated for sets of points in a poset $\TT$, to sets of RVs embedded in a poset (ORVs).
In doing so, we gain additional complexity due to multiple ORVs $\XX, \YY \in \SC$ potentially sharing the same location $O(\XX) = O(\YY)$ in $\TT$.

\begin{definition}[Spanning Elements for ORVs]
    For a set of ORVs $\XX$, we define its set of \emph{spanning elements}, denoted by $\spann(\XX)$, as the union of all sets $s_\XX \subseteq \XX$ which satisfy
    \begin{equation}
        \label{eq:span-orv}
        \Fut_s (s_\XX) = \Fut_s (\XX) \quad \text{and} \quad
        \not \exists t_\XX \subsetneq s_\XX : \Fut_s (t_\XX) = \Fut_s(\XX) \, .
    \end{equation}
\end{definition}

Indeed, we obtain:

\begin{restatable}{lemma}{conicalORV}
    \label{def:conical-orv}
    Let $\TT$ be a conical poset and $L \subset \TT$ be finite. Let $\XX$ be a set of ORVs on this poset. 
    Then, the knowledge of $\Fut_s (\XX)$ implies the locations $O (\XX_i)$ for all its spanning elements $\XX_i \in \spann (\XX)$.
\end{restatable}
\begin{proof}
    A proof of a generalized version of this lemma is given at \cref{def:conical-embedding-orv}.
\end{proof}

\medskip
Having established these notions for ORVs, we can derive a more technical equivalent formulation of conicality that we will refer to as \emph{location symmetry}, as is detailed in \cref{sec:location-symmetry}.
This formulation will be a vital ingredient in the proofs of our main results in \cref{sec:compat-indec}.
In the particularly simple case of $\XX \subset \SC, \YY_1, \YY_2 \in \SC$, it can be written as
\begin{equation}
    \label{eq:location-symmetry-simple}
	\Fut_s (\XX\YY_1) = \Fut_s (\XX\YY_2)
	\quad \implies \quad
	\Fut_s (\XX) \subseteq \Fut_s (\YY_1 \YY_2) \quad \text{or} \quad O(\YY_1) = O(\YY_2) \, .
\end{equation}
Generally, location symmetry captures the intuition that if two non-disjoint sets of ORVs (with $\XX$ being their overlap) share the same joint future, this is either due to the joint future of the overlapping ORVs $\XX$ being contained in joint futures of ORVs $\YY_1$ and $\YY_2$ outside the overlap, or due to the latter ORVs being embedded at the same location (i.e., the embedding is degenerate).

Beyond that, it is sensible to lift the notion of \textit{conicality} from the poset to the level of the embedding (cf.\ \cref{def:conical-embedding}), as primarily the properties of ORVs relative to one another are of relevance.
However, we will defer these generalizations to the respective proofs in \cref{sec:proofs}.

\begin{remark}
	We note that in \cite{VVC}, a distinct notion of \emph{accessible regions} $\RC_X\subseteq \mathcal{T}$ for each RV $X \in S$ was also introduced, which is the space-time region where the information associated with the respective RV is accessible. Accordingly, \emph{a priori} \cite{VVC} considers this region to be independent of the location or the future of the ORV.
    This allows to distinguish the information-theoretic aspects of accessibility from the partial order of $\TT$. Nonetheless, when imposing compatibility which captures relativistic principles in the space-time, $\RC_X = \Fut(\XX)$ is imposed, identifying the accessible regions with the causal future \emph{a posteriori}. Thus $\RC_X$ is then determined by the location $O(X)$ of the RV. Operationally, this identification allows for the broadcasting of classical information everywhere in the future light cone. For ease of presentation, we have dropped the distinction between the accessible region and the spatio-temporal future in this work (thereby not needing to introduce the former separately).
\end{remark}
\section{Compatibility of affects relations with space-time}
\label{sec:compat-top}

In this section, we will review compatibility conditions for embeddings of RVs into a space-time in the presence of affects relations as introduced in \cite{VVC}.
Thereby, we relate the information-theoretic concept of signalling, as encoded into the affects relations, with relativistic causality, as given by the space-time poset $\TT$. Subsequently we will present our main result: a connection between conical space-times and faithful causal models, both independently establishing a correspondence between causal inference and no superluminal signalling.

\subsection{Reviewing compatibility}
\label{sec:compat}

While demanding all causation to go into the future seems to be a natural assumption, it is actually not necessary to prevent superluminal \textit{signalling}, since causation and signalling are inequivalent concepts (owing to the possibility of fine-tuning). The latter is sufficient to disallow agents from being able to transmit information into the spatio-temporal past.

Operationally, transmitting information, or rather, signals, is encoded (in the current framework) into affects relations between the observed RVs.
More precisely, $X \affects Y \given \doo(Z)$ captures that an agent (Alice) who intervenes on $X$ can signal to another agent (Bob) who can observe $Y$ and is given information on interventions performed on $Z$.
If these RVs are embedded in space-time, Bob needs to jointly access $Y$ and $Z$ (which can be done in $\Fut_s (\YY\ZZ) \subseteq \Fut_s (\YY)$) to receive the signal from Alice, who has access to $X$ (accessible in $ \Fut_s (\XX)$).
To avoid superluminal signalling, we therefore need Bob's variables to be jointly accessible in the future of Alice's interventions only.\footnote{
    While this paragraph suggests a picture of \enquote{human} agents, we recall that interventions in causal modelling are not limited to this setting, but more broadly capture which systems and properties can, in principle, be independently accessed and modified within a physical theory, even if no human could realize such an intervention in practice \cite{Vilasini2025}.
    This perspective aligns with arguments towards causal reformulations of physics \cite{Spekkens2012} and the broader program towards information-theoretic and operational approaches to fundamental physics \cite{Hardy2016}.
}

Moreover, as motivated in \cite{VVC} we only need to consider Irred$_1$ affects relations to ensure that the conditions are not unnecessarily strong. For instance, the reducible relation $X_1 X_2 \affects Y$ may arise in a model where $X_1\dircause Y$ and $X_2$ is a dummy variable with no incoming or outgoing arrows (no causes or effects). In that case $\YY \succ \XX_2$ is not implied by any relativistic causality principle, and only the constraint implied by the equivalent irreducible relation $X_1\affects Y$ is relevant for no superluminal signalling, imposing $\YY\succ \XX_1$.
As motivated in \cref{sec:affects-red}, irreducibility serves to eliminate such redundancies.

This can be formalized to the following condition proposed in \cite{VVC}.
For the purposes of this work, we provide a simplified version, with regard to the embedding as well as in the restriction to unconditional affects relations.

\begin{definition}[\textbf{compat}]
    \label{def:compat}
    Let $\SC$ be a set of ORVs from a set of RVs $S$ and a poset $\TT$ with an embedding $\mathcal{E}$. Then a set of unconditional affects relations $\mathscr{A}$ is said to be \emph{compatible} with $\mathcal{E}$ (or satisfies \textbf{compat}) if the following condition holds:
    \begin{itemize}
        \item Let $X, Y \subset S$ be disjoint non-empty sets of RVs, $Z \subset S$ another disjoint set of RVs, potentially empty.
        If ($X \affects Y \given \doo(Z)$) $\in \mathscr{A}$ and is Irred$_1$, then
        $\Fut_s (\YY\ZZ) = \Fut_s (\YY) \cap \Fut_s (\ZZ) \subseteq \Fut_s (\XX) \, .$
    \end{itemize}
\end{definition}

This compatibility condition provides a necessary and sufficient condition to ensure that we have no superluminal signalling once the affects relations are embedded in the space-time \cite{VVC}. Notice that the relevant condition for strictly capturing no superluminal signalling (NSS) is on the intersections and not union of the futures (the latter relates to prohibiting superluminal causation, which is inequivalent to NSS \cite{VVCPR}). In particular, in the jamming scenario of \cref{ex:jamming}, we have $Y\affects AC$ (which is Irred$_1$) and NSS only requires the joint future of the ORVs $\AA$ and $\CC$ to be contained in the future of the ORV $\YY$ and not that $\AA$ or $\CC$ is embedded in the future light cone of $\YY$, since $Y \not\affects A$ and $Y\not\affects C$.
By contrast, a stronger compatibility condition requiring no superluminal \emph{causation} would exclude this scenario \cite{VVCPR}. Indeed, the causal influence $Y\dircause A$ in \cref{ex:jamming} would be superluminal in an embedding where $\AA$ is not in the future light cone of $\YY$, yet such an embedding of this example can respect the weaker principle, NSS as explained above.

\begin{remark}
    \label{thm:emptyset}
    For $X,Y,Z \in S$, consider an affects relation $X \affects Y \given \doo(Z)$.
    If $\Fut_s (\YY\ZZ) = \emptyset$, there is no location in $\TT$ where both $Y$ and $Z$ are accessible.
    Hence, it is impossible to experimentally verify the respective affects relation, rendering it operationally meaningless.
    Similarly, if $\Fut_s (\XX) = \emptyset$, compatibility implies that $\Fut_s (\YY\ZZ) = \emptyset$:
    It is impossible to signal from \emph{nowhere} to \emph{somewhere}.
    Therefore, we will usually disregard the case where either of these sets is empty.
\end{remark}

A physical example for the former case would be for $\YY$ and $\ZZ$ located in two distinct classical black holes or outside their respective cosmological event horizons, as these would have no joint future.

\subsection{A connection between conical space-times and causal models without clustering} \label{sec:compat-indec}
Consider a higher-order affects relation $X\vDash Y|\mathrm{do}(Z)$ between disjoint sets $X$, $Y$ and $Z$ of RVs. The Irred$_1$ property of such an affects relation enables us to infer that each element $e_X \in X$ is a cause of some element $e_Y \in Y$ \cite{VVC}, and as we have shown in \cref{sec:affects-to-cause}, the Irred$_3$ property enables an analogous inference relative to the other interventional argument, $Z$, that all $e_Z\in Z$ are a cause of some $e_Y\in Y$.
In other words: 
\begin{mdframed}
    \centering
    Imposing Irred$_1$ and Irred$_3$ for the higher-order affects relation $X\vDash Y|\mathrm{do}(Z)$ (or $Z\vDash Y|\mathrm{do}(X)$) has the same implications for causal inference as imposing Irred$_1$ for the 0$^\text{th}$-order relation $XZ\vDash Y$.
\end{mdframed}

This indicates an interchangeability between $X$ and $Z$ for causal inference statements derived from higher-order affects relations irreducible in the first and third arguments. Does this interchangeability also carry forth to compatibility constraints imposed by higher-order affects relations in a space-time?
The following example illustrates that the answer to this question is generally negative.

\begin{figure}[t]
    \centering
    \begin{tikzpicture}[dot/.style={circle,inner sep=1pt,fill,name=#1}]
        \node [dot=Z1,label=$\ZZ_1$] at (1.2,1.2) {};
        \node [dot=Y,label=$\YY$] at (2,2) {};
        \node [dot=Z2,label=$\ZZ_2$] at (3.7,1.8) {};

        \node (left) at (0.5,0.5) {};
        \node (right) at (5.0,0.5) {};
        \draw (left) -- ++(3,3);
        \draw (right) -- ++(-3,3);

        \node [dot=X,label=$\XX$] at (2.75,2.75) {};
    \end{tikzpicture}
	\caption{
        Sketch of a compatible non-degenerate embedding of \cref{ex:non-degenerate} into 1+1-Minkowski space-time.
        In this embedding, all ORVs which are located on a light-like surface, and therefore on the boundary of the light cone of the respective earlier RVs.
        As this space-time is not conical, even for $X \affects Y \given \doo(Z)$ satisfying Irred$_1$ and Irred$_3$, in accordance with \cref{thm:fine-tuned-embedding-cond}. we may have $\Fut (\YY) \not\subseteq \Fut (\XX) \cap \Fut_s (\ZZ)$.
    }
	\label{fig:non-degenerate}
\end{figure}
\begin{example}{(Correspondence does not generally hold)}
    \label{ex:non-degenerate}
    Let $X$ and $Y$ be two RVs while $Z:=\{Z_1,Z_2\}$. Then $X\vDash Y|\mathrm{do}(Z)$ is Irred$_1$ by construction. Suppose that it is also Irred$_3$. Then by \cref{def: red3}, it follows that we must have (1) $Z_1\vDash Y|\mathrm{do}(Z_2)$ or $Z_1\vDash Y|\mathrm{do}(XZ_2)$, and (2) $Z_2\vDash Y|\mathrm{do}(Z_1)$ or $Z_2\vDash Y|\mathrm{do}(XZ_1)$. We resolve this condition by picking the first affects relation of (1) and the second of (2) i.e., $Z_1\vDash Y|\mathrm{do}(Z_2)$ and $Z_2\vDash Y|\mathrm{do}(XZ_1)$.
    Consider the space-time embedding of these RVs in 1+1-Minkowski space-time where $\ZZ_1\prec \YY\prec \XX\succ \ZZ_2$, as depicted in \cref{fig:non-degenerate}.
    All three of the above affects relations are compatible with this embedding, as we have $\Fut (\YY) \cap \Fut_s (\ZZ) \subseteq \Fut (\XX)$, $\Fut (\YY) \cap \Fut (\ZZ_2) \subseteq \Fut (\ZZ_1)$ and $\Fut (\YY) \cap \Fut(\XX) \cap \Fut (\ZZ_1) \subseteq \Fut (\ZZ_2)$. However, the compatibility condition implied by an additional Irred$_1$, 0$\,^\text{th}$-order relation $XZ\vDash Y$ is violated as  $\Fut (\YY) \not\subseteq \Fut (\XX) \cap \Fut_s (\ZZ)$.\footnote{
    Further, from \cite{VVC}, we know that $X\vDash Y|\mathrm{do}(Z)$ implies that $Z\vDash Y$ or $XZ\vDash Y$ must hold, which can be resolved by choosing $Z\vDash Y$. This may be reducible or irreducible, we take it to be reducible to $Z_1\vDash Y$, which is by construction irreducible since $Z_1$ is a single RV. This imposes the additional compatibility condition $\ZZ_1\prec \YY$ which is already satisfied in the space-time embedding of our example. Similarly, we can resolve the implied affects relations of  $Z_1\vDash Y|\mathrm{do}(Z_2)$ and  $Z_2\vDash Y|\mathrm{do}(XZ_1)$ with the same affects relation $Z_1\vDash Y$. Thus the compatibility of the affects relations of this example with the given embedding continues to hold when considering implied affects relations.}
\end{example}
This example is described in terms of affects relations (that can arise from some underlying, unknown causal model). It however remains an open question to find a causal model that generates exactly a given set of affects relations (such as those of this example) and no more.
Further, the example is set in 1+1-Minkowski space-time, which is not a conical space-time.

Interestingly, as we show in the following theorem, in conical space-times, the answer to the aforementioned question is positive for any set of affects relations: Here, the interchangeability of $X$ and $Z$ does carry forth to compatibility. This results in a simplification of compatibility considerations, which can have useful applications for characterizing relativistic causality for information-processing protocols in conical space-times (such as our physical 3+1-Minkowski space-time), as further discussed in \cref{sec:conclusion}.

\begin{restatable}{theorem}{irrCompatConical}
	\label{thm:irr-compat}
    Let $\mathscr{A}$ be a set of unconditional affects relations and $\mathscr{A}'\subseteq \mathscr{A}$ consist of all affects relations in $\mathscr{A}$ that are both Irred$_1$ and Irred$_3$.
	Then for any non-degenerate embedding $\mathcal{E}$ into a conical space-time $\TT$ satisfying \hyperref[def:compat]{\textbf{compat}}, we have
	\begin{equation}
		\label{eq:irr-compat-conical-main}
		(X \affects Y \given \doo(Z)) \in \mathscr{A}'
		\quad \implies \quad
	    \Fut_s (\YY) \subseteq \Fut_s (\XX) \cap \Fut_s (\ZZ) \, .
	\end{equation}
\end{restatable}

We briefly sketch the proof of this statement, while deferring a more detailed argument to \cref{sec:proofs6}.\footnote{
    Note that this sketch organizes the main arguments in a different manner than in \cref{sec:proofs6}. There, intermediate results are carved out more explicitly, e.g.\ the implications of Irred$_3$ for compatibility in general space-times.
}
By Irred$_1$ and Irred$_3$, $X \affects Y \given \doo(Z)$ implies (among others) further relations \linebreak $e_X \affects Y \given \doo(ZX\setminus e_X)$ for all $e_X \in X$ (by \cref{def: red1}) as well as either $e_Z \affects Y \given \doo(Z \setminus e_Z)$ or $e_Z \affects Y \given \doo(XZ \setminus e_Z)$ for all $e_Z \in Z$  (by \cref{def: red3}).
By compatibility, these statements jointly imply
$\Fut_s(\ZZ\XX \setminus e_{\XX\ZZ}) \cap \Fut_s(\YY) \subseteq \Fut(e_{\XX\ZZ})$ for all $e_{\XX\ZZ} \in \XX\ZZ$.
Together, these statements imply
$\Fut_s(\YY\ZZ\XX \setminus e_{\XX\ZZ}) = \Fut_s(\YY\ZZ\XX)$ for all $e_{\XX\ZZ} \in \XX\ZZ$ (cf.~\cref{thm:compat-irr-to-strong-second}).
The statement then follows by applying the full version (cf.~\cref{def:location-symmetry-orvs}) of location symmetry (as motivated in \cref{eq:location-symmetry-simple}), which is equivalent to conicality.

\medskip

Beyond \cref{thm:irr-compat}, \cref{thm:irr-compat-embedding} demonstrates that this theorem can be generalized to some embeddings into non-conical space-times, which we introduce as \emph{conical embeddings} in \cref{def:conical-embedding}.

However, imposing restrictions on the embedding and the underlying space-time is not the only way to recover this resemblance between causal inference and space-time structure.
Alternatively, we can restrict the allowed sets of affects relations, and by extension, the space of causal models giving rise to them.
In particular, we can demand the absence of clustering in the third argument. 

\begin{restatable}{theorem}{irrCompatClusThree}
    \label{thm:irr-compat-clus}
    Let $\mathscr{A}$ be a set of unconditional affects relations arising from a causal model not yielding any affects relations with Clus$_3$
    and let $\mathscr{A}'\subseteq \mathscr{A}$ consist of all affects relations in $\mathscr{A}$ that are both Irred$_1$ and Irred$_3$. Then compatibility of $\mathscr{A}$ with any embedding $\cE$ in \underline{any} space-time $\cT$ implies \cref{eq:irr-compat-conical-main}.
\end{restatable}

This secondary result follows as for a causal model, absence of clustering in interventional arguments ultimately allows to infer that $e_{XZ} \affects Y$ for all $e_{XZ}$.
We defer a detailed proof to \cref{proof:compatClusThree}.

In particular, the above theorem holds for any faithful (or not fine-tuned) causal model since such models cannot have any clustered affects relations (\cref{lemma: clus_finetune}). Accordingly, we can check that the set of affects relations considered in \cref{ex:non-degenerate} does indeed admit affects relations clustered in the third argument, as $X \affects Y \given \doo(Z)$, yet $X \naffects Y$. This implies (by \cref{lemma: clus_finetune}) that any causal model that can give rise to these affects (and non-affects) relations must be fine-tuned.

\begin{table}
    \centering
    \scalebox{1.25}{$X \affects Y_1 Y_2 \given \doo(Z_1 Z_2)$ with Irred$_1$}\\
    \vspace{.2cm}
    \begin{tabular}{r|c|c|}
        \cline{2-3}
        & $X, Z_1, Z_2 \, \longrsquigarrow \, Y_1 Y_2$ &
        $\highlight{red!20}{\Fut_s (\XX \ZZ_1 \ZZ_2)} \supseteq \highlight{blue!20}{\Fut_s (\YY_1 \YY_2)}$ \\
        \rule{0pt}{2ex}    
        & for Causal Inference?
        & for Space-time Embedding? \\
        & & $^\text{with no superluminal signalling}$ \\
        \hline
        \multicolumn{1}{|r|}{in general}
        & \textcolor{red}{\xmark} & \textcolor{red}{\xmark} \\
        \multicolumn{1}{|r|}{w/ Irred$_3$}  & \textcolor{teal}{\cmark} & \textcolor{red}{\xmark} \\
        \multicolumn{1}{|r|}{w/ Irred$_3$ \& no Clus$_3$ in model} & \textcolor{teal}{\cmark} & \textcolor{teal}{\cmark} \\
        \multicolumn{1}{|r|}{w/ Irred$_3$\,\& \hyperref[def:conical]{conical space-time}} & \textcolor{teal}{\cmark} & \textcolor{teal}{\cmark} \\
        \multicolumn{1}{|r|}{w/ Irred$_3$\,\& \hyperref[def:conical-embedding]{conical embedding}} & \textcolor{teal}{\cmark} & \textcolor{teal}{\cmark} \\
        \hline
        & \begin{tikzpicture}[scale=0.9]
            \node (placeholder) at (0,2.5) {};
            \begin{scope}[every node/.style={circle,thick,draw,inner sep=0pt,minimum size=0.8cm}]
                \node (X) at (0,0) {$X$};
                \node (Z1) at (1.5,0) {$Z_1$};
                \node (Z2) at (3,0) {$Z_2$};
                \node (Y1) at (0.5,2) {$Y_1$};
                \node (Y2) at (2.5,2) {$Y_2$};
            \end{scope}

            \begin{scope}[>={Stealth[black]},
                          every edge/.style=vvarrow]
                \path [->] (X) edge (Y1);
                \path [->] (X) edge (Y2);
                \path [->] (Z1) edge (Y2);
                \path [->] (Z1) edge (Y1);
                \path [->] (Z2) edge (Y2);
            \end{scope}
        \end{tikzpicture}        
        & \begin{tikzpicture}[dot/.style={circle,inner sep=1pt,fill,name=#1}]
    		\node [dot=A,label=$\mathcal{X}$] at (0,0) {};
    		\node [dot=B,label=$\mathcal{Z}_1$] at (1.5,0) {};
    		\node [dot=C,label=$\mathcal{Z}_2$] at (3,0) {};
            \node [dot=Y1,label=$\mathcal{Y}_1$] at (0.75,1.5) {};
            \node [dot=Y2,label=$\mathcal{Y}_2$] at (2.25,1.25) {};
    		
    		\fill[fill=red!20] (1.5,1.5) -- (2.5,2.5) -- (.5,2.5) -- (1.5,1.5);
            \fill[fill=blue!20] (1.375,2.125) -- (1.75,2.5) -- (1,2.5) -- (1.375,2.125);
    		
    		\draw (A) -- ++(-1,1);
    		\draw (A) -- ++(2.5,2.5);
    		\draw (B) -- ++(-2,2);
    		\draw (B) -- ++(2.125,2.125);
    		\draw (C) -- ++(-2.5,2.5);
    		\draw (C) -- ++(1,1);

            \draw (Y1) -- ++(1,1);
            \draw (Y1) -- ++(-1,1);
            \draw (Y2) -- ++(1.25,1.25);
            \draw (Y2) -- ++(-1.25,1.25);
    	\end{tikzpicture} \\
        & \scriptsize  Possible compliant causal structure\! 
        & \scriptsize Possible compliant space-time embedding \\
        \cline{2-3}
    \end{tabular}
    \caption{The consequences of an affects relation $X \affects Y_1 Y_2 \given \doo(Z_1 Z_2)$ regarding causal inference and space-time compatibility, given that various additional conditions apply.
    Here, $X, Y_1, Y_2, Z_1, Z_2$ represent individual RVs in the causal model.
    Therefore, this affects relation is Irred$_1$ by definition.
    The juxtaposition within the table demonstrates the correspondence between causal and spatiotemporal order, which is established under certain additional conditions only.    
    The example is based on \cref{ex:non-degenerate}, albeit considering $Y_1$ and $Y_2$ rather than a single RV $Y$ in its second argument.}
    \label{tab:summary}
\end{table}

This theorem yields a connection between properties of causal models and the geometry of space-time, namely between (1) causal models with no clustering in the interventional arguments of its affects relations and (2) space-times satisfying conicality.\footnote{
    This parallel is rooted in the fact that both of these properties affirm the impact of \textit{individual} elements in the interventional arguments of affects relation, albeit in different ways:
    On the one hand information-theoretically, the absence of clustering conveys the intuition that generally, there exist RVs in the respective affects relation arguments that individually contribute to signalling (cf.\ \cref{sec:clus-prop}).
    On the other hand relativistically, by conicality, for a non-degenerate compatible embedding, there exist ORVs (the spanning elements) which individually contribute to the joint future.
}
Specifically, whenever condition (1) or (2) holds, we have an interchangeability between $X$ and $Z$ for higher-order affects relations irreducible in the first and the third arguments, for the purpose of constraints imposed by no superluminal signalling (i.e., compatibility of the causal model with the space-time). This interchangeability always holds for causal inference constraints, by \cref{thm:irr-indec}.
Effectively, this then tells us that in scenarios satisfying (1) or (2), compatibility implies that the interventional data given by $X$ and $Z$ in a higher-order affects relation must be entirely in the ``past'' of the region where the observational data captured by $Y$ is completely accessible, irrespective of whether ``past'' is defined relative to the relations $\dircause$ of the information-theoretic causal model or the relations $\prec$ capturing the light cone structure of the space-time.
This main result is illustrated and summarized in \cref{tab:summary}.

In \cref{sec:no-irreducible}, we introduce another compatibility condition named \textbf{compat-atomic}, which is a restriction of \textbf{compat} (\cref{def:compat}) to affects relations $X\affects Y\given \doo(Z)$  where $X$ is a single RV.
While this does not fully capture no superluminal signalling and is not equivalent to \textbf{compat} in general, we will show that \textbf{compat-atomic} and \textbf{compat} are equivalent under analogous restrictions as the above theorem: namely when restricting the space-time to be conical or the causal model to have no clustered relations of a certain type. This suggests another correspondence between causal inference and space-time geometry which also leads to notable simplifications.

\subsection{Generalizing the connection to include conditional affects relations}

Any conditional affects relation $X \affects Y \given \doo(Z), W$ (\cref{def:affects-cond}) implies the presence of an unconditional affects relation  $X \affects YW \given \doo(Z)$, as shown in Lemma~IV.8 of \cite{VVC} (see also \cref{thm:decondition}).
Therefore, it is possible to transform each set of (possibly conditional) affects relations to a set of unconditional affects relations with the same implications for causal inference (according to \cref{sec:affects-to-cause}) and compatibility.
This follows from \cref{thm:decondition} and \cref{thm:decondition3}.
This means that all statements of the previous section also apply for conditional affects relations when replacing $Y$ with $YW$ accordingly.
Together with \cref{thm:irr-indec}, \cref{thm:irr-compat} and \cref{thm:irr-compat-clus}, this yields that if either the space-time is conical or the causal model is not fine-tuned, an affects relation $X \affects Y \given \doo(Z), W$ with Irred$_1$ and Irred$_3$ is equivalent to an affects relation $XZ \affects YW$ for the purpose of causal inference (according to \cref{sec:affects-to-cause}) and compatibility with space-time structure\footnote{For conditional affects relations, compatibility is spelled out in \cref{def:compat-cond}.}.

This allows for direct generalization of statements concerning the embeddability of sets of unconditional 0$^\text{th}$-order affects relations into space-time to sets of conditional higher-order affects relations.
Conjoining the contrapositions of our two main theorems and generalizing to conditional affects relations, we obtain:
\begin{corollary}
    \label{thm:fine-tuned-embedding-cond}
    Let $\mathscr{A}$ be a set of (conditional) affects relations (c.f. \cref{def:affects-cond}) and $\mathscr{A}'\subseteq \mathscr{A}$ consist of all affects relations in $\mathscr{A}$ that are both Irred$_1$ and Irred$_3$. Then the existence of a non-degenerate embedding $\cE$ in a space-time $\cT$ with 
    \begin{equation}
		(X \affects Y \given \doo(Z), W) \in \mathscr{A}'
		\quad \text{and} \quad
		\Fut_s (\YY) \cap \Fut_s (\WW) \not\subseteq \Fut_s (\XX) \cap \Fut_s (\ZZ) \, .
	\end{equation}
    implies that (1) the affects relations $\mathscr{A}$ exhibit interventional fine-tuning (Clus$_3$) and that (2) the space-time $\TT$ does not satisfy conicality.
\end{corollary}
This shows that in a conditional higher-order affects relation, the interventional (1$^\text{st}$ and 3$^\text{rd}$) and observational (2$^\text{nd}$ and 4$^\text{th}$) arguments respectively exhibit interchangeability amongst each other when either restricting to conical space-times or affects relations without clustering.

To summarize, we have established that in cases (1) and (2) of the above corollary, an affects relation $X\affects Y\given \doo(Z), W$ being Irred$_1$ and Irred$_3$ gives us the following causal inference and compatibility statements, for any $e_{XZ}\in XZ$:
\begin{align}
    \label{eq: correspondence_main}
    \begin{split}
        \textsf{\textbf{Causal inference: }} & e_{XZ} \text{ is a cause of } YW \\
        \textsf{\textbf{Compatibility: }}    & \Fut (e_{\XX\ZZ}) \supseteq \Fut_s (\YY\WW)
    \end{split}
\end{align}
These exhibit an analogous order structure.
The former statement follows from the results of \cref{sec:affects-to-cause} holding independently of restrictions (1) and (2), while the latter statement does not generally hold without these restrictions. 
\section{Conclusions}
\label{sec:conclusion}

The affects framework \cite{VVC} provides a formal platform for investigating the interplay between information-theoretic and spatio-temporal causation in rather general scenarios. For addressing many of the open problems brought to light by this approach, it is beneficial to develop a tighter characterization of the associated concepts and techniques, relating to causal modelling, space-time structures and the compatibility of the two. Within our work we have improved on these characterizations, generating new insights and tools for many of the open problems identified in \cite{VVC}.

For causal models, our contributions are three-fold. We generalized the concept of reducibility to identify redundancies in different arguments of an affects relation, which captures information-theoretic signalling through interventions in a causal model. Introducing the concept of clustering of affects relations, we characterized different types of operationally detectable fine-tunings depending on whether it involves \emph{observational} or \emph{interventional} arguments. Subsequently, we derived applications of these concepts for causal inference, showing that interestingly, the absence of certain affects relations can also be employed for causal inference. 

For space-time structure, we have introduced the order-theoretic property of \emph{conicality} showing that it is satisfied in Minkowski space-times with $d>1$ spatial dimensions and violated for $d=1$. Using this, we have shown a connection (\cref{thm:fine-tuned-embedding-cond}) between conical space-times and causal models without a form of clustering.
Both cases lead to a correspondence between the causal order constraints on the observed variables coming from compatibility with a space-time embedding and those coming from purely information-theoretic causal inference (although these two types of constraints behave differently in general).
Moreover, this correspondence reveals that in conical space-times, the principle of no superluminal signalling ensures a clear temporal ordering originating from an affects relation $X \affects Y \given \doo(Z),W$ irreducible in $X$ and $Z$: between when the interventional arguments $\XX$ and $\ZZ$ and the observational arguments $\YY$ and $\WW$ are jointly accessible, with the former ordered before the latter.

While we have focused on foundational questions here, the relevance of fine-tuning in the security of cryptographic protocols (as highlighted in \cite{VVC}) motivates potential applications of our results in more practical scenarios. Moreover, given the widespread use of causal modelling and inference in data driven disciplines, combined with the interest in order-theoretic properties of space-time in general relativity and quantum gravity approaches, this work may be of interest in broader communities beyond quantum information and physics.

\subsection{Open questions}

There is still much scope for future work. We discuss some interesting future directions that can be investigated by building on this work and the affects framework.

\subsubsection*{Causal models and inference}

Here, we have focused on causal models defined using the d-separation property. This property holds for all acyclic models \cite{Pearl2009, Henson2014} and for a class of cyclic models \cite{VVC, VVR} in quantum and post-quantum theories, but can fail in certain cyclic causal models already in classical theories \cite{Bongers_2021}. A natural open question is whether all the techniques used here generalize to the case when the d-separation condition is replaced by the $\sigma$-separation condition \cite{arxiv.1710.08775}, which applies to an even wider variety of cyclic (and possibly continuous variable) classical causal models, or to $p$-separation, which holds in all consistent finite-dimensional cyclic quantum causal models \cite{Ferradini2025classical, Ferradini2025quantum}.
Further, we have focused on unconditional affects relations. While we generalized the concepts and some of the results to conditional relations, we leave the full generalization for future work, as discussed in \cref{sec:conditional}.

Another important open problem relates to the classification of fine-tuning, as initiated in \cref{sec:affects-clus} through the concept of clustering. As highlighted there, clustering in different arguments presents us with observational and interventional types of fine-tuning, but there can be additional forms of fine-tuning relating to the fine-grainability of the causal model (in the sense introduced in \cite{VVR}), or to the lack of full support in the probability distribution (i.e.\ with some probabilistic events that never happen).
It remains to be explored whether there can exist further types of fine-tuning of a causal model, which are not explained by these types. 

Such classifications of fine-tuning have applications in developing robust algorithms for causal discovery in the presence of fine-tuning, as existing methods typically assume faithfulness.
This assumption poses inherent challenges for causal inference in the presence of fine-tuning \cite{Spirtes1993, Pearl2009}.
Causal discovery under relaxed faithfulness assumptions, such as adjacency faithfulness, has been explored in \cite{Zhang2008, Zhalama2017}.
Adjacency faithfulness requires that only adjacent nodes in a causal graph be correlated, allowing for fine-tuned situations where correlations are non-transitive.
We have seen a concrete example of such a scenario in \cref{ex:sorkin} and \cref{fig:fine-graining}, where fine-graining of $B$ into smaller nodes highlights the non-transitivity of correlations.
The concept of clustering introduced in this paper applies to forms of fine-tuning that do not satisfy adjacency faithfulness, such as the one-time pad causal model in \cref{eg: red3_1}, where $X$ and $Y$ as well as $Z$ and $Y$ are adjacent but uncorrelated.
Despite this, we have shown that causal inference is still possible in such fine-tuned models via interventions (cf.\ \cref{sec:affects-to-cause}).
This suggests that the techniques presented here may extend the applicability of causal discovery algorithms—both for classical and non-classical scenarios—to a broader class of fine-tuned or unfaithful models, which we leave for future work.

A related question is whether it is possible to detect fine-tuning through interventions when it is impossible to detect it through observations alone. Formally, this relates to a question raised in \cref{sec:d-separation} regarding equivalence of two definitions of fine-tuning: one which compares d-separation and conditional independence in the pre-intervention causal model alone and the other where all post-intervention models are also taken into account. Presently, it is unknown if there is a gap between these definitions.  

The affects framework significantly abstracts the usual causal modelling approaches, showing that many known results (such as Pearl's rules of do-calculus \cite{Pearl2009}) can be obtained from more minimal assumptions. 
However, proving the completeness of the causal inference results (cf.\ \cref{sec:affects-to-cause}) derived under these minimal assumptions remains an important open challenge.
A related open problem lies in understanding the interaction between the presence and absence of certain affects relations for a given causal model, especially to certify the completeness of a set of affects relations (i.e.\ a set of affects relations which does not imply the presence of further affects relations, e.g.\ via the transformation rules established in \cite{VVC}).
This is exacerbated further due to the wide variety of causal models, differing in causal structure and both domain (binary / higher cardinality / continuous) and distribution of their individual RVs.
It would be interesting to prove such completeness results, even under certain restrictions on the causal models (possibly coming from the form of causal mechanisms assumed for the underlying theory). For further discussions in this direction see Section 3.1 of~\cite{Master}.

\subsubsection*{Space-time structure}

Setting aside the manifold structure of space-time, there are interesting questions regarding the characterization of space-times using purely order-theoretic statements.
In particular, can we characterize the set of space-times which satisfy \textit{conicality}, both for pseudo-Riemannian manifolds and for discrete models of space-time?
Are there further, yet to be uncovered, properties of light cones that are reflected within the associated partial order, which exhibit useful correspondences with causal inference concepts? Conversely, can we formulate notions such as space-time distance and curvature in terms of graph-theoretic properties studied in causal models, such as d-separation? Progress in this direction could provide valuable tools exchange of techniques with programs that aim on the discretization of space-time, as is the case for causal set quantum gravity \cite{PhysRevLett.59.521, Surya2019} or causal dynamic triangulations \cite{Loll2019} as at least in the limit, certain properties of light cones need be recovered by physical space-times.

\subsubsection*{Compatibility of causal models and space-time}

The possibility of operationally detectable causal loops in 1+1-Minkowski space-time without superluminal signalling was shown in \cite{VVC_Letter}. In the technical language of the framework, these are cyclic causal loops in a causal model whose presence can be certified through the resulting affects relations (and hence called affects causal loops or ACLs), which can be compatibly and non-degenerately embedded in a partially ordered set.
A crucial question in this regard has been whether such loops are possible in higher-dimensional Minkowski space-time.

The formal concept of conicality introduced here distinguishes between Minkowski space-time with $d=1$ and $d>1$ spatial dimensions, as the former does not satisfy it (and can embed such ACLs) while the latter does (and no such embeddings are known). Moreover, it is only possible to compatibly and non-degenerately embed such ACLs in a poset when the causal model is fine-tuned \cite{VVC}.
The link between compatible embeddings in conical space-times and those of faithful (i.e., not fine-tuned) causal models found in \cref{sec:compat-indec} lends support to the conjecture that compatible and non-degenerate embeddings of ACLs in conical (and hence higher-dimensional Minkowski, \cref{thm:conical}) space-times are impossible.
In fact, the briefly discussed notion of a conical embedding (cf.\ \cref{def:conical-embedding}) allows to extend this connection to many embeddings into non-conical space-times.
This may allow to extend the scope of this conjecture further:
For instance in 1+1-dim.\ Minkowski space-time, most embeddings are in fact conical, and slight perturbations of non-conical embeddings are sufficient to restore their conicality.
This indicates a notion of fine-tuning of the space-time embedding, which may be necessary for any compatible embeddings of ACLs. 
This notion is analogous to fine-tuning in operational causal models, where a non-signalling relation can turn into signalling through slight perturbations of the causal mechanisms.

In this regard, it is worth noting that the structure of ACLs can be rather complicated, even more so when higher-order affects relations are involved.
Already for the 0$^\text{th}$-order case, several distinct classes have been identified in \cite{VVC}.
Our results show that when studying compatible embeddings in conical space-times, for affects relations irreducible in the first and third arguments, higher-order affects relations can be replaced equivalently with 0$^\text{th}$-order ones.
This provides a useful simplification for studying compatible embeddings of general ACLs involving higher-order affects relations, in conical space-times.
Building on the present work, we aim to prove the aforementioned conjecture about such ACLs in a follow-up work \cite{Paper13}. Although this is not the case for non-conical space-times \cite{VVC_Letter}, the conjecture being true would show that in conical space-times -- and more generally, for conical embeddings -- the fundamental principle of no superluminal signalling is sufficient for ruling out all (operationally detectable) causal loops. 

\subsubsection*{Applications to cryptography and quantum tasks in space-time}

The concept of clustering captures signalling between sets of nodes without signalling between individual nodes. This is a property that is necessary for the security of cryptographic protocols such as the one-time pad and secret sharing schemes \cite{PhysRevA.59.1829, PhysRevA.61.042311, PhysRevA.59.162}, where information is distributed over multiple systems and cannot be recovered from a subset thereof. This is also similar to the desired properties of quantum error correcting codes \cite{Gottesmanphd, Cleve1999} that are intimately linked to quantum information processing tasks in space-time such as summoning \cite{Kent_2012_summon, Kent_2012_tasks, Hayden2016}, where achieving the task efficiently can require quantum information to be cleverly distributed over multiple space-time locations.
Within the language of tensor networks, some examples of this have previously been studied in \cite{Cotler2019}.
Developing a causal modelling approach to such protocols, and investigating the applications of these concepts there, hence presents an intriguing inter-disciplinary avenue for future research.

In this regard, it is important to note that the affects framework and our work focus on signalling between classical variables obtained by interacting with (e.g., by measuring) non-classical systems, and it makes no assumptions about the non-classical causal mechanisms. This facilitates intrinsically device- and theory-independent statements. In quantum theory, we understand the causal mechanisms and can define signalling between quantum in/output systems of quantum channels \cite{Schumacher2004, Portmann_2017, arxiv.1906.10726, Ormrod2023} and study their compatibility with space-time \cite{VVR, Salzger2025}. Developing the language of conditional higher-order affects relations further for the case of quantum theory, can enable us to study more refined notions of signalling in quantum circuits, where conditioning on additional interventions and measurements outcomes may allow or forbid signalling between agents. Moreover, we focused on space-time embeddings where variables are well-localized at space-time points. More generally, systems may be embedded into space-time regions and we can have protocols where parties exchange quantum systems at superpositions of space-time locations \cite{Portmann_2017}, leading to situations where quantum information is distributed in space-time in a non-localized manner. The interplay of quantum information and relativistic notions of causality in such space-time based quantum protocols has been formalized in frameworks such as \cite{VVR}.
Investigating the aforementioned applications of our results on affects relations and causal inference for quantum protocols in space-time will therefore benefit from combining and linking techniques from the frameworks of \cite{VVC} and \cite{VVR}.

\section*{Acknowledgements}

We thank Roger Colbeck, Augustin Vanrietvelde, Alastair Abbott and the anonymous reviewers of TQC 2023 for their helpful comments. We also thank Stav Zalel, Federico Grasselli, Ettore Minguzzi and Andrés Agustí Casado for interesting exchanges and discussions on the topic.
MG acknowledges financial support by l’Agence Nationale de la Recherche (ANR), project ANR-22-CE47-0012.
VV acknowledges support from an ETH Postdoctoral Fellowship, the ETH Zurich Quantum Center, the Swiss National Science Foundation via project No.\ 200021\_188541, the QuantERA programme via project No.\ 20QT21\_187724 and the PEPR integrated project EPiQ ANR-22-PETQ-0007 as part of Plan France 2030.
For the purpose of open access, the authors have applied a CC-BY public copyright licence to any Author Accepted Manuscript (AAM) version arising from this submission.

\printbibliography

\newpage
\appendix

\section{Further details on the causal modelling framework}
\label{sec:d-separation}

In this section, we provide a formal definition for the technical concept of d-separation, which was described on an intuitive level in \cref{sec:review}. This concept was originally introduced in the classical causal modelling literature (e.g., \cite{Pearl2009}) and has been applied to define a general class of causal models applicable to situations with non-classical, cyclic and fine-tuned causal influences, in the affects framework of \cite{VVC}. We provide a more technical overview of these concepts here.

\begin{definition}[Blocked Paths]
    \label{def:blocked}
    Let $\mathcal{G}$ be a directed graph, where $X$ and $Y$ are distinct nodes and $Z$ is a set of nodes not containing $X$ and $Y$.
    An (undirected) path from $X$ to $Y$ is \emph{blocked} by $Z$ if the path contains $A, B$ such that either
    $A \dircause W \dircause B$,
    $A \diresuac W \dircause B$ with $W \in Z$, or
    $A \dircause V \diresuac B$ with neither $V$ nor any descendant of $V$ in $Z$.
\end{definition}

\begin{definition}[d-separation]
    \label{def:d-sep}
    Let $\mathcal{G}$ be a directed graph, with $X, Y, Z$ being disjoint subsets of nodes.
    $X$ and $Y$ are \emph{d-separated} by $Z$ in $\mathcal{G}$, denoted as $(X \perp^d Y \,|\, Z)_\mathcal{G}$, if every path from an element of $X$ to an element of $Y$ is blocked by $Z$.
    Otherwise, $X$ is \emph{d-connected} to $Y$ given $Z$. If obvious from context, the index $\mathcal{G}$ may be suppressed.
\end{definition}

Using d-separation, the following minimal definition of a causal model was proposed in \cite{VVC}.

\begin{definition}[Causal model]
	\label{def: causalmodel}
	A causal model over a set of observed random variables $\{X_1,...,X_n\}$ consists of a directed graph $\mathcal{G}$ over them (possibly additionally involving classical/quantum/GPT unobserved systems) and a joint distribution $P_{\cG}(X_1,...,X_n)$ that respects the d-separation property relative to $\cG$. 
	\end{definition}

\begin{definition}[d-separation property]
	\label{definition: compatdist}
	Let $\{X_1,...,X_n\}$ be a set of random variables denoting the observed nodes of a directed graph $\mathcal{G}$, and $P(X_1,...,X_n)$ be a joint probability distribution over them. Then $P$ is said to satisfy the \emph{d-separation property} with respect to $\mathcal{G}$ if for all disjoint subsets $X$, $Y$ and $Z$ of $\{X_1,...,X_n\}$,
	\begin{equation}
		X\perp^d Y|Z \quad\implies\quad X\indep Y|Z \quad \text{ i.e., $P(XY|Z)=P(X|Z)P(Y|Z)$.}
	\end{equation}
\end{definition}

We continue by stating formally how an intervention on a set of nodes $I$ influences the causal model, i.e.\ how the post-intervention causal model relates to the pre-intervention model, following the formulation of \cite{VVCJ}.
This procedure, as developed within \cite{VVC}, is independent on causal mechanisms or other information on the underlying theory, but is based directly on \cref{definition: compatdist} instead.

\begin{definition}[Post-intervention causal model]
\label{def:post_intervention}
Consider a causal model on a graph $\cG$ specified by the graph together with a distribution $P_{\cG}$  satisfying Definition~\ref{definition: compatdist}. A post-intervention causal model associated with interventions on a subset $X$ of the observed nodes of $\cG$ is specified by a graph $\cG_{\mathrm{do}(X)}$ obtained from $\cG$ by removing all incoming directed edges to the set $X$, together with a distribution  $P_{\cG_{\mathrm{do}(X)}}$ satisfying Definition~\ref{definition: compatdist} relative to $\cG_{\mathrm{do}(X)}$. 
Further, consider a (potentially empty) set of observed nodes $Y$ in $\cG$ and a disjoint set of nodes $x$ exogenous in $\cG$.
Then, we require
\begin{equation}
    P_{\cG_{\mathrm{do}(Y)}} (S|XY) =
    P_{\cG_{\mathrm{do}(XY)}} (S|XY)
\end{equation}
where $S$ is the set of the remaining observed nodes.
\end{definition}

For all the results in this paper we only use the above-mentioned general definition of a post-intervention causal model. However, for specific classical examples, we will apply the usual causal modelling approach of Pearl \cite{Pearl2009} to describe the post-intervention scenario (which adheres with the minimal requirements of the above definition). Explicitly, a classical causal model on a graph $\cG$ is typically specified by providing a probability distribution $P(X)$ for every parentless node $X$ together with a function $f_X:\mathrm{par}(X)\mapsto X$ from the set of all parents of $X$ to $X$. The post-intervention model associated with interventions on a set $S$ of nodes is then obtained by considering the post-intervention graph $\cG_{\mathrm{do}(S)}$ as defined before (where the incoming arrows to $S$ are removed), together with a causal model that is identical to the original model except that for every $X\in S$, we replace the dependence $X=f_X(\mathrm{par}(X))$ with $X=x$ (denoting that $X$ takes on a fixed value $x$).

Having defined post-interventional models, we now formalize what we mean by fine-tuning and faithfulness (the absence of fine-tuning). Faithfulness is related to the idea that conditional independences in the distribution imply a corresponding d-separation in the graph (intuitively, the probabilities are faithful to the graph structure). Usually this is defined only by considering the original pre-intervention model, but here we extend the definition by accounting for all post-intervention causal models.

\begin{definition}[Faithfulness and fine-tuning]
    \label{def:faithful}
    A causal model over a directed graph $\cG$ is said to be faithful if for every post-intervention causal structure $\cG_{\mathrm{do}(I)}$ obtained from $\cG$ by intervening of a set $I$ of observed nodes, and for any mutually disjoint subsets $X, Y, Z$ of the observed nodes (possible overlapping with $I$),
    \begin{equation}
        (X \upmodels Y|Z)_{\cG_{\mathrm{do}(I)}}\quad \implies\quad  (X \perp^d Y|Z)_{\cG_{\mathrm{do}(I)}}.
    \end{equation}
    Otherwise, the causal model is said to be fine-tuned.
\end{definition}

To the best of our knowledge, it is generally unknown whether the above definition is equivalent to its restriction to the case $I = \emptyset$ i.e., whether fine-tuning in a post-intervention model (for some choice $I$ of nodes being intervened on) always implies fine-tuning in the original (or pre-intervention) causal model.
While this direction is interesting to investigate further, our results here are independent of whether or not this equivalence holds.‌
One way to prove the equivalence (which we leave for future work) would be to show that 
in any general (d-separation) causal model which is faithful by the restricted definition, an intervention such that $(X \not\upmodels Y|Z)_{\cG}$ and $(X \upmodels Y|Z)_{\cG_{\mathrm{do}(I)}}$ always implies that $(X \perp^d Y|Z)_{\cG_{\mathrm{do}(I)}}$.\footnote{Note that for $(X\upmodels Y|Z)_{\cG}$, we have the d-separation condition $(X\perp^d Y|Z)_{\cG}$, which follows from the restricted definition of fine-tuning, noting that $\cG:=\cG_{\mathrm{do}(I)}$ for $I=\emptyset$. And $(X\perp^d Y|Z)_{\cG}$ implies $(X\perp^d Y|Z)_{\cG_{\mathrm{do}(I)}}$ for any choice of $I$ since removing arrows cannot add d-connections.}

We conclude by pointing out that while it does capture a variety of cyclic causal models, some functional equation models are known which do not satisfy d-separation \cite{Neal2000}.
To capture these scenarios, the generalized notion of $\sigma$-separation can be used \cite{arxiv.1710.08775}. However, $\sigma$-separation is not the most general property either, and there are causal models known to violate $\sigma$-separation as well. Another work involving one of us proposes another graph separation property called $p$-separation that is valid for all finite-dimensional, possibly cyclic quantum causal models \cite{Ferradini2025quantum} (and therefore for all finite and discrete variable classical causal models \cite{Ferradini2025classical}). The extension of our results to these alternative graph separation properties is an interesting avenue for future work.

\section{Properties of partially ordered sets}
\label{sec:poset-props}
In this section, we will provide a formal introduction of posets.
We will study properties posets may show as well as their relation to the properties of relevant physical examples for $\TT$.

\begin{definition}[Partial Order]
    A \emph{(strict) partial order} is a binary relation $\prec$ on a set $T$ which satisfies
    \begin{itemize}
        \item Irreflexivity: $a \not\prec a$.
        \item Asymmetry: $a \prec b \implies b \not\prec a$.
        \item Transitivity: $a \prec b$ and $b \prec c \implies a \prec c$.
    \end{itemize}
    for all $a, b, c \in T$.
\end{definition}

\begin{definition}[Poset]
    A poset $\TT$ is given by a set $T$ endowed with a partial order: $\TT := (T, \prec)$.
\end{definition}

We start by introducing a notion of immediate neighbors:
\begin{definition}
	\label{def:cover}
	Let $\TT$ be a poset and $x, y \in \TT$. We say $y$ \emph{covers} $x$ if $x \prec y$ and there is no $z \in \TT$ such that $x \prec z \prec y$.
\end{definition}
In contrast to the remaining definitions in this section, this notion is only meaningful for discrete posets.
We can use this definition to give discrete posets a graphical representation in terms of Hasse diagrams \cite{Birkhoff1940}, as illustrated in \cref{fig:hasse}.

In the reminder of this section, we review the definitions of join and meet and how they yield the notion of an \textit{(order) lattice}.
We will then identify how these generic notions relate to physical space-time.

\begin{definition}[Join]
	\label{def:join}
	Let $\TT$ be a poset. For two elements $x, y \in \TT$, their \emph{join} $x \lor y$ is an element of $\TT$ such that:
	\begin{itemize}
		\item $x \preceq x \lor y$ and $y \preceq x \lor y$
		\item if $x \preceq a$ and $y \preceq a$, then $x \lor y \preceq a \ \forall a \in \TT$
	\end{itemize}
	or equivalently: $\forall a \in \TT, \ x \lor y \preceq a \iff x \preceq a \text{ and } y \preceq a$.
\end{definition}

\begin{definition}[Meet]
	Let $\TT$ be a poset. For two elements $x, y \in \TT$, their \emph{meet} $x \land y$ is an element of $\TT$ such that:
	\begin{itemize}
		\item $x \succeq x \land y$ and $y \succeq x \land y$
		\item if $x \succeq a$ and $y \succeq a$, then $x \land y \succeq a \ \forall a \in \TT$
	\end{itemize}
	or equivalently: $\forall a \in \TT, \ x \land y \succeq a \iff x \succeq a \text{ and } y \succeq a$.
\end{definition}

For each pair of elements in a poset, the existence of their join and meet is not guaranteed. However if they exist, they are respectively unique. Both fulfill the following properties, where $\odot$ stands either for $\land$ or $\lor$:
\begin{itemize}
	\item Associativity: $(a \odot b) \odot c = a \odot (b \odot c)$
	\item Commutativity: $a \odot b = b \odot a$
	\item Idempotency: $a \odot a = a$
\end{itemize}

Generalizing the concept of the join, we will define the notion of the minimal elements of a poset. Analogously, we could also define the maximal elements for the meet.
This, however, will not be required going forward, as we will focus on the causal future to model the availability of information.

\begin{definition}[Minimal Elements of a Poset]
	\label{def:min}
	Let $M \subseteq \TT$. Then
	$\min M := \{ x \in M \,|\, \not\exists \, y \in M : y \prec x \}$.
\end{definition}

\begin{lemma}
	\label{thm:join-min}
	Let $\TT$ be a poset.
	If and only if two elements $x, y \in \TT$ have a join $x \lor y$, then $\min \{ a \in \TT | x \preceq a \succeq y \} = \{ x \lor y \}$.
\end{lemma}%
\begin{proof}
This follows directly from the definition of the join.
\end{proof}
\begin{definition}[Semilattice]
	\label{def:semilattice}
	A poset $\TT$ is a \emph{join-semilattice} if $\forall x, y \in \TT$ the join $x \lor y$ exists. Dually, $\TT$ is a \emph{meet-semilattice} if $\forall x, y \in \TT$ the meet $x \land y$ exists.
\end{definition}

Intuitively, it becomes apparent that precisely for a join-semilattice, we can associate \textit{sets} of points in $\TT$, or equivalently sets of ORVs, with a single location in space-time to encode their joint future.

\begin{definition}[Join- and Meet-free Poset]
	\label{def:join-free}
	Let $\TT$ be a poset.
	We call $\TT$ a \emph{join-free poset} if for all $x, y \in \TT$, the existence of a join $x \lor y$ implies $x \lor y = x$ or $x \lor y = y$.
	Dually, $\TT$ is a \emph{meet-free poset} if no meet $x \land y$ exists for analogous choices of $x$ and $y$.
\end{definition}

Hence, a join-free poset can be considered the \enquote{opposite} of a join-semilattice, since in the former, for a pair of points $x, y$, joins exist only for the trivial case of $x \prec y$ or $x \succ y$.
All other posets either have no non-trivial joins at all -- totally ordered sets in particular -- or are in-between these two extremes: Some non-trivial combinations of two points have a join there, while others do not.

\begin{figure}[t]
	\centering
	\begin{subfigure}[b]{0.35\textwidth}
		\centering
		\begin{tikzpicture}
			\begin{scope}[every node/.style={rectangle,thick,draw}]
				\node (A) at (2,1) {$A$};
				\node (B) at (0,2) {$B$};
				\node (C) at (4,2) {$C$};
				\node (D) at (0,4) {$D$};
				\node (E) at (4,4) {$E$};
				\node (F) at (2,5) {$F$};
			\end{scope}
			
			\begin{scope}[every edge/.style={draw=black,thick}]
				\path [-] (A) edge (B);
				\path [-] (B) edge (D);
				\path [-] (A) edge (C);
				\path [-] (C) edge (E);
				\path [-] (E) edge (F);
				
				\path [-] (A) edge (C);
				\path [-] (B) edge (E);
				\path [-] (C) edge (D);
				\path [-] (D) edge (F);
			\end{scope}
		\end{tikzpicture}
		\caption{Poset which is not a lattice}
		\label{fig:hasse-poset}
	\end{subfigure}%
    \hspace{0.05\textwidth}
	\begin{subfigure}[b]{0.5\textwidth}
		\centering
		\begin{tikzpicture}
			\begin{scope}[every node/.style={rectangle,thick,draw}]
				\node (A) at (2,1) {$A$};
				\node (B) at (0,2) {$B$};
				\node (C) at (4,2) {$C$};
				\node (D) at (0,4) {$D$};
				\node (E) at (4,4) {$E$};
				\node (F) at (2,5) {$F$};
				\node (G) at (-2,3) {$G$};
				\node (H) at (2,3) {$H$};
				\node (I) at (6,3) {$I$};
			\end{scope}
			
			\begin{scope}[every edge/.style={draw=black,thick}]
				\path [-] (A) edge (B);
				\path [-] (B) edge (G);
				\path [-] (G) edge (D);
				\path [-] (A) edge (C);
				\path [-] (C) edge (I);
				\path [-] (I) edge (E);
				\path [-] (H) edge (D);
				\path [-] (H) edge (E);
				\path [-] (E) edge (F);
				
				\path [-] (A) edge (C);
				\path [-] (B) edge (H);
				\path [-] (C) edge (H);
				\path [-] (G) edge (D);
				\path [-] (D) edge (F);
			\end{scope}
		\end{tikzpicture}
		\caption{Lattice}
		\label{fig:hasse-lattice}
	\end{subfigure}%
	\caption[Hasse diagrams, representing finite posets: A non-lattice and a lattice.]{
		Finite posets can be depicted using \textit{Hasse diagrams}. Here, two nodes $A$ and $B$ are connected and $B$ is shown above $A$ if $B$ covers $A$ (cf.\ \cref{def:cover}).
		In (a), we see a poset which is not a lattice, since $B \prec D \succ C$ and $B \prec E \succ C$. Therefore, since $D \unord E$, $B$ and $C$ have no least upper bound.
		In (b), we see a poset which forms a lattice. Here, $H = B \lor C$.
	}
	\label{fig:hasse}
\end{figure}

\begin{definition}[Lattice]
	\label{def:lattice}
	A \emph{lattice} is a poset $\TT$ which is both a join- and a meet-semilattice.
\end{definition}

An example of a lattice, in contrast to a poset which does not form a lattice, is depicted in \cref{fig:hasse}.
The theory of (order) lattices constitutes a well-studied subbranch of mathematics, which can be approached both from an order-theoretic (as done here) and from an algebraic perspective.
More precisely, instead of deriving join and meet from a partial order, one may axiomatically consider any two binary operations $\land$ and $\lor$ that are commutative and associative and satisfy the \emph{absorption law}
\begin{equation}
	a \lor (a \land b) = a \land (a \lor b) = a \quad \forall a,b \in \TT \, .
\end{equation}
to give rise to a lattice.
It can be shown that satisfying these algebraic properties guarantees that both join and meet define the same partial order.
Accordingly, a huge amount of further properties a lattice may fulfill are known \cite{Davey2002}. Most of these have not proven directly relevant to this work and will therefore not be reviewed here.

We continue by making a connection of these more generic poset properties with the notion of conicality introduced in \cref{sec:conicality}, and apply them to the particular case of Minkowski space-time.

\begin{lemma}
	\label{thm:join-free}
	Any conical poset (cf.\ \cref{def:conical}) is join-free.
\end{lemma}
\begin{proof}
	We prove this via contraposition.
	Let $\TT$ be a conical poset with $a, b \in \TT$. Assume that $a \unord b$ and that there exists a join $a \lor b \in \TT$, which by definition is unique.
	Then, $f(\{ a \lor b \}) = \bar{J}^+ (a \lor b) = \bar{J}^+ (a) \cap \bar{J}^+ (b) = f (\{a, b \})$.
	As $\spann (\{a, b\}) = \{a, b \}$ due to $a \unord b$, and $\spann(\{a\lor b\})=a\lor b$, we have two distinct spanning sets which admit the same joint future, posing a contradiction to conicality.
\end{proof}

\begin{proposition}[1+1-Minkowski space-time]
	\label{thm:1-plus-1}
	The \emph{light cone structure of 1+1-Minkowski space-time}, having one spatial and one temporal dimension, forms a lattice of its points. The meet of two points is given by the latest point where their past light cones intersect, while their join is given by the earliest point where their future light cones intersect \cite{Mattern2002}.
\end{proposition}

\begin{corollary}[Higher-dim.\ Minkowski space-time]
	\label{thm:no-jsl}
	In the case of $d$ spatial dimensions, we speak of $d$+1-Minkowski space-time.
	For $d \geq 2$, Minkowski space-time (as depicted in \cref{fig:minkowski}) is a join- and meet-free poset.
\end{corollary}
\begin{proof}
	That Minkowski space-time is join-free follows from \cref{thm:join-free}, as it is conical.
	As Minkowski space-time is also invariant under time inversion, it is also meet-free.		
\end{proof}

\medskip

This can also be seen from the fact that there are no unique minimal and maximal points, respectively. Specifically, the intersection of the future of two light cones is larger than any light cone that can be found within the intersection \cite{Mattern2002}.

Nonetheless, causally closed subsets of Minkowski space-time actually assemble a so-called \emph{complete orthocomplemented lattice} for any number of dimensions. The respective sets are sometimes referred to as causal diamonds. Interestingly, this mathematical structure is also exhibited by the set of projectors in Hilbert space, as used in quantum mechanics \cite{Casini2002}\cite[p.~109--115]{Minkowski2010}. Furthermore, a connection between causal order and algebraic quantum field theory has been explored in \cite[Chap.~III.4]{Haag1996}.
Recently, this connection has been studied further in \cite{Castrigiano2024, Kostecki2024}.

\section{A formal motivation for reducibility in the second argument}
\label{sec:red2}

In this section, we provide further details, proofs and examples on the definition of reducibility in the second argument, which in contrast to the first and third argument, is observational rather than interventional.
We begin by repeating the definition below for convenience. 

\redTwoDef*

\redTwo*
\begin{proof}
	Suppose $X \vDash Y | \mathrm{do}(Z)$ holds and is a Red$_2$ affects relation. Then writing out this affects relation along with the first condition implied by the reducibility (while recalling that $s_Y\cup\tilde{s}_Y=Y$), we have 
	\begin{align}
        \label{eq: red2-1}
		P_{\cG_{\mathrm{do}(XZ)}}(Y|XZ) &\neq P_{\cG_{\mathrm{do}(Z)}}(Y|Z) \quad \text{and}
        \\
        \label{eq: red2-2}
		P_{\cG_{\mathrm{do}(XZ)}}(s_Y|\tilde{s}_YXZ) &= P_{\cG_{\mathrm{do}(Z)}}(s_Y|\tilde{s}_Y Z)
    \end{align}
    as well as 
    \begin{equation}
        \label{eq:red2-indep}
        \begin{split}
            P_{\cG_{\mathrm{do}(XZ)}}(Y|XZ)
            &= P_{\cG_{\mathrm{do}(XZ)}}(s_Y|XZ) P_{\cG_{\mathrm{do}(XZ)}}(\tilde{s}_Y|XZ) \quad\ \text{or} \quad\\
		      P_{\cG_{\mathrm{do}(Z)}}(Y|Z)
            &= P_{\cG_{\mathrm{do}(Z)}}(s_Y|Z) P_{\cG_{\mathrm{do}(Z)}}(\tilde{s}_Y|Z)\,
        \end{split}
    \end{equation}
    from the second condition of Red$_2$.
    The latter condition in particular is equivalent to
    \begin{equation}
        \label{eq:red2-indep-rewritten}
        P_{\cG_{\mathrm{do}(XZ)}}(s_Y|XZ)
        = P_{\cG_{\mathrm{do}(XZ)}}(s_Y|\tilde{s}_Y XZ) \qquad\ \text{or} \qquad
        P_{\cG_{\mathrm{do}(Z)}}(s_Y|Z)
        = P_{\cG_{\mathrm{do}(Z)}}(s_Y|\tilde{s}_Y Z)\,.
    \end{equation}
    If the first of these alternative equalities holds, we can rewrite the left-hand side of \cref{eq: red2-2} to obtain
    \begin{equation}
        \label{eq:combined1}
        P_{\cG_{\mathrm{do}(XZ)}}(s_Y|XZ) = P_{\cG_{\mathrm{do}(Z)}}(s_Y|\tilde{s}_Y Z) \,.
    \end{equation}
    Fombining \cref{eq: red2-1} and the first condition of \cref{eq:red2-indep} (equivalent to the first condition of \cref{eq:red2-indep-rewritten}), we additionally obtain
    $P_{\cG_{\mathrm{do}(XZ)}}(s_Y|XZ) P_{\cG_{\mathrm{do}(XZ)}}(\tilde{s}_Y|XZ) \neq P_{\cG_{\mathrm{do}(Z)}}(Y|Z)$.
    Using that generally, $P_{\cG_{\mathrm{do}(Z)}}(Y|Z) = P_{\cG_{\mathrm{do}(Z)}}(s_Y|\tilde{s}_Y Z) P_{\cG_{\mathrm{do}(Z)}}(\tilde{s}_Y | Z)$, we have
    \begin{equation}
        P_{\cG_{\mathrm{do}(XZ)}}(s_Y|XZ) P_{\cG_{\mathrm{do}(XZ)}}(\tilde{s}_Y|XZ) \neq
        P_{\cG_{\mathrm{do}(Z)}}(s_Y|\tilde{s}_Y Z) P_{\cG_{\mathrm{do}(Z)}}(\tilde{s}_Y | Z) \,.
    \end{equation}
    Using \cref{eq:combined1}, we can cancel out a term in this product. This yields $P_{\cG_{\mathrm{do}(XZ)}}(\tilde{s}_Y | XZ) \neq P_{\cG_{\mathrm{do}(Z)}}(\tilde{s}_Y|Z)$, which is equivalent to $X \affects \tilde{s}_Y \given \doo(Z)$, as desired.

    \medskip
    If the second of the alternative conditions in \cref{eq:red2-indep-rewritten} holds, we can rewrite the right-hand side of \cref{eq: red2-2} to  obtain
    \begin{equation}
        \label{eq:combined2}
        P_{\cG_{\mathrm{do}(XZ)}}(s_Y|\tilde{s}_YXZ) = P_{\cG_{\mathrm{do}(Z)}}(s_Y|Z) \,.
    \end{equation}
    Combining \cref{eq: red2-1} and the second condition of \cref{eq:red2-indep} (equivalent to the second condition of \cref{eq:red2-indep-rewritten}), we additionally obtain
    $P_{\cG_{\mathrm{do}(XZ)}}(Y|XZ) \neq P_{\cG_{\mathrm{do}(XZ)}}(s_Y|Z) P_{\cG_{\mathrm{do}(Z)}}(\tilde{s}_Y|Z)$.
    Using that generally, $P_{\cG_{\mathrm{do}(XZ)}}(Y|XZ) = P_{\cG_{\mathrm{do}(XZ)}}(s_Y|\tilde{s}_Y XZ) P_{\cG_{\mathrm{do}(XZ)}}(\tilde{s}_Y | XZ)$, we have
    \begin{equation}
        P_{\cG_{\mathrm{do}(XZ)}}(s_Y|\tilde{s}_Y XZ) P_{\cG_{\mathrm{do}(XZ)}}(\tilde{s}_Y | XZ) \neq
        P_{\cG_{\mathrm{do}(XZ)}}(s_Y|Z) P_{\cG_{\mathrm{do}(Z)}}(\tilde{s}_Y|Z) \,.
    \end{equation}
    Using \cref{eq:combined2}, we can cancel out a term in this product, again yielding the desired conclusion, $X \affects \tilde{s}_Y \given \doo(Z)$.
\end{proof}

\medskip

{\bf Operational motivation for the definition} Consider again the example mentioned in the main text (\cref{sec:affects-red}), in relation to Red$_2$, where we have a causal structure with $X\longrsquigarrow Y_1$ while $Y_2$ is a node with no in or outgoing arrows. Clearly $X \affects Y_1 Y_2$ must be reducible to $X \affects Y_1$ here, according to any reasonable notion of Red$_2$, as $Y_2$ is entirely redundant in this scenario. Operationally, 
we want the definition to capture that the original affects relation $X \affects Y_1 Y_2$ and the reduced relation $X \affects Y_1$ carry the same information.
As mentioned in the main text, in this case, this operational intuition cannot be captured mathematically by simply equating the left and right hand sides of the expressions for the two affects relations.
Instead, the redundancy of $Y_2$ in $X \affects Y_1 Y_2$ is captured through conditional independences.
\begin{itemize}
    \item $(Y_1\upmodels Y_2)_{\cG}$, i.e.\ $P_{\cG}(Y_1Y_2) = P_{\cG}(Y_1) P_{\cG}(Y_2)$
    \item $(Y_1\upmodels Y_2|X)_{\cG_{\mathrm{do}(X)}}$, i.e.\ $P_{\cG_{\mathrm{do}(X)}}(Y_1Y_2|X) = P_{\cG_{\mathrm{do}(X)}}(Y_1|X) P_{\cG_{\mathrm{do}(X)}}(Y_2|X)$
\end{itemize}

While both the above independences are satisfied in our simple example, we show in the proof of \cref{lemma: red2} that even if one of these independences is satisfied, along with $X\not\vDash Y_2 \given Y_1$ (this is the first condition of \cref{def: red2} applied to this example, which is also satisfied here), then the two affects relations $X \affects Y_1 Y_2$ and $X \affects Y_1$ are equivalent, in the sense that the expressions corresponding to those two affects relations are the same and thus capture the same information. Since satisfying one of these conditional independences is the second condition of \cref{def: red2}, for the case of this example, this establishes that under the defining conditions of Red$_2$, the original and reduced affects relations are indeed equivalent as we require. More generally, this holds for any general higher-order affects relation $X\affects Y\given \doo(Z)$, if the conditions of \cref{def: red2} are satisfied for some $s_Y\subsetneq Y$, then through similar arguments as we present below, it can be shown that the original affects relation is equivalent to the reduced relation $X\affects \tilde{s}_Y\given \doo(Z)$.

Some examples are as follows, these examples also illustrate that all of the conditions of \cref{def: red2} can be violated independently, resulting in irreducible affects relations.
\begin{example}{(Violating the first Red$_2$ condition)}
    \label{eg: red_2_1}
To illustrate that the first condition of \cref{def: red2} can be violated independently, consider the causal structure of \cref{fig: eg_red2_1} with the causal model:
$E_1$ is non-uniform, $E_2$ is uniform, $X=E_1\oplus E_2$, $Y_2=X\oplus E_2$ and $Y_1=X$. We notice that $P_{\cG}(Y_1)$ is uniform (since $Y_1=X$ but $X=E_1\oplus E_2$ and $E_2$ is uniform) and $P_{\cG}(Y_2)=P_{\cG}(E_1)$ is non-uniform (since $Y_2=X\oplus E_2$ and $X=E_1\oplus E_2$ gives $Y_2=E_1$).
In $\cG_{\mathrm{do}(X)}$, we only have the dependences: $Y_2=X\oplus E_2$ and $Y_1=X$ and we can see that $P_{\cG_{\mathrm{do}(X)}}(Y_1Y_2|X)$ will be different from $P_{\cG}(Y_1Y_2)$ (in particular since $Y_1$ is fully determined by $X$) hence $X\vDash Y_1Y_2$ holds.
Moreover, $(Y_1\upmodels Y_2)_{\cG}$. We also see that $(Y_1\upmodels Y_2|X)_{\cG_{\mathrm{do}(X)}}$ holds, as this follows from the d-separation $(Y_1\perp^d Y_2|X)_{\cG_{\mathrm{do}(X)}}$ in the causal structure.
Noting that the only non-empty strict subsets $s_Y$ in this case are $\{Y_1\}$ and $\{Y_2\}$, the second condition is satisfied.
Finally, notice that $P_{\cG_{\mathrm{do}(X)}}(Y_1|X)$ is deterministic and $P_{\cG_{\mathrm{do}(X)}}(Y_2|X)$ is uniform (since $Y_2=X\oplus E_2$ with $X$ independent of $E_2$ and $E_2$ uniform in $\cG_{\mathrm{do}(X)}$). Therefore we have $X\vDash Y_1$ and $X\vDash Y_2$ which in particular imply that $X \vDash Y_1 | Y_2$ and $X \vDash Y_2 | Y_1$, as additional conditioning preserves an affects relation if $Y_1$ and $Y_2$ are conditionally independent both in the pre- and post-intervention causal structure.
Therefore, the first condition is violated. In summary, $X\vDash Y_1Y_2$ is Irred$_2$ in this example, it violates the first and satisfies the second two conditions of \cref{def: red2}.
\end{example}

\begin{example}{(Violating the second Red$_2$ condition)}
    \label{eg: red_2_3}
To illustrate that the second condition of \cref{def: red2} can be violated independently, consider the causal structure $Y_1\longlsquigarrow X \longlsquigarrow E\longrsquigarrow Y_2$ and the causal model where $Y_1=X$, $X=E$, $Y_2=E$ and any distribution over $E$. Clearly $X\vDash Y_1$, $X\vDash Y_1Y_2$. We also have $X\not\vDash Y_2$ and $X\not\vDash Y_2 \given Y_1$ due to d-separation and therefore the first condition of \cref{def: red2} is satisfied for the affects relation $X\vDash Y_1Y_2$. Moreover $(Y_1\upmodels Y_2|X)_{\cG_{\mathrm{do}(X)}}$ holds due to the d-separation $(Y_1\perp^d Y_2|X)_{\cG_{\mathrm{do}(X)}}$ which means that the first alternative second condition would also be satisfied. However, the second alternative is violated as we clearly have $(Y_1\not\upmodels Y_2)_{\cG}$.
Therefore, $X \affects Y_1 Y_2$ is Red$_2$.
\end{example}
\begin{figure}
    \centering
    \begin{subfigure}[b]{0.4\textwidth}
        \centering
        \begin{tikzpicture}[scale=0.9]
            \node[shape=circle,draw=black] (X) at (0,0) {$X$};
            \node[shape=circle,draw=black] (Y1) at (-2,2) {$Y_1$};
            \node[shape=circle,draw=black] (Y2) at (2,2) {$Y_2$};
            \node[shape=circle,draw=black] (E2) at (2,-2) {$E_2$};
            \node[shape=circle,draw=black] (E1) at (-2,-2) {$E_1$};
            
            \draw[vvarrow, arrows={-Stealth}] (E1) -- (X);
            \draw[vvarrow, arrows={-Stealth}] (E2) -- (X);
            \draw[vvarrow, arrows={-Stealth}] (E2) -- (Y2);
            \draw[vvarrow, arrows={-Stealth}] (X) -- (Y1);
            \draw[vvarrow, arrows={-Stealth}] (X) -- (Y2);
        \end{tikzpicture}
        \caption{Causal structure for \cref{eg: red_2_1}.}
        \label{fig: eg_red2_1}
    \end{subfigure}%
    \hspace{0.1\textwidth}    
    \begin{subfigure}[b]{0.4\textwidth}
        \centering
        \begin{tikzpicture}[scale=0.9]
            \begin{scope}[every node/.style={circle,thick,draw,inner sep=0pt,minimum size=0.8cm}]
                \node (X) at (0,-2) {$X$};
                \node (Y) at (2,0) {$Y$};
                \node (Z) at (0,2) {$Z$};
            \end{scope}

            \begin{scope}[>={Stealth[black]},
                          every node/.style={fill=white,circle},
                          every edge/.style=vvarrow]
                \path [->] (X) edge (Y);
                \path [->] (X) edge (Z);
                \path [->] (Y) edge (Z);
            \end{scope}
        \end{tikzpicture}
        \caption{Causal structure for \cref{ex:clus3}.}
        \label{fig:clus3}
    \end{subfigure}%
    \caption{Causal structures for various examples.}
\end{figure}

\section{Further properties of clustering}
\label{sec:clus-prop}

In this section we present interesting properties of clustering in the first three arguments of an affects relation, which may be of independent interest, but are not pivotal for our main results.
First, we show that the presence of an affects relation that is clustered in its first or third argument is incompatible with the presence of a wide variety of affects relations.

\begin{restatable}{lemma}{clusOne}
    \label{thm:clus1}
    For any affects relation $X\vDash Y|\mathrm{do}(Z)$ which is Clus$_1$,
    $\tilde{s}^1_X \not\vDash Y|\mathrm{do}(\tilde{s}^2_X Z)$ for all $\tilde{s}^1_X, \tilde{s}^2_X \subsetneq X$ disjoint with $\tilde{s}^1_X \tilde{s}^2_X \neq X$.
\end{restatable}
\begin{proof}
    Due to $X\vDash Y|\mathrm{do}(Z)$ being Clus$_1$, $\tilde{s}_X \naffects Y|\mathrm{do}(Z)$ for all $\tilde{s}_X \subsetneq X$. Hence,
    \begin{align}
      P_{\cG_{\mathrm{do}(Z)}}(Y|Z)&=  P_{\cG_{\mathrm{do}(\tilde{s}_XZ)}}(Y|\tilde{s}_XZ), \quad \forall \tilde{s}_X\subsetneq X \\
      \implies
      P_{\cG_{\mathrm{do}(\tilde{s}^1_XZ)}}(Y|\tilde{s}^1_XZ) &= P_{\cG_{\mathrm{do}(\tilde{s}^2_XZ)}}(Y|\tilde{s}^2_XZ), \quad \forall \tilde{s}^1_X, \tilde{s}_X^2 \subsetneq X \\
      \implies
      P_{\cG_{\mathrm{do}(\tilde{s}_XZ)}}(Y|\tilde{s}^1_XZ) &= P_{\cG_{\mathrm{do}(\tilde{s}^1_X \tilde{s}^2_X Z)}}(Y|\tilde{s}^2_X \tilde{s}^1_X Z), \quad \forall \tilde{s}^2_X \subseteq \tilde{s}^1_X\tilde{s}^2_X \subsetneq X \,.
    \end{align}
    However, this line is equivalent to stating $\tilde{s}^1_X \not\vDash Y|\mathrm{do}(\tilde{s}^2_X Z)$ for all $\tilde{s}^1_X, \tilde{s}^2_X$ disjoint with $\tilde{s}^1_X \tilde{s}^2_X \neq X$.
\end{proof}

\begin{restatable}{lemma}{clusThree}
    For any affects relation $X\vDash Y|\mathrm{do}(Z)$ which is Clus$_3$,
    $\tilde{s}^1_Z \not\vDash Y|\mathrm{do}(\tilde{s}^2_Z)$ for all $\tilde{s}^1_Z, \tilde{s}^2_Z \subsetneq Z$ disjoint with $\tilde{s}^1_Z \tilde{s}^2_Z \neq Z$.
\end{restatable}
\begin{proof}
    Due to $X\vDash Y|\mathrm{do}(Z)$ being Clus$_3$, $X \naffects Y|\mathrm{do}(Xs_Z)$ for all $\tilde{s}_Z \subsetneq Z$. Hence,
    \begin{align}
      P_{\cG_{\mathrm{do}(Z)}}(Y|XZ)&=  P_{\cG_{\mathrm{do}(\tilde{s}_ZX)}}(Y|X \tilde{s}_Z W), \quad \forall \tilde{s}_Z\subsetneq Z \\
      \implies
      P_{\cG_{\mathrm{do}(\tilde{s}^1_Z X)}}(Y|X\tilde{s}^1_Z) &= P_{\cG_{\mathrm{do}(\tilde{s}^2_Z X)}}(Y|X\tilde{s}^2_Z), \quad \forall \tilde{s}^2_z, \tilde{s}^1_Z \subsetneq Z \\
      \implies
      P_{\cG_{\mathrm{do}(\tilde{s}^1_ZX)}}(Y|X\tilde{s}^1_Z) &= P_{\cG_{\mathrm{do}(\tilde{s}^1_Z \tilde{s}^2_Z ZX)}}(Y|X \tilde{s}^2_Z \tilde{s}^1_Z X), \quad \forall \tilde{s}^2_X \subseteq \tilde{s}^1_Z\tilde{s}^2_Z \subsetneq Z\,.
    \end{align}
    However, this line is equivalent to stating $\tilde{s}^1_Z \not\vDash Y|\mathrm{do}(\tilde{s}^2_Z X)$ for all $\tilde{s}^1_Z, \tilde{s}^2_Z$ disjoint with $\tilde{s}^1_Z \tilde{s}^2_Z \neq Z$.
\end{proof}

\medskip
Second, we demonstrate that for a causal model to admit affects relations that are Clus$_1$, it is necessary but not sufficient for the model to also admit affects relations that are Clus$_3$.

\begin{restatable}{lemma}{clusThreeOne}
\label{thm:clusThreeOne}
    For any affects relation $X\vDash Y|\mathrm{do}(Z)$ which is Clus$_1$,
    $s_X\vDash Y|\mathrm{do}(\tilde{s}_X Z)$ is Clus$_3$ for each partition $s_X \tilde{s}_X = X$.
\end{restatable}
\begin{proof}
    By \cref{theorem: clus_irred}, $X\vDash Y|\mathrm{do}(Z)$ is Irred$_1$.
    Hence, for all $s_X \tilde{s}_X = X$, we have $s_X\vDash Y|\mathrm{do}(\tilde{s}_X Z)$. Due to \cref{thm:clus1}, we also have $s_X \not\vDash Y|\mathrm{do}(\tilde{s}'_X Z)$, where $\tilde{s}'_X \subsetneq \tilde{s}_X$.
    This is equivalent for Clus$_3$ for $s_X\vDash Y|\mathrm{do}(\tilde{s}_X Z)$.
\end{proof}

\medskip
By contrast, we can construct an example which has clustering only in the third argument. 
\begin{example}{(Clus$_3$ only)}
    \label{ex:clus3}
    Consider a causal model over $S = \{ X, Y, Z \}$, where $X$ is exogenous and non-uniform, $Y = X$ and $Z = Y \oplus X$.
    The respective causal structure is shown in \cref{fig:clus3}.
    Then pre-intervention, $Z$ is deterministically $0$.
    Clearly, this causal model is fine-tuned with regard to the value of $Z$:
    No matter which intervention we choose on $X$, this does not change, hence $X \naffects Z$, yet clearly $X \affects Y$ and also $Y \affects Z$, as $Z$ is no longer uniform after an intervention.
    Accordingly, since we have only three observable nodes in this model, any clustered affects relation must have $X$ in its first and $Z$ in its second argument, while not having $Y$ in its first and second argument.
    Hence, $X \affects Z \given \doo(Y)$, clustered in its third argument, is the only clustered affects relation in this example. Finally, we would like to point out that in this causal model, affects relations are not transitive: We have $X \affects Y$ and $Y \affects Z$, yet $X \naffects Z$.
\end{example}

Finally, we demonstrate that in absence of clustering in a certain argument, a causal model must allow each affects relation to be reduced to minimal cardinality for the respective argument.

\begin{lemma}[$\neg$ Clus$_1$]
    \label{thm:not-clus1}
    Consider an affects relation $X \affects Y \given \doo(Z)$ in a causal model which does not yield any affects relation that are Clus$_1$. Then there exists $e_X \in X$ such that $e_X \affects Y \given \doo(Z)$.
\end{lemma}
\begin{proof}
    As $X \affects Y \given \doo(Z)$ is not Clus$_1$, there exists a subset $s_X \subsetneq X$ such that $s_X \affects Y \given \doo(Z)$. However, as this is again not Clus$_1$, we can repeat this inductively until we reach $e_X \in Y$ with $e_X \affects Y \given \doo(Z)$.
\end{proof}

\begin{lemma}[$\neg$ Clus$_2$]
    \label{thm:not-clus2}
    Consider an affects relation $X \affects Y \given \doo(Z)$ in a causal model which does not yield any affects relation that are Clus$_2$. Then there exists $e_Y \in Y$ such that $X \affects e_Y \given \doo(Z)$.
\end{lemma}
\begin{proof}
    Directly analogous to $\neg$ Clus$_1$, but for the second argument.
\end{proof}

\begin{lemma}[$\neg$ Clus$_3$]
    \label{thm:not-clus3}
    Consider an affects relation $X \affects Y \given \doo(Z)$ in a causal model which does not yield any affects relation that are Clus$_3$. Then there exists $e_X \in X$ such that $e_X \affects Y$.
\end{lemma}
\begin{proof}
    As $X \affects Y \given \doo(Z)$ is not Clus$_3$, there exists a subset $s_Z \subsetneq Z$ such that $X \affects Y \given \doo(s_Z)$. However, as this is again not Clus$_3$, we can repeat this inductively until we reach the 0$^{th}$-order affects relation $X \affects Y$.
    Additionally, by \cref{thm:clusThreeOne}, the absence of any affects relations with Clus$_1$ implies the same for Clus$_3$. Therefore, we can use \cref{thm:not-clus1} to deduce the claim.
\end{proof}

\medskip
Considering that by \cref{lemma: clus_finetune}, the absence of clustering is necessary for faithfulness, we can summarize these results accordingly.

\begin{corollary}
    \label{thm:cause-no-clus}
    Consider an affects relation $X \affects Y \given \doo(Z)$ in a causal model without clustering.
    Then there exist $e_X \in X, e_Y \in Y$ such that $e_X \affects e_Y$.
    In particular, this holds for any faithful causal model.
\end{corollary}

Recollecting \cref{ex:sorkin}, this indicates that in absence of clustering, any fine-tuning detectable from affects relations alone is possible to deduce via non-transitive affects relations.

\section{Conditional higher-order affects relations}
\label{sec:conditional}

Within this section, we generalize results from the main text to the case of conditional affects relations, and provide some additional useful transformation rules which are not relevant to the main results of this work. To start off, we provide a formal definition of conditional affects relations.
\begin{definition}[Conditional (Higher-Order) Affects Relations \cite{VVC}]
    \label{def:affects-cond}
    Consider a causal model over a set of $S$ observed nodes, associated with a causal structure $\mathcal{G}$.
    For pairwise disjoint subsets $X, Y, Z, W \subset S$, with $X, Y$ non-empty, we say
    \begin{equation}
        X \, \text{affects} \; Y \, \text{given} \, \doo(Z), W \, ,
    \end{equation}
    which we alternatively denote as
    \begin{equation}
        X \vDash Y \,|\, \doo(Z), W \, ,
    \end{equation}
    if there exist values $x$ of $X$, $z$ of $Z$ and $w$ of $W$  such that
    \begin{equation}
        P_{\mathcal{G}_{\doo(XZ)}} (Y | X=x, Z=z, W=w) \neq
        P_{\mathcal{G}_{\doo(Z)}} (Y | Z=z, W=w)
    \end{equation}
    For $W \neq \emptyset$, we speak of a \emph{conditional affects relation}, denoted by $X \affects Y \given W$.
    If $Z \neq \emptyset$, we have a \emph{higher-order (HO) affects relation}, denoted by $X \affects Y \given \doo(Z)$.
    More specifically, it is also called a \emph{$\abs{Z}^\text{th}$-order affects relation}.
    The trivial case of $W = Z = \emptyset$ is called an \emph{unconditional 0$^\text{th}$-order affects relation}, denoted by $X \affects Y$.
\end{definition}

\subsection{Properties}
We begin by restating Lemma IV.8 of \cite{VVC}, which makes clear how conditional and unconditional affects relations relate with one another.
\begin{lemma}
    \label{thm:decondition}
    For a causal model over a set $S$ of RVs, where $X, Y, Z, W \subset S$ disjoint,
    \begin{enumerate}
        \item $X \affects Y \given \doo(Z), W \implies X \affects YW \given \doo(Z)$.
        \item $X \affects Y \given \doo(Z), W$ is Irred$_1 \implies X \affects YW \given \doo(Z)$ is Irred$_1$.
        \item $X \affects YW \given \doo(Z) \iff X \affects Y \given \doo(Z), W \ \text{or} \ X \affects W \given \doo(Z)$.
    \end{enumerate}
\end{lemma}
By the same argument as for \cref{thm:decondition}.2., one can also derive:
\begin{corollary}
    \label{thm:decondition3}
     For a causal model over a set $S$ of RVs, where $X, Y, Z, W \subset S$ disjoint,\\
     $X \affects Y \given \doo(Z), W$ is Irred$_3 \implies X \affects YW \given \doo(Z)$ is Irred$_3$.
\end{corollary}

Additionally, it is useful to restate Theorem IV.1 of \cite{VVC}, that shows that Pearl's three rules of do-calculus \cite{Pearl2009} also hold in the general affects framework and can be derived from the d-separation property alone.

\begin{restatable}{theorem}{DoRules}
\label{theorem:dorules}
Given a causal model over a set $S$ of observed nodes, an associated causal graph $\cG$ and a distribution $P_S$ satisfying the d-separation property relative to $\cG$ (\cref{eq:compat}), the following 3 rules of do-calculus~\cite{Pearl2009} hold for interventions on this causal model.
\begin{itemize}
    \item \textbf{Rule 1: Ignoring observations} 
    \begin{equation}
    \label{eq: rule1}
        P_{\cG_{\mathrm{do}(X)}}(y|x,z,w)=P_{\cG_{\mathrm{do}(X)}}(y|x,w) \qquad \text{if } (Y\perp^d Z|XW)_{\cG_{\overline{X}}}
    \end{equation}
     \item \textbf{Rule 2: Action/observation exchange} 
    \begin{equation}
     \label{eq: rule2}
        P_{\cG_{\mathrm{do}(X Z)}}(y|x,z,w)=P_{\cG_{\mathrm{do}(X)}}(y|x,z,w) \qquad \text{if } (Y\perp^d Z|XW)_{\cG_{\overline{X}\underline{Z}}}
    \end{equation}
    \item \textbf{Rule 3: Ignoring actions/interventions} 
     \begin{equation}
      \label{eq: rule3}
        P_{\cG_{\mathrm{do}(X Z)}}(y|x,z,w)=P_{\cG_{\mathrm{do}(X)}}(y|x,w) \qquad \text{if } (Y\perp^d Z|XW)_{\cG_{\overline{XZ(W)}}},
    \end{equation}
\end{itemize}
where $X$, $Y$, $Z$ and $W$ are disjoint subsets of the observed nodes, $Z(W)$ denotes the set of nodes in $Z$ that are not ancestors of $W$, and the above hold for all values $x$, $y$, $z$ and $w$ of $X$, $Y$, $Z$ and $W$. 
\end{restatable}

We continue with a general result about conditional affects relations that will be useful, this is essentially a generalization of \cref{lemma: aff_corr} to the conditional case. The proof method is exactly analogous to that of \cref{lemma: aff_corr}, but we repeat it for completeness.

\begin{restatable}{lemma}{affectsCorrCond}
\label{lemma: aff_corr_cond}
    For a causal model over a set $S$ of RVs, where $X, Y, Z, W \subset S$ disjoint,\\
    $X\not\vDash Y| \mathrm{do}(Z),W \implies (X\upmodels Y|ZW)_{\cG_{\mathrm{do}}(XZ)}$.
\end{restatable}

\begin{proof}
    Suppose by contradiction that $X\not\vDash Y|\{\mathrm{do}(Z), W\}$ (which amounts to $P_{\cG_{\mathrm{do}(XZ)}}(Y|XZW)=  P_{\cG_{\mathrm{do}(Z)}}(Y|ZW)$) and $(X\not\upmodels Y|ZW)_{\cG_{\mathrm{do}}(XZ)}$. The latter is equivalent to saying that there exist values $z$ of $Z$, $w$ of $W$ and distinct values  $x$ and $x'$ of $X$ such that $P_{\cG_{\mathrm{do}(XZ)}}(Y|X=x,Z=z, W=w)\neq P_{\cG_{\mathrm{do}(XZ)}}(Y|X=x',Z=z, W=w)$. However, this implies that $P_{\cG_{\mathrm{do}(XZ)}}(Y|XZ)$ does have a non-trivial dependence on $X$.
    As this contradicts $P_{\cG_{\mathrm{do}(XZ)}}(Y|XZW)=  P_{\cG_{\mathrm{do}(Z)}}(Y|ZW)$, which would imply that it is equal to an $X$-independent quantity, this proves the result. 
\end{proof}

\begin{restatable}{lemma}{dsepAffectsCond}
\label{lemma: dsep_aff_cond}
    For a causal model over a set $S$ of RVs, where $X, Y, Z, W \subset S$ disjoint,
    $$(X\perp^d Y|ZW)_{\cG_{\mathrm{do}(X(W)Z)}} \implies X\not\vDash Y|\{\mathrm{do}(Z),W \},$$
    with $X(W):=X\backslash (X \cap \mathrm{anc}(W))$. where $\mathrm{anc}(W)$ is the set of all ancestors of $W$ in the original graph $\cG$.
\end{restatable}

The above lemma follows immediately from Pearl's third rule of do-calculus, i.e. Rule 3 of \cref{theorem:dorules}. 

Notice that between the statements of \cref{lemma: aff_corr_cond} and \cref{lemma: dsep_aff_cond}, there is a subtle difference regarding the relation between $X$ and $W$ in the causal structure. Specifically, the requirement in \cref{lemma: dsep_aff_cond} is satisfied whenever $W=\emptyset$  and whenever $X$ is not a cause of $W$.
This impacts the generalization of some of the statements from the main text for unconditional affects relations ($W=\emptyset$) to conditional relations.
For the case that both $X$ is not a cause of $W$ (and thus, $X = X(W)$) and the causal model is faithful (hence, we have $(X\upmodels Y|ZW)_{\cG_{\mathrm{do}}(XZ)}\iff  (X \perp^d Y|ZW)_{\cG_{\mathrm{do}}(XZ)}$), then \cref{lemma: aff_corr_cond} and \cref{lemma: dsep_aff_cond} taken together show that,
\begin{equation}
	X\not\vDash Y|\{\mathrm{do}(Z),W \} \iff
	(X\upmodels Y|ZW)_{\cG_{\mathrm{do}}(XZ)} \iff
    (X \perp^d Y|ZW)_{\cG_{\mathrm{do}}(XZ)},
\end{equation}
while for unfaithful models, conditional independence (middle) need not imply absence of affects relations (left) or d-separation (right).

\subsection{Generalizing reducibility and clustering}
\label{sec:affects-conditional}

In this section, we present the generalization of the results of \cref{sec:affects-new}, originally presented for unconditional affects relations, to the case of conditional affects relations. While some of the results fully generalize, the others require an additional condition when introducing a non-trivial fourth argument $W\neq \emptyset$ of the conditional affects relation, owing to the condition appearing in \cref{lemma: dsep_aff_cond}. 

We begin by defining the reducibility concepts introduced in \cref{sec:affects-red} for the first three arguments of unconditional affects relations $X \affects Y \given \doo(Z)$ to all four arguments of conditional affects relations $X \vDash Y | \mathrm{do}(Z),W$. For the first three arguments, this is the same as \cref{def: red1}, \cref{def: red2} and \cref{def: red3} but with the inclusion of $W$, but we repeat them here for completeness.

\begin{definition}[Reducibility in the first argument \cite{VVC}]
 We say that an affects relation $X \vDash Y | \mathrm{do}(Z),W$  is reducible in the first argument (or Red$_1$) if there exists a non-empty subset $s_X\subsetneq X$ such that $s_X \not\vDash Y | \mathrm{do}(\tilde{s}_XZ),W$, where $\tilde{s}_X:=X\backslash s_X$. Otherwise, we say that $X \vDash Y | \mathrm{do}(Z),W$ is irreducible in the first argument and denote it as Irred$_1$.
\end{definition}

\begin{definition}[Reducibility in the second argument]
 We say that an affects relation $X \vDash Y | \mathrm{do}(Z),W$  is reducible in the second argument (or Red$_2$) if there exists a non-empty subset $s_Y\subsetneq Y$ such that both the following conditions hold, where $\tilde{s}_Y:=Y\backslash s_Y$.
 \begin{itemize}
     \item $X \not\vDash s_Y | \mathrm{do}(Z), \tilde{s}_Y W$
     \item $(s_Y\upmodels \tilde{s}_Y|XZW)_{\cG_{\mathrm{do}(XZ)}}$ or $(s_Y\upmodels \tilde{s}_Y|ZW)_{\cG_{\mathrm{do}(Z)}}$
 \end{itemize}
Otherwise, we say that $X \vDash Y | \mathrm{do}(Z),W$ is irreducible in the second argument and denote it as Irred$_2$.
\end{definition}

\begin{definition}[Reducibility in the third argument]
 We say that an affects relation $X \vDash Y | \mathrm{do}(Z),W$  is reducible in the third argument (or Red$_3$) if there exists a non-empty subset $s_Z\subseteq Z$ such that both the following conditions hold, where $\tilde{s}_Z:=Z\backslash s_Z$
 \begin{itemize}
     \item  $s_Z \not\vDash Y | \mathrm{do}(X\tilde{s}_Z),W$
     \item $s_Z \not\vDash Y | \mathrm{do}(\tilde{s}_Z),W$
 \end{itemize}
Otherwise, we say that $X \vDash Y | \mathrm{do}(Z),W$ is irreducible in the third argument and denote it as Irred$_3$.
\end{definition}

\begin{definition}[Reducibility in the fourth argument]
	\label{def: red4}
	 We say that an affects relation $X \vDash Y | \mathrm{do}(Z),W$  is reducible in the fourth argument (or Red$_4$) if there exists a non-empty subset $s_W\subseteq W$ such that all the following conditions hold, where $\tilde{s}_W:=W\backslash s_W$.
	 \begin{itemize}
		 \item  $(Y\upmodels s_W|XZ\tilde{s}_W)_{\cG_{\mathrm{do}(XZ)}}$
		 \item $(Y\upmodels s_W|Z\tilde{s}_W)_{\cG_{\mathrm{do}(XZ)}}$
	 \end{itemize}
	Otherwise, we say that $X \vDash Y | \mathrm{do}(Z),W$ is irreducible in the fourth argument and denote it as Irred$_4$.
\end{definition}

It is then easy to check that the following generalizations of \cref{lemma: red1}, \cref{lemma: red2} and \cref{lemma: red3} for the first three arguments follow from the same proof method (the first of these is proven in \cite{VVC}).

\begin{lemma}
    If $X \vDash Y | \mathrm{do}(Z),W$ is a Red$_1$ affects relation, then there exists $\tilde{s}_X\subsetneq X$ such that  $\tilde{s}_X \vDash Y | \mathrm{do}(Z),W$ holds.    
\end{lemma}

\begin{lemma}
    If $X \vDash Y | \mathrm{do}(Z),W$ is a Red$_2$ affects relation, then there exists $\tilde{s}_Y\subsetneq Y$ such that  $X \vDash \tilde{s}_Y | \mathrm{do}(Z),W$ holds. 
\end{lemma}

\begin{lemma}
       If $X \vDash Y | \mathrm{do}(Z),W$ is a Red$_3$ affects relation, then there exists $\tilde{s}_Z\subsetneq Z$ such that  $X \vDash Y | \mathrm{do}(\tilde{s}_Z),W$ holds. 
\end{lemma}

For the fourth argument, we have the following analogous lemma.
	\begin{restatable}{lemma}{redFour}
		If $X \vDash Y | \mathrm{do}(Z),W$ is a Red$_4$ affects relation, then there exists $\tilde{s}_W\subsetneq W$ such that  $X \vDash Y | \mathrm{do}(Z),\tilde{s}_W$ holds. 
	\end{restatable}
\begin{proof}

 Suppose $X \vDash Y | \mathrm{do}(Z),W$ holds and is a Red$_4$ affects relation. Then writing out this affects relation along with the two conditions implied by the reducibility (while recalling that $s_W\cup\tilde{s}_W=W$), we have 
	 \begin{align}
	 \label{eq: red4}
		 \begin{split}
		  P_{\cG_{\mathrm{do}(XZ)}}(Y|XZW)&\neq  P_{\cG_{\mathrm{do}(Z)}}(Y|ZW)\\
		  P_{\cG_{\mathrm{do}(XZ)}}(Ys_W|XZ\tilde{s}_W)&= P_{\cG_{\mathrm{do}(XZ)}}(Y|XZ\tilde{s}_W)P_{\cG_{\mathrm{do}(XZ)}}(s_W|XZ\tilde{s}_W) \\
		P_{\cG_{\mathrm{do}(XZ)}}(Ys_W|Z\tilde{s}_W)&= P_{\cG_{\mathrm{do}(XZ)}}(Y|Z\tilde{s}_W)P_{\cG_{\mathrm{do}(XZ)}}(s_W|Z\tilde{s}_W).
		 \end{split}
	 \end{align}
	 The last two equalities are equivalent to the following.
   
	 \begin{align}
	 \label{eq: red4_2}
		 \begin{split}
			P_{\cG_{\mathrm{do}(XZ)}}(Y|XZW)  &= P_{\cG_{\mathrm{do}(XZ)}}(Y|XZ\tilde{s}_W)  \\
			P_{\cG_{\mathrm{do}(XZ)}}(Y|ZW)  &= P_{\cG_{\mathrm{do}(XZ)}}(Y|Z\tilde{s}_W) 
		 \end{split}
	 \end{align}
	 Plugging these into the first expression of \cref{eq: red4}, we find $P_{\cG_{\mathrm{do}(XZ)}}(Y|XZ\tilde{s}_W)\neq  P_{\cG_{\mathrm{do}(Z)}}(Y|Z\tilde{s}_W)$, which is equivalent to $X \vDash Y | \mathrm{do}(Z),\tilde{s}_W$.
   \end{proof}

 We note that the operational motivations for Red$_1$, Red$_2$ and Red$_3$ given in \cref{sec:affects-red} also extend to the conditional case (as can be easily checked). The operational motivation for Red$_4$ introduced in this appendix is discussed below. 
 
	\textbf{Operational motivation for the definition of Red$_4$} The rationale in this case mirrors that of Red$_3$, the two conditions of \cref{def: red4} demand equality between the left hand side and right hand side respectively of the expressions for the original affects relation $X \vDash Y | \mathrm{do}(Z),W$ and the reduced affects relation $X \vDash Y | \mathrm{do}(Z),\tilde{s}_W$. This ensures that the two affects relations are equivalent and carry the same information.

	The following examples illustrate two instances of Irred$_4$ affects relations, showing that the two conditions of \cref{def: red4} can be independently violated.
	
	\begin{example}{(Violating the first Red$_4$ condition)}
		\label{eq: red_4_1}
	To illustrate that the first condition of \cref{def: red4} can be independently violated, consider the causal structure $X\longrsquigarrow Y$, $X\longrsquigarrow W$ and $E\longrsquigarrow W$ with any distribution over $X$, $P(E)$ being non-uniform and non-deterministic, $W=X\oplus W$, $Y=X$. Since $X$ is exogenous and $P_{\mathcal{G}}(Y|XW)\neq P_{\mathcal{G}}(Y|W)$ whenever $W$ is non-deterministic, we have $X\vDash Y|W$. However $Y$ is correlated with $W$ since $E$ is non-uniform and we have $(Y\not\upmodels W)_{\cG_{\mathrm{do}(X)}}$ which violates the second condition of \cref{def: red4} for the affects relation $X\vDash Y|W$ with $s_W=W$ (which is the only possibility for $s_W$ in this case). However, $(Y\upmodels W|X)_{\cG_{\mathrm{do}(X)}}$ holds due to the d-separation property and we therefore satisfy the second condition for $X\vDash Y|W$ with $s_W=W$. 
	\end{example}
	
	\begin{example}{(Violating the second Red$_4$ condition)}
		\label{eq: red_4_2}
	To illustrate that the second condition of \cref{def: red4} can be independently violated, consider again a jamming causal model with $X\longrsquigarrow Y$, $\Lambda\longrsquigarrow Y$ and $\Lambda\longrsquigarrow W$, with $\Lambda$ uniformly distributed, $Y=\Lambda\oplus X$ and $W=\Lambda$. We clearly have $X\vDash Y|W$ which will be our original affects relation, where $s_W=W$ is the only possible non-empty subset of the fourth argument. Since $X$ is exogenous, we have $(Y\upmodels W)_{\cG_{\mathrm{do}(X)}}$ and $(Y\not\upmodels W|X)_{\cG_{\mathrm{do}(X)}}$ i.e., for the affects relation $X\vDash Y|W$, the first condition of \cref{def: red4} is satisfied but the second condition is violated. 
	\end{example}

Analogously, we can generalize the definition of clustering in different arguments to conditional affects relations. The following three definitions are a direct generalization of \cref{def: clus1}, \cref{def: clus2} and \cref{def: clus3}, with an inclusion of $W$.
\begin{definition}[Clustering in the first argument]
  An affects relation $X\vDash Y|\mathrm{do}(Z),W$ is called clustered in the first argument (denoted Clus$_1$) if $|X|\geq 2$ and there exists no $s_X\subsetneq X$ such that   $s_X\vDash Y|\mathrm{do}(Z),W$. 
\end{definition}

\begin{definition}[Clustering in the second argument]
  An affects relation $X\vDash Y|\mathrm{do}(Z),W$ is called clustered in the second argument (denoted Clus$_2$) if $|Y|\geq 2$ and there exists no $s_Y\subsetneq Y$ such that   $X\vDash s_Y|\mathrm{do}(Z),W$. 
\end{definition}

\begin{definition}[Clustering in the third argument]
  An affects relation $X\vDash Y|\mathrm{do}(Z),W$ is called clustered in the third argument (denoted Clus$_3$) if $|Z|\geq 1$ and there exists no $s_Z\subsetneq Z$ such that   $X\vDash Y|\mathrm{do}(s_Z),W$. 
\end{definition}

\begin{definition}[Clustering in the fourth argument]
	\label{def: clus4}
	An affects relation $X\vDash Y|\mathrm{do}(Z),W$ is called clustered in the fourth argument (denoted Clus$_4$) if there exists no non-empty subset $s_W\subsetneq W$ such that   $X\vDash Y|\mathrm{do}(Z),s_W$. 
\end{definition}

Notice that due to the requirement that $s_W$ must be a non-empty and strict subset of $W$, the minimum cardinality of $W$ needed for a conditional affects relation to be Clus$_4$ is 2. The reason for this is because the fourth argument captures post-selection without intervention, thus allowing $W$ to act as a collider to $X$ and $Y$. This means that in any faithful causal model where $W$ is a collider between $X$ and $Y$ and $X$ is not a cause of $Y$, we would generically have  $X\vDash Y|W$ even though $X \not\vDash Y$ (in contrast, we have seen in \cref{lemma: clus_finetune} that the analogous situation $X\vDash Y|\mathrm{do}(Z)$ and $X\not\vDash Y$ for the third argument is only possible in fine-tuned models). With the above definitions, we can now generalize \cref{theorem: clus_irred} (relating clustering and irreducibility) to conditional affects relations and all four arguments, as shown below.

\begin{restatable}{theorem}{clusIrredFour}[Clustering implies irreducibility]
\label{theorem: clus_irred4}
    For any affects relation  $X\vDash Y|\mathrm{do}(Z),W$, Clus$_i$ $\Rightarrow$ Irred$_i$ for all $i\in\{1,2,3,4\}$.
\end{restatable}
\begin{proof}
    For $i\in\{1,2,3\}$, the proofs are entirely analogous to that of \cref{theorem: clus_irred} and will not be repeated. We carry out the proof for $i=4$ below.

	Again, start by assuming that $X\vDash Y|\mathrm{do}(Z),W$ holds and satisfies Clus$_4$ and Red$_4$. From \cref{def: clus4} and \cref{def: red4} this implies the following conditions (see also \cref{eq: red4_2}).
	\begin{align}
		\begin{split}
	P_{\cG_{\mathrm{do}(XZ)}}(Y|XZ\tilde{s}_W)&=  P_{\cG_{\mathrm{do}(Z)}}(Y|Z\tilde{s}_W), \quad \forall \tilde{s}_W\subsetneq W,\\
	P_{\cG_{\mathrm{do}(XZ)}}(Y|XZW)  &= P_{\cG_{\mathrm{do}(XZ)}}(Y|XZ\tilde{s}_W)\quad \exists \tilde{s}_W\subseteq W, \\
		P_{\cG_{\mathrm{do}(XZ)}}(Y|ZW)  &= P_{\cG_{\mathrm{do}(XZ)}}(Y|Z\tilde{s}_W)\quad \exists \tilde{s}_W\subseteq W.
		\end{split}
	\end{align}
	These imply that $ P_{\cG_{\mathrm{do}(XZ)}}(Y|XZW)= P_{\cG_{\mathrm{do}(Z)}}(Y|ZW)$, which is equivalent to $X\not\vDash Y|\mathrm{do}(Z),W$. This contradicts our initial assumption that $X\vDash Y|\mathrm{do}(Z),W$ and establishes the result.
\end{proof}

\begin{example}{(Irred$_4$ $\not\Rightarrow$ Clus$_4$)}
    Consider the jamming-type causal model over binary variables, with $X\longrsquigarrow Y$, $\Lambda \longrsquigarrow Y$ and $\Lambda \longrsquigarrow W$, with $P(\Lambda)$ being neither uniformly distributed nor a deterministic distribution, $W=\Lambda$ and $Y=X\oplus \Lambda$. Here we have $X\vDash Y$ since $\Lambda$ is non-uniform as well as $X\vDash Y|W$ since $\Lambda$ is non-deterministic. Therefore $X\vDash Y|W$ does not satisfy the Clus$_4$ property. However since $Y\not\upmodels W|X$ (and $X$ is exogenous to the same holds under interventions on $X$), $X\vDash Y|W$ is not reducible in the fourth argument and is an Irred$_4$ affects relation. 
\end{example}

\Cref{lemma: clus_finetune} in the main text proves that clustering in the first, second and third argument of an unconditional affects relation $X\vDash Y|\mathrm{do}(Z)$ implies fine-tuning of the underlying causal model. Here we extend the result to conditional affects relations and clustering in the fourth argument, introduced in this section. Due to the restriction to the subset $X(W)$ of $X$ in \cref{lemma: dsep_aff_cond}, which is only relevant when considering non-trivial conditional affects relations with $W\neq \emptyset$, our result for the conditional case is not fully general and applies to cases where $X(W)=X$, or equivalently where $X$ is not a cause of $W$.\footnote{$X$ is a cause of $W$ is equivalent to saying that there exists an element $e_X$ that has a directed path to some element $e_W$ of $W$, therefore $X$ is not a cause of $W$ is equivalent to $X\cap \mathrm{anc}(W)=\emptyset$, which is equivalent to $X(W)=X$ for $X(W)$ as defined in \cref{lemma: dsep_aff_cond}.} We conjecture that the theorem also holds more generally, without this restriction and leave a proof of this to future work, as it is not directly relevant for the main results of this paper.

\begin{restatable}{theorem}{clusFineTuningCond}[{Clustering implies fine-tuning (conditional version)}]
    \label{lemma: clus_finetune_cond}
    Any affects relation $X\vDash Y|\mathrm{do}(Z),W$ that satisfies Clus$_i$ for $i\in \{1,2,3,4\}$ and generated in a causal graph $\cG$ where $X$ is not a cause of $W$, necessarily arises from a fine-tuned causal model on $\cG$.  
\end{restatable}
\begin{proof}
    The proofs for $i\in\{1,2,3\}$ are entirely analogous as those of \cref{lemma: clus_finetune} whenever we are given that $X$ is not a cause of $W$. This is because the main step of the proof is to infer from $X\vDash Y \given \mathrm{do}(Z)$ that $(X\not\perp^d Y|Z)_{\cG_{\mathrm{do}(XZ)}}$, and the analogous inference holds for conditional affects relations when $X$ is not a cause of $W$, due to \cref{lemma: dsep_aff_cond} i.e., $X\vDash Y \given \mathrm{do}(Z),W$ implies $(X\not\perp^d Y|ZW)_{\cG_{\mathrm{do}(XZ)}}$ whenever $X$ is not a cause of $W$. We therefore do not repeat the proof.

    The proof for $i=4$ is also analogous, but we repeat it for completeness. Suppose that $X\vDash Y \given \mathrm{do}(Z),W$ holds and is a Clus$_4$ affects relation. This implies in particular that for all $e_W\in W$, $X\not\vDash Y|\mathrm{do}(Z),e_W$. Writing this out, we have
    \begin{align}
    \begin{split}
    \label{eq: clus4ftproof1}
        P_{\cG_{\mathrm{do}(XZ)}}(Y|XZW)&\neq P_{\cG_{\mathrm{do}(Z)}}(Y|ZW)\\
        P_{\cG_{\mathrm{do}(XZ)}}(Y|XZe_W)&= P_{\cG_{\mathrm{do}(Z)}}(Y|Ze_W), \qquad \forall e_W\in W.
    \end{split}
    \end{align}
    Further, using \cref{lemma: aff_corr_cond}, the non-affects relations $X\not\vDash Y|\mathrm{do}(Z),e_W$ imply the conditional independences $(X \upmodels Y|Z,e_W)_{\cG_{\mathrm{do}(XZ)}}$. Now consider the corresponding d-separation for each $e_W\in W$, we can either have $(X \perp^d Y|Z,e_W)_{\cG_{\mathrm{do}(XZ)}}$ or $(X \not\perp^d Y|Z,e_W)_{\cG_{\mathrm{do}(XZ)}}$. In the latter case we could have a d-connection with a corresponding conditional independence, which would make the model fine-tuned. So we consider the case where 
    $(X \perp^d Y|Z,e_W)_{\cG_{\mathrm{do}(XZ)}}$ holds for all $e_W\in W$ (for all other cases the above argument establishes the result about fine-tuning). However, by assumption that $X$ is not a cause of $W$, we have $X(W)=W$ and we know from \cref{lemma: dsep_aff_cond} that $(X \perp^d Y|Z,e_W)_{\cG_{\mathrm{do}(XZ)}}$ implies $X\not\vDash Y|\{\mathrm{do}(Z),W\}$, which contradicts our initial assumption and established the result.
\end{proof}

Finally, we provide an additional result relating Clus$_2$ and Clus$_4$ (clustering in observational arguments) which is analogous to \cref{thm:clusThreeOne} which links Clus$_1$ and Clus$_3$ (clustering in the interventional arguments). 
However, due to the minimum cardinality for affects relations to be Clus$_4$, it only holds in a more restricted case.

\begin{restatable}{corollary}{clusTwoFour}
    For any unconditional affects relation $X\vDash Y|\mathrm{do}(Z)$ which is Clus$_2$ and has $\abs{Y} \geq 3$,
    $X\vDash s_Y|\mathrm{do}(Z),\tilde{s}_Y$ is Clus$_4$ for each partition $s_Y \tilde{s}_Y = Y$ with $\abs{\tilde{s}_Y} \geq 2$.
\end{restatable}
\begin{proof}
	Any $X\vDash Y|\mathrm{do}(Z)$ which is Clus$_2$, by \cref{thm:decondition}.3 is equivalent to $X\vDash s_Y|\mathrm{do}(Z),\tilde{s}_Y$ for each partition $s_Y \tilde{s}_Y = Y$.

	If for the case of $\tilde{s}_Y \geq 2$ -- which is only possible for $Y \geq 2$ -- this affects relation were not Clus$_4$, there would exist $X \affects s_Y|\mathrm{do}(Z),\tilde{s}'_Y$ with $\tilde{s}'_Y \subsetneq \tilde{s}_Y$, implying $X \affects s_Y\tilde{s}'_Y|\mathrm{do}(Z)$ by \cref{thm:decondition}.1. This is in contradiction to the original affects relation being Clus$_2$ and proves the claim.
\end{proof}

\medskip
The above result is only established for unconditional affects relations, and the full generalization of these statements to conditional relations remains open. 

\subsection{Compatibility for conditional affects relations}

We provide the generalization of \cref{def:compat} to conditional affects relations, as given in \cite{VVC}, yet simplified with regard to the embedding.

\begin{corollary}
    \label{def:compat-cond}
    Let $\SC$ be a set of ORVs from a set of RVs $S$ and a poset $\TT$ with an embedding $\mathcal{E}$. Then a set of affects relations $\mathscr{A}$ is said to be \emph{compatible} with $\mathcal{E}$ (or satisfies \textbf{compat}) if the following condition holds:
    \begin{itemize}
        \item Let $X, Y \subset S$ be disjoint non-empty sets of RVs, $Z, W \subset S$ further disjoint sets of RVs, potentially empty.
        If ($X \affects Y \given \doo(Z), W$) $\in \mathscr{A}$ and is Irred$_1$, then
        $\Fut_s (\YY\WW\ZZ) \subseteq \Fut_s (\XX) \, .$
    \end{itemize}
\end{corollary}

\noindent This generalization follows from the unconditional definition of compatibility by applying \cref{thm:decondition}.

\section{Relaxing compatibility to atomic affects relations}
\label{sec:no-irreducible}

As we have seen, the causal structure is a graph whose nodes are associated with individual random variables or physical systems (of some theory) while the signalling structure, which is captured by affects relations in our framework, is defined on sets of random variables. In other frameworks, which consider notions of signalling defined for quantum systems in a circuit, it has been shown that in unitary quantum circuits, signalling relations between individual in/output systems are sufficient to fully characterize all the signalling relations.\footnote{This is not the case in non-unitary circuits. There, we can have signalling between sets of in/output systems without signalling between individual elements of the set (analogous to clustered affects relations we have encountered in \cref{sec:affects-clus}).} This property of unitary circuits is referred to as atomicity \cite{Ormrod2023}. In d-separation causal models, whenever there is an inequivalence between causation and signalling, this indicates fine-tuning \cite{VVC}. The assumption of a unitary quantum circuit is closely related to the assumption of a faithful (not fine-tuned) causal model in the quantum case, as discussed in \cite{arxiv.1906.10726}.

This motivates us to consider a similar property for affects relations (which capture signalling), and call an affects relation atomic if it originates from a single RV: $X \affects Y \given \doo(Z)$ with $\abs{X} = 1$.\footnote{In contrast to the atomicity property on unitaries, this notion still allows the second argument of the affects relation to refer to sets of RVs. As this is not relevant for our main results, we leave an exploration of atomicity in different arguments to future work, focusing only on the first argument here.}
This is particularly appealing as all such affects relations are Irred$_1$ by construction.

Then we ask, when does imposing compatibility only for atomic affects relations with a given space-time imply general compatibility for all irreducible affects relations. We find that this is not true in general but it is true when we either restrict to faithful causal models or to conical space-times, revealing yet another correspondence between these distinct causality concepts. 

\begin{definition}[\textbf{compat-atomic}]
	\label{def:compat-atomic}
    Let $\SC$ be a set of ORVs from a set of RVs $S$ and a poset $\TT$ with an embedding $\mathcal{E}$. Then a set of unconditional affects relations $\mathscr{A}$ satisfies \textbf{compat-atomic} with respect to  $\mathcal{E}$, if it satisfies \textbf{compat} restricted to such affects relations whose first argument $X \subset S$ satisfies $\abs{X} = 1$.
\end{definition}

While \hyperref[def:compat]{\textbf{compat}} implies \textbf{compat-atomic}, both notions are not equivalent for general causal models and space-times as illustrated by the following example.

\begin{figure}[t]
    \centering
    \begin{tikzpicture}[dot/.style={circle,inner sep=2pt,fill,name=#1}]
        \pgfsetblendmode{multiply}
        \fill[fill=blue!20] (2.2,2.2) -- (3.5,3.5) -- (0.7,3.5) -- (2.2,2.2);
        \fill[fill=red!20] (3.7,1.8) -- (5.4,3.5) -- (2.0,3.5) -- (3.7,1.8);
        
        \node [dot=Z,label=$\ZZ$] at (1,1) {};
        \node [dot=X,label=$\XX$] at (2.2,2.2) {};
        \node [dot=Y,label=$\YY$] at (3.7,1.8) {};
        \node [dot=a,style={inner sep=1pt},label=$a$] at (2.75,2.75) {};

        \node (left) at (0.5,0.5) {};
        \node (right) at (5.0,0.5) {};
        \draw (left) -- ++(3,3);
        \draw (right) -- ++(-3,3);

        \begin{scope}[>={Stealth[black]},
                      every node/.style={fill=white,circle},
                      every edge/.style=vvarrow]
            \path [->] (X) edge (Y);
            \path [->] (Z) edge (Y);
        \end{scope}
    \end{tikzpicture}
	\caption{
        Combined representation of the causal structure of \cref{ex:atomic} and its space-time embedding into 1+1-Minkowski space-time satisfying \hyperref[def:compat-atomic]{\textbf{compat-atomic}} and accordingly, $\Fut_s(\YY\ZZ) = \Fut_s (\YY\XX)$. Even though $XZ \affects Y$, we have $\Fut(\YY) \not\subseteq \Fut_s (\XX\ZZ)$ and hence violate \hyperref[def:compat]{\textbf{compat}}.}
	\label{fig:atomic}
\end{figure}
\begin{example}{(\textbf{compat-atomic} $\not\Rightarrow$ \textbf{compat})}
    \label{ex:atomic}
	Consider a one-time pad of binary variables, where $X\longrsquigarrow Y$ and $Z\longrsquigarrow Y$, with a model satisfying $Y = X \oplus Z$. Then for any distribution on the parentless nodes $X$ and $Z$, the affects relations with a singleton in their first argument ($X\affects Y \given \doo(Z)$ and $Z\affects Y \given \doo(X)$) imply $\Fut_s(\YY\ZZ) \subseteq \Fut (\XX)$ as well as $\Fut_s(\YY\XX) \subseteq \Fut (\ZZ)$ as compatibility conditions.
	Therefore \textbf{compat-atomic} yields $\Fut_s (\YY\ZZ) = \Fut_s (\YY\XX)$.
	For 1+1-Minkowski space-time, this is satisfied in an embedding where $\ZZ \prec \XX \prec a \succ \YY$ for $a \in \TT$.
	By contrast, due to $XZ \affects Y$ being Irred$_1$, \textbf{compat} additionally demands $\Fut (\YY) \subseteq \Fut_s (\XX\ZZ)$, which is not satisfied in this embedding. Note however that the stronger condition imposed by \textbf{compat} is necessary for having no superluminal signalling: In the given embedding, an agent with access to $Y$ can learn about the parity of $X$ and $Z$ outside the future light cone of any of these variables.
\end{example}

We now show that analogous to the main result of \cref{thm:irr-compat} / \cref{thm:irr-compat-embedding}, under almost identical restrictions on either the space-time or the causal model / affects relations, the desired implication does hold, hinting at a deeper correspondence between conical space-times and causal models without clustered affects relations.
For this, recall that according to \cref{def:reducible}, an affects relation $X \affects Y \given \doo(Z)$ to be Irred$_1$ is actually just a short-hand for a whole family of affects relations to hold:
\begin{equation}
	s_X \affects Y \given \doo(Z \tilde{s}_X) \quad \forall s_X \subseteq X \, ,
\end{equation}
where $\tilde{s}_X := X \setminus s_X$.
Therefore, we have a rather analog situation to \cref{thm:compat-atomic-to-strong} concerning irreducibility in the third argument, and can follow the same idea for the proof.

\begin{lemma}
	\label{thm:compat-atomic-to-irr}
	Let $\mathscr{A}$ be a set of unconditional affects relations.
	Then for any non-degenerate conical embedding $\mathcal{E}$ (cf.\ \cref{def:conical-embedding}) into space-time $\TT$,
	\begin{equation}
		\hyperref[def:compat-atomic]{\textbf{compat-atomic}} \implies \hyperref[def:compat]{\textbf{compat}} .
	\end{equation}
\end{lemma}
\begin{proof}
	By definition, $X \affects Y \given \doo(Z)$ being Irred$_1$ implies
	$e_X \affects Y \given \doo(Z X \setminus e_X)$.
	\hyperref[def:compat-atomic]{\textbf{compat-atomic}} then yields
	\begin{equation}
        \label{eq: atomic_compat1}
		\Fut_s (\YY\ZZ) \cap \Fut_s (\XX \setminus e_\XX) \subseteq \Fut (e_\XX) \quad \forall e_\XX \in \XX \, .
	\end{equation}
	With \cref{thm:compat-irr-to-strong-second} for $\AA = \YY\ZZ$ and $\BB = \XX$ it follows that
	\begin{equation}
		\Fut_s (\YY\ZZ\XX) = \Fut_s (\YY\ZZ\XX \setminus e_\XX) \quad \forall e_\XX \in \XX \, .
	\end{equation}
	We proceed analogous to \cref{thm:compat-atomic-to-strong}, deriving that the embedding is either degenerate or
	\begin{equation}
		\Fut_s (\ZZ\YY) \subseteq \Fut_s (\XX) \, .
	\end{equation}
	This is precisely \hyperref[def:compat]{\textbf{compat}}, concluding the proof.
\end{proof}

\begin{lemma}
	Let $\mathscr{A}$ be a set of unconditional affects relations which does not contain any affects relations that are Clus$_1$.
	Then for any non-degenerate embedding $\mathcal{E}$ of $\mathscr{A}$ into \underline{any} space-time $\TT$,
	\begin{equation}
		\hyperref[def:compat-atomic]{\textbf{compat-atomic}} \implies \hyperref[def:compat]{\textbf{compat}} .
	\end{equation}
    In particular, this holds for any faithful causal model.
\end{lemma}
\begin{proof}
    By definition, $X \affects Y \given \doo(Z)$ being Irred$_1$ implies
	$e_X \affects Y \given \doo(Z X \setminus e_X)$.
 
    By \cref{thm:not-clus1}, the absence of Clus$_1$ relations in $\mathscr{A}$ implies that there exists $e^1_X \in X$ such that $e^1_X \affects Y \given \doo(Z)$, and therefore, by \hyperref[def:compat-atomic]{\textbf{compat-atomic}}, $\Fut_s (\YY\ZZ) \subseteq \Fut (e^1_\XX)$, which again implies $\Fut_s (\YY\ZZ e^1_\XX) = \Fut_s (\YY\ZZ)$.
    However, due to Irred$_1$, we also have $X \setminus e^1_X \affects Y \given \doo(Z e^1_X)$.
    Invoking again the absence of Clus$_1$ relations and applying \cref{thm:not-clus1}, we know that there exists $e^2_X \in X$ such that $e^2_X \affects Y \given \doo(Z e^1_X)$, yielding $\Fut_s (\YY\ZZ e^1_\XX) \subseteq \Fut(e^2_\XX)$.
	Together by the earlier condition on $\Fut(e^1_\XX)$, this yields $\Fut_s (\YY\ZZ) \subseteq \Fut(e^2_\XX)$.

    As Irred$_1$ for the original affects relation implies the same for $X \setminus e^1_X \affects Y \given \doo(Z e^1_X)$, we can now repeat this procedure recursively, yielding $\Fut_s (\YY\ZZ) \subseteq \Fut(e_\XX) \ \forall e_\XX \in \XX$. Taken together, this is equivalent to $\Fut_s (\YY\ZZ) \subseteq \Fut_s (\XX)$, yielding the claim.
\end{proof}

\medskip

Using the absence of affects relations with Clus$_2$ or Clus$_3$ as a condition, similar results can be obtained for atomic compatibility in the second and third argument.

Analogous to \cref{thm:irr-compat} and \cref{thm:irr-compat-clus}, this demonstrates another useful implication which does not hold in general space-times and causal models, but does hold when either restricting to conical space-times or restricting to causal models without a type of clustering.
A difference however is that \cref{thm:irr-compat-clus} uses the (stronger) absence of Clus$_3$ while the above result uses only the absence of Clus$_1$, indeed the former was related to Irred$_3$ affects relations while the latter is about Irred$_1$ affects relations.
Moreover, while our main results imply an ability to reduce compatibility statements for higher-order affects relations to equivalent statements for 0$^\text{th}$-order relations, the results of this section enable a reduction of compatibility statements for non-atomic affects relations to equivalent statements for atomic relations. There is scope for exploring further relations and implications of these results, as well as other possible correspondences between conical space-times and causal models without clustering of a certain type. We leave this for future work. 

\begin{remark}
	Due to \cref{thm:third-embedding}, any affects relation $X \affects Y \given \doo(Z)$ irreducible in its third argument implies that for \underline{any} embedding satisfying \hyperref[def:compat-atomic]{\textbf{compat-atomic}}, we have
	$\Fut_s (\YY\XX) \cap \Fut_s (\ZZ \setminus e_\ZZ) \subseteq \Fut (e_\ZZ)$.
	Conjoining this with \cref{eq: atomic_compat1} (implied by \hyperref[def:compat-atomic]{\textbf{compat-atomic}}), we obtain
	\begin{equation}
		\Fut_s (\YY) \cap \Fut_s (\XX\ZZ \setminus e_{\XX\ZZ}) \subseteq \Fut (e_{\XX\ZZ}) \quad \forall e_{\XX\ZZ} \in \XX\ZZ \, .
	\end{equation}
    Since \hyperref[def:compat]{\textbf{compat}} implies \hyperref[def:compat-atomic]{\textbf{compat-atomic}}, the same follows from imposing \hyperref[def:compat]{\textbf{compat}} for an Irred$_3$ affects relation.
    Hence, even without imposing the restrictions of \cref{thm:irr-compat} or \cref{thm:irr-compat-clus}, namely the absence of clustering or conical space-times, we obtain this weaker type of interchangeability of $\XX$ and $\ZZ$ directly from \hyperref[def:compat-atomic]{\textbf{compat-atomic}} (or \hyperref[def:compat]{\textbf{compat}}). However, this is not strong enough to yield the correspondence of \cref{eq: correspondence_main} between causal inference statements and compatibility statements.
\end{remark}

\section{Proofs of results}
\label{sec:proofs}

\subsection{Proofs for Section 4} \label{sec:affects-proofs}

\affectsCorr*
\begin{proof}
    \label{proof:affectsCorr}
    Suppose by contradiction that $X\not\vDash Y|\mathrm{do}(Z)$ (equivalently $P_{\cG_{\mathrm{do}(XZ)}}(Y|XZ)=  P_{\cG_{\mathrm{do}(Z)}}(Y|Z)$) and $(X\not\upmodels Y|Z)_{\cG_{\mathrm{do}}(XZ)}$. The latter is equivalent to saying that there exist distinct values $z$ of $Z$ and distinct values $x$ and $x'$ of $X$ such that $P_{\cG_{\mathrm{do}(XZ)}}(Y|X=x,Z=z)\neq P_{\cG_{\mathrm{do}(XZ)}}(Y|X=x',Z=z)$. However, this implies that $P_{\cG_{\mathrm{do}(XZ)}}(Y|XZ)$ does have a non-trivial dependence on $X$ and contradicts $P_{\cG_{\mathrm{do}(XZ)}}(Y|XZ)=  P_{\cG_{\mathrm{do}(Z)}}(Y|Z)$ (which would imply that it is equal to an $X$-independent quantity). This proves the result. 
\end{proof}

\clusIrreducible*
\begin{proof}
    \label{proof:clusIrreducible}
	\textbf{Clus$_1$ $\Rightarrow$ Irred$_1$}
	Suppose that $X\vDash Y|\mathrm{do}(Z)$ holds and satisfies Clus$_1$ but not Irred$_1$ i.e., it is Red$_1$. Then by \cref{def: red1} and \cref{def: clus1} of these properties, we have the following requirements: $\tilde{s}_X\not\vDash Y|\mathrm{do}(Z)$ for all $\tilde{s}_X\subsetneq X$ and $\exists \tilde{s}_X\subsetneq X$, with $\tilde{s}_X\cup s_X=X$ such that $s_X\not\vDash Y|\mathrm{do}(\tilde{s}_XZ)$. Writing these out explicitly, we have:
	\begin{align}
		\begin{split}
	  P_{\cG_{\mathrm{do}(\tilde{s}_XZ)}}(Y|\tilde{s}_XZ)&=  P_{\cG_{\mathrm{do}(Z)}}(Y|Z), \quad \forall \tilde{s}_X\subsetneq X,\\
	  P_{\cG_{\mathrm{do}(XZ)}}(Y|XZ)&=  P_{\cG_{\mathrm{do}(\tilde{s}_XZ)}}(Y|\tilde{s}_XZ), \quad \exists \tilde{s}_X\subsetneq X,
		\end{split}
	\end{align}
	
	These imply $ P_{\cG_{\mathrm{do}(XZ)}}(Y|XZ)=P_{\cG_{\mathrm{do}(Z)}}(Y|Z)$ which is equivalent to $X\not\vDash Y|\mathrm{do}(Z)$. This contradicts our initial assumption that $X\vDash Y|\mathrm{do}(Z)$ holds and therefore proves that Clus$_1$ implies Irred$_1$ for any affects relation $X\vDash Y|\mathrm{do}(Z)$.

	\textbf{Clus$_2$ $\Rightarrow$ Irred$_2$}
	We assume that $X\vDash Y|\mathrm{do}(Z)$ holds, it is Clus$_2$ but not Irred$_2$ and derive a contradiction, which will establish that Clus$_2$ $\Rightarrow$ Irred$_2$ for any affects relation $X\vDash Y|\mathrm{do}(Z)$.
	Clus$_2$ implies the following
	\begin{equation}
	\label{eq: clus2}
		P_{\cG_{\mathrm{do}(XZ)}}(s_Y|XZ)=P_{\cG_{\mathrm{do}(Z)}}(s_Y|Z) \quad \forall s_Y\subsetneq Y.
	\end{equation}

	On the other hand, Red$_2$ implies that there exists $s_Y\subsetneq Y$ such that $X \naffects s_Y \given \doo(Z), \tilde{s}_Y$. However, by \cref{thm:decondition}.3, both of these together imply $X \affects \tilde{s}_Y \given \doo(Z)$, which contradicts our assumption that $X \affects Y \given \doo(Z)$ is Clus$_2$.
	
	  \textbf{Clus$_3$ $\Rightarrow$ Irred$_3$} 
	As before, we assume that $X\vDash Y|\mathrm{do}(Z)$ holds and satisfies Clus$_3$ and Red$_3$. From \cref{def: clus3} and \cref{def: red3} this implies the following conditions.
	\begin{align}
		\begin{split}
	  P_{\cG_{\mathrm{do}(X\tilde{s}_Z)}}(Y|X\tilde{s}_Z)&=  P_{\cG_{\mathrm{do}(\tilde{s}_Z)}}(Y|\tilde{s}_Z), \quad \forall \tilde{s}_Z\subsetneq Z,\\
	 P_{\cG_{\mathrm{do}(XZ)}}(Y|XZ)&=  P_{\cG_{\mathrm{do}(X\tilde{s}_Z)}}(Y|X\tilde{s}_Z), \quad \exists \tilde{s}_Z\subsetneq Z,\\
	 P_{\cG_{\mathrm{do}(Z)}}(Y|Z)&=  P_{\cG_{\mathrm{do}(\tilde{s}_Z)}}(Y|\tilde{s}_Z), \quad \exists \tilde{s}_Z\subsetneq Z.
		\end{split}
	\end{align}
	
	Combining these, we obtain $ P_{\cG_{\mathrm{do}(XZ)}}(Y|XZ)= P_{\cG_{\mathrm{do}(Z)}}(Y|Z)$ which is equivalent to $X\not\vDash Y|\mathrm{do}(Z)$, which contradicts our initial assumption and therefore proves the claim. 
\end{proof}

\clusFineTuning*
\begin{proof}
    \label{proof:clusFinetuning}
	{\bf Clus$_1$ $\Rightarrow$ fine-tuning} Suppose that $X\vDash Y|\mathrm{do}(Z)$ is a Clus$_1$ affects relation, which means that $|X|\geq 2$ and $s_X\not\vDash Y|\mathrm{do}(Z)$ for all $s_X\subsetneq X$. 
	From the proof of Lemma~IV.3 of \cite{VVC}, it follows that $X\vDash Y|\mathrm{do}(Z)$ implies that there exists $e_X\in X$ with a directed path from $e_X$ to $Y$ in $\cG_{\mathrm{do}}(XZ)$, which in turn implies $(e_X\not\perp^d Y|Z)_{\cG_{\mathrm{do}}(XZ)}$.
    Further, this implies the same d-connection in the graph where we intervene only on $e_X$, i.e., $(e_X\not\perp^d Y|Z)_{\cG_{\mathrm{do}}(e_XZ)}$.
	This is because the graph $\cG_{\mathrm{do}}(XZ)$ will have the same nodes and less edges than $\cG_{\mathrm{do}}(e_XZ)$ (since interventions cut off incoming edges), and d-connection cannot be lost by adding edges.
	However Clus$_1$ implies in particular that $e_X\not\vDash Y|\mathrm{do}(Z)$ .
	This in turn implies that $(e_X\upmodels Y|Z)_{\cG_{\mathrm{do}}(s_XZ)}$ (\cref{lemma: aff_corr}), which along with $(e_X\not\perp^d Y|Z)_{\cG_{\mathrm{do}}(e_XZ)}$ which we have previously established, implies that any underlying causal model giving rise to these affects relations must be fine-tuned.

 {\bf Clus$_2$ $\Rightarrow$ fine-tuning} Suppose that $X\vDash Y|\mathrm{do}(Z)$ is a Clus$_2$ affects relation, which means that $|Y|\geq 2$ and $X\not\vDash s_Y|\mathrm{do}(Z)$ for all $s_Y\subsetneq Y$. As before, $X\vDash Y|\mathrm{do}(Z)$ implies that $(X\not\perp^d Y|Z)_{\cG_{\mathrm{do}}(XZ)}$. By the definition of d-separation, this tells us that there must exist $e_Y\in Y$ such that $(X\not\perp^d e_Y|Z)_{\cG_{\mathrm{do}}(XZ)}$. However Clus$_2$ implies that for all $e_Y\in Y$, $X\not\vDash e_Y|\mathrm{do}(Z)$ which gives $(X\upmodels e_Y|Z)_{\cG_{\mathrm{do}}(XZ)}$ for all $e_Y\in Y$ (\cref{lemma: aff_corr}). Then it is clear that there is at least one d-connection $(X\not\perp^d e_Y|Z)_{\cG_{\mathrm{do}}(XZ)}$ not matched by a corresponding conditional independence, thus making the model fine-tuned.

 {\bf Clus$_3$ $\Rightarrow$ fine-tuning} Suppose that $X\vDash Y|\mathrm{do}(Z)$ is a Clus$_3$ affects relation, which means that $|Z|\geq 1$ and $X\not\vDash Y|\mathrm{do}(s_Z)$ for all $s_Z\subsetneq Z$. As in the proof for the previous case, we start with the fact that $X\vDash Y|\mathrm{do}(Z)$ implies $(X\not\perp^d Y|Z)_{\cG_{\mathrm{do}}(XZ)}$. Since all nodes of the conditioning set $Z$ are parentless in this d-connection, it follows that $(X\not\perp^d Y|s_Z)_{\cG_{\mathrm{do}}(XZ)}$ holds for subsets $s_Z\subsetneq Z$. This is because the only case where removing elements from the conditioning set could remove a d-connection is when those elements act as a collider (or descendent of a collider) on a sole unblocked path between $X$ and $Y$, which is not possible if those elements are parentless. Now, observe that the only distinction between the graphs $\cG_{\mathrm{do}}(Xs_Z)$ for $s_Z\subsetneq Z$ and $\cG_{\mathrm{do}}(XZ)$ is that the former has a strictly larger set of edges (the incoming edges to nodes in $Z\backslash s_Z$ are present in the former and not in the latter). Since d-connection cannot be lost by adding edges, $(X\not\perp^d Y|s_Z)_{\cG_{\mathrm{do}}(XZ)}$ implies $(X\not\perp^d Y|s_Z)_{\cG_{\mathrm{do}}(Xs_Z)}$. We then use the fact that Clus$_3$ implies $X\not\vDash Y|\mathrm{do}(s_Z)$ for all $s_Z\subsetneq Z$, and consequently that $(X\upmodels Y|s_Z)_{\cG_{\mathrm{do}}(XZ)}$ for all $s_Z\subsetneq Z$ (\cref{lemma: aff_corr}), which when taken together with the d-connection $(X\not\perp^d Y|s_Z)_{\cG_{\mathrm{do}}(Xs_Z)}$ indicates that the causal model must be fine-tuned.
\end{proof}

\subsection{Proofs for Section 5} \label{sec:location-symmetry}

\conicalMinkowski*
\begin{proof}
\label{proof:conicalMinkowski}
In $d$+1-dim.\ Minkowski space-time $\TT$, the boundary of the light cone of a space-time point $(\vec{x}_0,t_0)\in \TT$ is given as follows, where $\vec{x}_0=(x^1_0,x^2_0,...,x^d_0)$ are the $d$-dimensional spatial co-ordinates of the point (in some chosen co-ordinate system).
\begin{equation}
    c^2 (t-t_0)^2 = \sum_{1 < i \leq d} (x^i-x^i_{0})^2 \, .
\end{equation}

Observe that for each time slice (fixed value of $t$), the above equation corresponds to a $d$-dimensional sphere with radius $c^2 (t-t_0)^2$. Whenever $d\geq 2$, for a fixed $t$, given any finite portion of this spherical boundary, we can determine the radius of the sphere and therefore its center, which is associated with the spatial co-ordinate $\vec{x}_0$. If we are considering the future light cone, this is non-empty only for $t>t_0$. Given a non-empty time slice of the future light cone of a space-time point $(\vec{x}_0,t_0)\in \TT$, we can consider light rays emanating from each point on the boundary of the sphere and extending towards the negative $t$ direction, these light rays would necessarily intersect at a unique time $t_0$ which allows us to fully determine the original point $(\vec{x}_0,t_0)$, starting from any portion of the future light cone boundary on a given non-empty time slice. 

Now let $L$ be a finite subset of points in $d$+1-dim.\ Minkowski space-time and consider the joint future $f(L)=\bigcap_{x \in L} \bar{J}^+ (x)$ of all the points in $L$. The boundary of this joint future is compiled from pieces of the light cone boundaries of points in $L$. In particular, the boundary of $f(L)$ associated with a non-empty time slice is compiled from spherical pieces which are portions of light cone boundaries of points in $L$, associated with that time slice (for $d=2$, these are circular arcs). We can show that each $x\in \spann(L)$ contributes a unique, finite portion to the boundary of the joint future region $f(L)$ at every non-empty time slice. 
For this, recall that by definition, if $x\in \spann(L)$, then $x\in L'\subseteq L$ such that $f(L')=f(L)$ and there is no $L''\subsetneq L'$ satisfying $f(L'')=f(L)$. If we assume that $x$ doesn't contribute a finite piece to the boundary of $f(L)$, then this yields a contradiction as it would mean that for $L'':=L'\setminus x$, we have $f(L'')=f(L)$. To show uniqueness, suppose $x\in \spann(L)$ contributes a finite portion to the boundary of $f(L)$ at a time slice $t$. No other $x'\in \TT$ can contribute the same portion of the boundary of $f(L)$, due to the argument presented in the first paragraph of this proof: each such portion allows us to uniquely reconstruct the point $x$ as it is also a portion of the future light cone of $x$.

Thus, the distinct spherical portions that comprise the boundary of $f(L)$ can each be mapped back uniquely to a distinct spanning element of $L$. Since $L$ is arbitrary here, this implies that $f(L)$ is invertible on all sets $L$ with $L=\spann(L)$, which proves the conicality of Minkowski space-times with 2 or more spatial dimensions.
\end{proof}

\conicalORV*

We prove a generalization of this lemma from conical space-times to conical embeddings in arbitrary space-times, which are defined as follows:
\begin{definition}[Conical Embedding]
    \label{def:conical-embedding}
    An embedding of a set $\SC$ of ORVs in a space-time $\cT$ is called \emph{conical} if for any two finite subsets $L_i, L_j \subseteq O(\SC) = \{O(X) | X \in S \} \subseteq \TT$,
    \begin{equation}
        f(L_i) = f(L_j)
        \quad \implies \quad 
        \spann(L_i) = \spann(L_j) 
    \end{equation}
    holds.
\end{definition}

Intuitively, with this definition, we consider the condition of conicality only for the futures of all points $O(\SC)$ in the image of the embedding and the intersections of their futures.
Hence, any embedding into a conical poset $\TT$ is trivially conical. However, the conicality of $\TT$ is not necessary to obtain a conical embedding:
Actually, the vast majority of embeddings into 1+1-dim.\ Minkowski space-time is conical as well, as for instance shown in \cref{tab:summary}.\footnote{
    Specifically, such embeddings can be non-conical only if one ORV is embedded into the future light cone surface of another,
    constituting a kind of fine-tuning of the embedding.  
    For instance, consider the example shown in \cref{fig:non-degenerate}, with $\YY$ embedded on the future light cone surface of $\ZZ_1$ and $\Fut_s (\YY \ZZ_2) = \Fut_s (\ZZ_1 \ZZ_2)$.
    As a possible sufficient criterion for being located \enquote{on the light cone surface} we suggest, in purely order-theoretical terms:
    For $a, b \in \TT$, $b$ is embedded in the future light cone surface of $a$ iff $M = \bar{J}^+ (a) \cap \bar{J}^- (b) = \{ x \in \TT | a \preceq x \preceq b \}$ is totally ordered by $\prec$, i.e.\ $\not\exists \, b, c \in M : \ b \unord c$.
    Here, $M$ is a degenerate example of a causal diamond \cite{Hounnonkpe2019, Witten2020}.
    For Minkowski space-time, this criterion is also necessary, while singularities or self-intersecting light cones could lead to violations.
}

We finally state the generalization of \cref{def:conical-orv}:

\begin{lemma}
    \label{def:conical-embedding-orv}
    Let $\SC$ be a set of ORVs from a poset $\TT$ with a \emph{conical} embedding $\mathcal{E}$. 
    Then, for any $\XX \subseteq \SC$, the knowledge of $\Fut_s (\XX)$ implies the locations $O (\XX_i)$ for all its spanning elements $\XX_i \in \spann (\XX)$.
\end{lemma}
\begin{proof}
    Recalling that $O(\XX) = \{ O(\XX_i) \mid \XX_i \in \XX \} \subset O (\SC)$ and $f(L) := \bigcap_{x \in L} \bar{J}^+ (x)$, then by definition, for any $\XX \subset \SC$,
    \begin{equation}
        \Fut_s (\XX) \equiv \bigcap_{\XX_i \in \XX} \bar{J}^+ (O(\XX_i)) = f(O(\XX)) = f(O({\spann(\XX)})) \, .
    \end{equation}
    As the embedding is conical, we have that for $O(\SC) \subseteq \TT$, the conicality condition is satisfied and therefore, by \cref{def:conical}, $f(\spann(O(\XX)))$ is injective (at least when restricted to $O(\XX)$).
    Hence,
    it remains to show that $\spann(O(\XX)) \overset{!}{=} O(\spann(\XX))$.
    To this purpose, we differentiate between the case that the embedding is degenerate or not.
    If the embedding is non-degenerate, we have a bijection between $\XX$ and $L_\XX$, which immediately yields $\spann(O(\XX)) = O(\spann(\XX))$.
    Otherwise, there are multiple ORVs $u_\XX \in \XX$ sharing the same location.
    In this case, it is easy to see that either all of them or none of them belong to $\spann(\XX)$.\footnote{If at least one $u_\XX \ni e_\XX \in \spann(\XX)$, for each allowed choice of $s_\XX$ (as according to the definition), $\abs{s_\XX \cap u_\XX} = 1$. However, due to considering the union of such sets we then regain the entirety of $u_\XX \subset \bigcup s_\XX$.}
    Therefore, we can group every set of ORVs embedded degenerately at the same location, into a single ORV embedded at the same location without changing the span. Thus we have reduced a degenerate problem to a non-degenerate one such that $\spann(O(\XX)) = O(\spann(\XX))$ holds in the former if and only if it holds in the latter. Since we have established that this always holds in the non-degenerate case, this proves the claim.
\end{proof}

\medskip

As outlined in the main text, our primary results rely on showing that conicality of a poset is equivalent to another poset property called location symmetry. We proceed by motivating this property. Suppose we have a family $\{S_i\}_{i\in I}$ where each element is a set of locations $S_i \subset \TT$ all of which share the same span, let us denote this by $K$. Then we can write $S_i=K L_i$ (where $KL_i$ is short for the union of the sets), making clear that the locations in $K$ must be shared among the sets and here the different $L_i$ may also share locations. 
Plugging $S_i$ for $L_i$ and $K$ for $L_j$ into \cref{eq:conical}, which holds in conical space-times, we obtain that the the joint futures of these sets must be identical in such space-times,
\begin{equation}
   f(K) = f(S_i) \subseteq f(S_i \setminus K)=f(L_i) \, , \forall i\in I
\end{equation}
where $f(L) := \bigcap_{x \in L} \bar{J}^+ (x)$.
Here, the subset relation follows from the fact that removing sets from an intersection always yields a (potentially improper) superset of the original intersection. As this holds for all $i\in I$, we can replace $f(L_i)$ with $\bigcap_{i\in I} f(L_i) $.
This yields that in conical space-times where $\spann(KL_i)=K$ for all $i\in I$ (or equivalently $\bigcap_{i\in I}\spann(KL_i)=K$),
\begin{equation}
    f(KL_i) = f(KL_k) \quad \forall i, k \in I \quad \implies \quad
    f(K) \subseteq \bigcap_{i \in I} f(L_i) 
\end{equation}

More generally, we can consider sets $S_i=KL_i$ where $K$ is not the common span of the sets. Then $\bigcap_{i\in I}\spann(KL_i)\neq K$. This implies that there must be further shared locations in $\spann(K L_i) \setminus K$ for each $i$, that are commonly shared between all the $L_i$'s, and using conicality, it can be shown that for all $i$, $\spann(K L_i) \setminus K$ must be identical in this case (as we do in the proofs below). Based on this motivation, we have the following definition of location symmetry which captures the two alternative conditions discussed above.

\begin{definition}[Location Symmetry]
	\label{def:location-symmetry-points}
	Let $\TT$ be a poset and $I$ be an index set. Let $K, L_i \subset \TT$ discrete with $i \in I$.
	Then $\TT$ satisfies \emph{location symmetry} if
	\begin{equation}
		\begin{split}
			&
			f(K) \cap f(L_i) = f(K) \cap f(L_k) \quad \forall i, k \in I \\
			\ \implies \
			&
			f(K) \subseteq \bigcap_{i \in I} f(L_i) \quad \text{or} \quad
			\exists ! M \subset \TT : \ \spann(K L_i) \setminus K = M \neq \emptyset \quad \forall i \in I
		\end{split}
	\end{equation}
	where
	\begin{equation*}
        f(L) := \bigcap_{x \in L} \bar{J}^+ (x)
    \end{equation*}

	If $K$ or $L_i$ is empty, we understand their intersection as being equal to the full poset $\TT$.
\end{definition}

Alternatively, we can also phrase location symmetry in terms of a set of ORVs $\SC$ on this poset.

\begin{definition}[Location Symmetry for ORVs]
    \label{def:location-symmetry-orvs}
    Let $\TT$ be a poset and $I$ be an index set. 
	Let $\XX, \YY_i \subset \SC$ with $i \in I$, not necessarily disjoint or non-empty. Then $\TT$ satisfies \emph{location symmetry}
 \begin{equation}
 \label{eq: LS_ORVs}
		\begin{split}
			& \Fut_s (\XX \YY_i) = \Fut_s (\XX \YY_k) \quad \forall i, k \in I \\
			\ \implies \
			& 
            \Fut_s (\XX) \subseteq \Fut_s \left( \textstyle{\bigcup_{i \in I}} \YY_i \right) \quad \mathrm{or} \quad
			\exists ! M \subset \TT : \ 
			O (\spann(\XX \YY_k)) \setminus O (\XX) = M \neq \emptyset \quad \forall k \in I
		\end{split}
	\end{equation}
\end{definition}

Especially in the respective second alternative, these definitions take care to play well with the respective sets / their locations potentially not being disjoint, as this will be of importance for considering degenerate embeddings of ORVs.

\begin{lemma}
	\label{thm:conicality-location-symmetry}
	$\TT$ is a conical poset if and only if it satisfies location symmetry (for points in the poset).
    Further, $\mathcal{E}$ is a conical embedding if and only if it satisfies location symmetry (for all its ORVs).
\end{lemma}
\begin{proof}
	For this proof, we will refer to ORVs directly, as that is more compact and matches the notation we will use for the proofs going forward.
    Also, the proof for both statements is identical.
    All notation matches \cref{def:location-symmetry-orvs}.

	{
	\hin
Here we assume location symmetry and prove conicality. We only require the former for the special case 
where $\XX = \emptyset$ while $\YY_1$ and $\YY_2$ are any disjoint sets of ORVs with $\Fut_s (\YY_1)=\Fut_s (\YY_2)$. Then, applying \cref{def:supp-future} to empty sets of ORVs, we recover $\Fut_s (\XX) = \TT$. Moreover, the pre-condition $\Fut_s (\XX \YY_1)=\Fut_s (\XX \YY_2)$ for applying location symmetry is satisfied as it is equivalent to $\Fut_s (\YY_1)=\Fut_s (\YY_2)$. Therefore, location symmetry implies that
	\begin{equation}
		\TT = \Fut_s (\YY_1 \YY_2)
        \quad \text{or} \quad
		O (\spann(\XX \YY_1)) \setminus O (\XX) = O (\spann(\XX \YY_2)) \setminus O (\XX) \, .
	\end{equation}
Here, the second  alternative immediately reduces to $O(\spann(\YY_1)) = O(\spann(\YY_2))$, due to $\XX$ being empty.
The first is an equality as $\Fut_s (\YY_1 \YY_2)\subseteq \TT$ trivially holds. 
Moreover, notice that the first case $\TT = \Fut_s (\YY_1 \YY_2)$ is equivalent to $\TT = \Fut_s (\YY_1) = \Fut_s (\YY_2)$, and this can only be satisfied if both sets of ORVs are embedded into the minimal element of the poset (which may not exist). If the minimal element exists, we therefore have $O (\YY_1) = O (\YY_2)$, which in particular implies $O(\spann(\YY_1)) = O(\spann(\YY_2))$
In other words, for both alternatives of location symmetry we get that for any disjoint sets of ORVs $\YY_1$ and $\YY_2$, $\Fut(\YY_1) = \Fut(\YY_2) \implies O(\spann(\YY_1)) = O(\spann(\YY_2))$, reproducing the definition of conicality (\cref{eq:conical}).} 
	
	\rueck 
	Let $\XX, \YY_i \subset \SC$ with $i \in I$ such that
	\begin{equation}
 \label{eq: LS_proof1}
		\Fut_s (\XX \YY_i) = \Fut_s (\XX \YY_k) \quad \forall i, k \in I \, .
	\end{equation}
 We now show that in any conical space-time, this implies one of the two conditions in the second line of \cref{eq: LS_ORVs}, thus proving location symmetry.

Applying the property \cref{eq:conical} of conical space-times to \cref{eq: LS_proof1}, we see that for any choice of $i \in I$, we can deduce the value of $O(\AA) $ for all $\AA \in \spann(\XX \YY_i)$ for ORVs embedded in any conical poset.
	This implies for the set of respective locations to be identical.
	\begin{equation}
  \label{eq: LS_proof2}
		O (\spann(\XX \YY_i)) = O (\spann(\XX \YY_k)) \quad  \forall i, k \in I \, .
	\end{equation}

We now consider two cases and show that these yield respectively the two alternatives for location symmetry, thus concluding the proof.
\begin{enumerate}
\item {\bf Case 1 ($\exists i: O(\XX) = O(\spann(\XX \YY_i))$) :} In this case, we can deduce that the same holds for all $\YY_k$ with $k \in I$, due to \cref{eq: LS_proof1}. Hence,
	\begin{align}
		\Fut_s (\XX) &= \Fut_s (\XX \YY_i) \\
		\implies \Fut_s (\XX) &\subseteq \Fut_s (\YY_i) \quad \forall i \in I \\
		\iff \Fut_s (\XX) &\subseteq \bigcap_{i \in I} \Fut_s (\YY_i)
		= \Fut_s \left( \bigcup_{i \in I} \YY_i \right) \, ,
	\end{align}
	yielding the first alternative of \cref{eq: LS_ORVs}.
\item {\bf Case 2 ($\forall i: O(\XX) \neq O(\spann(\XX \YY_i))$) :} 
This implies that $O(\spann(\XX \YY_i)) \setminus O(\XX) \neq \emptyset$.
We have already established in  \cref{eq: LS_proof2} that in conical space-times where \cref{eq: LS_proof1} is satisfied, $O(\spann(\XX \YY_k))$ must be identical for all $i \in I$, this implies the same holds for $O(\spann(\XX \YY_i)) \setminus O(\XX)$, and yields the second alternative of \cref{eq: LS_ORVs}.

\end{enumerate}
\end{proof}

We conclude by providing a weaker, but significantly shorter version of the second alternative that will be useful for clarity in some later proofs. 

\begin{lemma}
    \label{eq:location-symmetry-degen}
    Let $M \subseteq \TT$ non-empty. Let $I, \XX, \YY_k$ as in \cref{def:location-symmetry-orvs}. Then 
    \begin{equation}
        O (\spann(\XX \YY_k)) \setminus O (\spann(\XX)) = M \quad \forall k \in I
		\quad \implies \quad
		\bigcap_k O (\YY_k) \neq \emptyset \, .
    \end{equation}
\end{lemma}
\begin{proof}
    First, we see that $\spann(\XX\YY_i) \setminus \XX \subseteq \YY_i$.
    As this carries over to their respective locations, we obtain
    \begin{align}
    & \emptyset \neq M = O (\spann(\XX \YY_k)) \setminus O (\spann(\XX)) \quad \forall k \in I \\
    \iff \ & \emptyset \neq M = O (\spann(\XX \YY_k)) \cap O(\YY_k) \setminus O (\spann(\XX))\quad \forall k \in I \\
    \implies \ & \emptyset \neq M \subseteq O (\spann(\XX \YY_k)) \cap O(\YY_k) \quad \forall k \in I \\
    \implies \ & \emptyset \neq M \subseteq \bigcap_k O (\spann(\XX \YY_k)) \cap O(\YY_k) \\
    \implies \ & \emptyset \neq \bigcap_k O (\YY_k) \, ,
    \end{align}
    where the last implication holds as intersecting larger sets always yields a (potentially improper) superset.
\end{proof}

\subsection{Proofs for Section 6}
\label{sec:proofs6}

To prove \cref{thm:irr-compat}, we require some auxiliary lemmas, which we first prove before stating and proving the main theorem.

\begin{lemma}
	\label{thm:compat-irr-to-strong-second}
	Let $\AA, \BB \subset \SC$. Then
	\begin{equation}
		\Fut_s (\AA) \cap \Fut_s (\BB \setminus e_\BB)\subseteq \Fut (e_\BB) \ \forall e_\BB \in \BB
		\quad \implies \quad
		\Fut_s (\AA\BB) = \Fut_s (\AA\BB \setminus e_\BB) \ \forall e_\BB \in \BB \, .
	\end{equation}
\end{lemma}
\begin{proof}
	Let $e^i_\BB := e_\BB$. We can transform
	\begin{align}
		\Fut_s (\AA) \cap \Fut_s (\BB \setminus e^i_\BB) &\subseteq \Fut (e^i_\BB) \\
		\implies
		\Fut_s (\AA) \cap \Fut_s (\BB \setminus e^i_\BB) &\subseteq \Fut_s (\AA) \cap \Fut (e^i_\BB) &\forall e^i_\BB \in \BB \\
		\implies
		\Fut_s (\AA) \cap \Fut_s (\BB \setminus e_\BB^i) &\subseteq \Fut_s (\AA) \cap \Fut (e_\BB^i) \cap \Fut_s (\BB \setminus \{ e^i_\BB, e^j_\BB \}) &\forall e_\BB^i, e_\BB^j \in \BB \\
		\iff
		\Fut_s (\AA) \cap \Fut_s (\BB \setminus e_\BB^i) &\subseteq \Fut_s (\AA) \cap \Fut_s (\BB \setminus e_\BB^j) &\forall e_\BB^i, e_\BB^j \in \BB \, . \!\!\!
	\end{align}
	Combining these for all $e_\BB$, we get
	\begin{align}
		\Fut_s (\AA) \cap \Fut_s (\BB \setminus e_\BB^i) &= \Fut_s (\AA) \cap \Fut_s (\BB \setminus e_\BB^j) \quad \forall e_\BB^i, e_\BB^j \in \BB \\
		\implies \Fut_s (\AA\BB) &= \Fut_s (\AA\BB \setminus e_\BB) \quad \forall e_\BB \in \BB \, .
	\end{align}
	This is precisely the claim.
\end{proof}

\begin{lemma}
    \label{thm:third-embedding}
    Let $\mathscr{A}$ be a set of unconditional affects relations and $\mathscr{A}'\subseteq \mathscr{A}$ consist of all affects relations in $\mathscr{A}$ that are Irred$_3$.
	Then for any non-degenerate embedding $\mathcal{E}$ into an arbitrary space-time $\TT$ satisfying \hyperref[def:compat]{\textbf{compat}}\, we have
	\begin{equation}
		(X \affects Y \given \doo(Z)) \in \mathscr{A}'
		\quad \implies \quad
		\Fut_s (\YY\XX) \cap \Fut_s (\ZZ \setminus e_\ZZ) \subseteq \Fut (e_\ZZ) \ \forall e_\ZZ \in \ZZ \, .
	\end{equation}
\end{lemma}
\begin{proof}
	As $X \affects Y \given \doo(Z)$ satisfies Irred$_3$, we have for $\tilde{Z}:=Z\backslash e_Z$
	\begin{equation}
		\forall e_Z \in Z : \quad
		e_Z \affects Y \given \doo(\tilde{Z})
        \quad \text{or} \quad
		e_Z \affects Y \given \doo(\tilde{Z} X) \, ,
	\end{equation}
	either of which is Irred$_1$, if it holds.
	Therefore, by \hyperref[def:compat]{\textbf{compat}}, we get
	\begin{align}
		&
		\Fut_s (\YY) \cap \Fut_s (\tilde{\ZZ}) \subseteq \Fut (e_\ZZ) \quad \text{or} \quad
		\Fut_s (\YY) \cap \Fut_s (\XX\tilde{\ZZ}) \subseteq \Fut (e_\ZZ) \\
		\implies \ &
		\Fut_s (\YY\XX) \cap \Fut_s (\tilde{\ZZ}) \subseteq \Fut (e_\ZZ)
	\end{align}
	for any $e_Z$.
\end{proof} 

\medskip
With these three auxiliary lemmas in place, we can continue to prove the theorem.

\begin{lemma}
    \label{thm:compat-atomic-to-strong}
    Let $\mathscr{A}$ be a set of unconditional affects relations and $\mathscr{A}'\subseteq \mathscr{A}$ consist of all affects relations in $\mathscr{A}$ that are Irred$_3$.
	Then for any non-degenerate conical embedding $\mathcal{E}$ into a space-time $\TT$ satisfying \hyperref[def:compat]{\textbf{compat}}, we have
	\begin{equation}
		(X \affects Y \given \doo(Z)) \in \mathscr{A}'
		\quad \implies \quad
		\Fut_s (\YY) \cap \Fut_s (\XX) \subseteq \Fut_s (\ZZ) \, .
	\end{equation}
\end{lemma}
\begin{proof}
	By conjoining \cref{thm:third-embedding} with \cref{thm:compat-irr-to-strong-second} for $\AA = \YY\XX$ and $\BB = \ZZ$, for any affects relation $X \affects Y \given \doo(Z)$ irreducible in the third argument, we get
	\begin{equation}
		\Fut_s (\XX\YY\ZZ) = \Fut_s (\XX\YY\ZZ \setminus e_\ZZ) \quad \forall e_\ZZ \in \ZZ \, .
	\end{equation}
	Due to conicality of the embedding, \cref{thm:conicality-location-symmetry} implies that $\TT$ satisfies location symmetry for points in $O(\SC)$.
	By using this property on the chain of equalities (setting $\XX$ in \cref{eq: LS_ORVs} to $\XX\YY$, and $\YY_i$ to $\ZZ_i = \ZZ \setminus e_\ZZ^i$) as well as the transformation of \cref{eq:location-symmetry-degen}, we arrive at 
	\begin{equation}
		\Fut_s (\XX\YY) \subseteq \Fut_s (\ZZ)
        \quad \text{or} \quad
		\bigcap_i O (\ZZ_i) \neq \emptyset \, ,
	\end{equation}
	where $Z_i := Z \setminus e_Z^i$.
	
	For the second alternative, beware that the different $\ZZ_i$ are not disjoint.
	Nonetheless, $\bigcap_i \ZZ_i = \emptyset$.
	Therefore, multiple ORVs must share the same location.
	As this corresponds to a degenerate embedding, we arrive at the claim.
\end{proof}

\irrCompatConical*

We prove a generalization of this theorem from conical space-times to conical embeddings in arbitrary space-times (cf.\ \cref{def:conical-embedding}).

\begin{theorem}
	\label{thm:irr-compat-embedding}
    Let $\mathscr{A}$ be a set of unconditional affects relations and $\mathscr{A}'\subseteq \mathscr{A}$ consist of all affects relations in $\mathscr{A}$ that are both Irred$_1$ and Irred$_3$.
	Then for any non-degenerate conical embedding $\mathcal{E}$ (cf.\ \cref{def:conical-embedding} into a space-time $\TT$ satisfying \hyperref[def:compat]{\textbf{compat}}, we have
	\begin{equation}
		\label{eq:irr-compat-conical}
		(X \affects Y \given \doo(Z)) \in \mathscr{A}'
		\quad \implies \quad
	    \Fut_s (\YY) \subseteq \Fut_s (\XX) \cap \Fut_s (\ZZ) \, .
	\end{equation}
    In particular, this holds when the space-time itself is conical.
\end{theorem}

\begin{proof}
	Due to \cref{thm:compat-atomic-to-strong}, \hyperref[def:compat]{\textbf{compat}} implies the following for an Irred$_1$ and Irred$_3$ affects relation
	\begin{equation}
		(X \affects Y \given \doo(Z)) \in \mathscr{A}'
		\quad \implies \quad
		\Fut_s (\YY) \cap \Fut_s (\XX) \subseteq \Fut_s (\ZZ) \, .
	\end{equation}
	Conjoining this with \hyperref[def:compat]{\textbf{compat}} for this affects relation (which implies $\Fut_s (\YY) \cap \Fut_s (\ZZ) \subseteq \Fut_s (\XX)$) and using elementary set theory yields

	\begin{align}
		& 
		\Fut_s (\YY) \cap \Fut_s (\ZZ) \subseteq \Fut_s (\XX) \cap \Fut_s (\YY) \quad \text{and} \quad
		\Fut_s (\YY) \cap \Fut_s (\XX) \subseteq \Fut_s (\ZZ) \cap \Fut_s (\YY) \\
		\implies \ &
		\Fut_s (\YY) \cap \Fut_s (\XX) = \Fut_s (\YY) \cap \Fut_s (\ZZ)
	\end{align}
	Due to conicality, \cref{thm:conicality-location-symmetry} implies that $\TT$ satisfies location symmetry.
	Using location symmetry, we arrive at
	\begin{equation}
		\Fut_s (\YY) \subseteq \Fut_s (\XX\ZZ)
        \quad \text{or} \quad
        O (\XX) \cap O(\ZZ) \neq \emptyset \, .
	\end{equation}
	As the second option would imply the embedding to be degenerate, the claim, corresponding to the first option, follows.
\end{proof}

\irrCompatClusThree*
\begin{proof}
    \label{proof:compatClusThree}
    Every affects relation $X\vDash Y|\mathrm{do}(Z)$ in $\mathscr{A}'$ is Irred$_1$ and Irred$_3$.
    By \cref{thm:not-clus3}, due to the absence of affects relations with Clus$_3$, there exists $e^1_X \in X$ such that $e^1_X \affects Y$. Therefore, for any compatible embedding $\cE$, we have $\Fut_s (\YY) \subseteq \Fut (e^1_\XX)$.
    However, due to $X\vDash Y|\mathrm{do}(Z)$ being Irred$_1$, we also have $X \setminus e^1_X \affects Y \given \doo(Z e^1_X)$, which again by \cref{thm:not-clus3} implies that there is $e^2_X \in X \setminus e^1_X$ such that $e^2_X \affects Y$.
    Repeating this argument recursively, we find that any compatible embedding $\cE$ of such affects relations in a space-time will satisfy $\Fut_s (\cY) \subseteq \Fut (e_\XX) \ \forall e_\XX \in \XX$, or in short, $\Fut_s (\cY) \subseteq \Fut_s (\cX)$.
    
    Next we use the fact that $X\vDash Y|\mathrm{do}(Z)$ is Irred$_3$, which gives $s_Z \affects Y \given \doo(\tilde{s}_Z X)$ or $s_Z \affects Y \given \doo(\tilde{s}_Z)$ for all choices of $s_Z \subsetneq Z$.
    Applying \cref{thm:not-clus3}, both these alternatives yield $\exists e^1_Z \in Z$ s.t. $e^1_Z \affects Y$ and $\Fut_s (\YY) \subseteq \Fut (e^1_\ZZ)$.
    Moreover, $X\vDash Y|\mathrm{do}(Z)$ being Irred$_3$ implies in particular that $Z \setminus e^1_Z \affects Y \given \doo(e_Z^1 X)$ or $Z \setminus e^1_Z \affects Y \given \doo(e_Z^1)$, both of which imply by \cref{thm:not-clus3} that $\exists e^2_Z \in Z \setminus e^1_Z$ s.t. $e^2_Z \affects Y$.
    Repeating this procedure recursively, we obtain that for any compatible embedding $\cE$ of these affects relations, we have $\Fut_s (\YY) \subseteq \Fut (e_\ZZ) \ \forall e_\ZZ \in \ZZ$, or in short, $\Fut_s (\YY) \subseteq \Fut_s (\ZZ)$.

    Combining both results, we recover the claim.
\end{proof}

\end{document}